\documentclass[11pt,a4paper]{article}
\pdfoutput=1
\usepackage{graphicx}
\usepackage{soul}
\usepackage{natbib}
\usepackage{amsmath,amsfonts,amssymb,amsbsy,amsthm,mathtools}
\usepackage{bbm,bm}
\usepackage{color,xcolor}
\usepackage{subfig}
\usepackage{booktabs}
\usepackage{url}
\usepackage[a4paper]{geometry}
\usepackage{a4wide}
\usepackage{hyperref}
\usepackage{paralist}
\usepackage{authblk}
\allowdisplaybreaks
\numberwithin{equation}{section}

\newtheorem{thm}{Theorem}[section]
\newtheorem{cor}[thm]{Corollary}
\newtheorem{lem}[thm]{Lemma}

\DeclareMathOperator*{\argmin}{arg\,min}
\newcommand{\dm}{d}

\newcommand{\RR}{\mathbb{R}}

\newcommand{\vc}[1]{\bm{#1}}
\newcommand{\diff}{\mathrm{d}}
\renewcommand{\Pr}{\mathbb{P}}
\newcommand{\dto}{\xrightarrow{d}} 
\newcommand{\pto}{\overset{\Pr}{\rightarrow}}
\newcommand{\eps}{\varepsilon}
\newcommand{\Omegan}{\widehat{\Omega}_n}
\newcommand{\Hessian}{\mathcal{H}}
\newcommand{\norm}[1]{\lVert {#1} \rVert}

\providecommand{\keywords}[1]{\textbf{Keywords:} #1}

\title{An M-estimator of spatial tail dependence}
\author[1]{John H.J.~Einmahl}
\author[2]{Anna Kiriliouk}
\author[3]{Andrea Krajina}
\author[2]{Johan Segers}
\affil[1]{\footnotesize \emph{Department of Econometrics \& OR and CentER, Tilburg University,
P.O. Box 90153,} \authorcr
\emph{5000 LE Tilburg, The Netherlands. E-mail:} \url{j.h.j.einmahl@uvt.nl}}
\affil[2]{\footnotesize \emph{Institut de Statistique, Biostatistique et Sciences Actuarielles, Universit\'e catholique de Louvain,} \authorcr
\emph{Voie du Roman Pays 20, B-1348 Louvain-la-Neuve, Belgium. E-mail:} \authorcr
\url{anna.kiriliouk@uclouvain.be}, \url{johan.segers@uclouvain.be}}
\affil[3]{\footnotesize \emph{Institute for Mathematical Stochastics, University of G\"{o}ttingen, G\"{o}ttingen, Germany.} \authorcr
\emph{E-mail:} \url{andrea.krajina@mathematik.uni-goettingen.de}}

\begin{document}
\maketitle
 \vspace{-1cm}

\begin{abstract}
Tail dependence models for distributions attracted to a max-stable law are fitted using observations above a high threshold. To cope with spatial, high-dimensional data, a rank-based M-estimator is proposed relying on bivariate margins only. A data-driven weight matrix is used to minimize the asymptotic variance.
Empirical process arguments show that the estimator is consistent and asymptotically normal.
Its finite-sample performance is assessed in simulation experiments involving popular max-stable processes perturbed with additive noise.
An analysis of wind speed data from the Netherlands illustrates the method.
\end{abstract}

\keywords{%
Brown--Resnick process;
exceedances;
multivariate extremes;
ranks;
spatial statistics;
stable tail dependence function
}


\section{Introduction}
Max-stable random processes have become the standard for modelling extremes of environmental quantities, such as wind speed, precipitation, or snow depth. In such a context, data are modelled as realizations of spatial processes, observed at a finite number of locations. The statistical problem then consists of modelling the joint tail of a multivariate distribution. This problem can be divided into two separate issues: modelling the marginal distributions and modelling the dependence structure. A popular practice is to transform the marginals into an appropriate form and to fit a max-stable model to componentwise monthly or annual maxima using composite likelihood methods. This is done either in a frequentist setting \citep{padoan2010,davison2012} or in a Bayesian one \citep{reich2012,cooley2012}.
Alternatively, \citet{yuen2013} propose an M-estimator based on finite-dimensional cumulative distribution functions.
Popular parametric models for max-stable processes include the ones proposed by \cite{smith1990}, \cite{schlather2002} and \cite{kabluchko2009}, going back to \citet{brown1977}. Recent review articles on spatial extremes include \citet{cooley2012nr2}, \citet{davison2012} and \citet{ribatet2013}.

The above approaches consider block maxima, whereas more information can be extracted from the data by using all data vectors of  which at least one component is large. Although threshold-based methods are common in multivariate extreme-value theory, in spatial extremes they are only starting to be developed. A first example is \citet{dehaanpereira2006}, where several one- and two-dimensional models for spatial extremes are proposed. Another parametric model for spatial tail dependence is introduced in \citet{buishand2008}. The parameter estimator is shown to be asymptotically normal and the method is applied to daily rainfall data.
In \citet{huser2014}, a pairwise censored likelihood is used to analyse space-time extremes. The method is applied to an extension of Schlather's model.
Another study of space-time extremes can be found in \citet{davis2013}, where asymptotic normality of the pairwise likelihood estimators of the parameters of a Brown--Resnick process is proven for a jointly increasing number of spatial locations and time points.
In \citet{jeon2012}, bivariate threshold exceedances are modelled using a composite likelihood procedure. Asymptotic normality of the estimator is obtained by assuming second-order regular variation for the distribution function that is in the max-domain of attraction of an extreme-value distribution. A numerical study comparing two distinct approaches for composite likelihoods can be found in \citet{bacro2013}.
In \citet{wadsworth2014}, a censored Poisson process likelihood is considered in order to simplify the likelihood expressions in the Brown--Resnick process and in \citet{engelke2014}, the distribution of extremal increments of processes that are in the max-domain of attraction of the Brown--Resnick process is investigated. Finally, in \citet{bienvenue2014}, a censored likelihood procedure is used to fit high-dimensional extreme value models for which the tail dependence function has a particular representation.

The above methods all require estimation of the tails of the marginal distributions. This is not necessarily an easy task if the number of variables is large. Moreover, they are likelihood-based
and therefore cannot be used to fit, e.g., spectrally discrete max-stable models \citep{wang2011}.

The aim of this paper is to propose a new method for fitting multivariate tail dependence models to high-dimensional data arising for instance in spatial statistics. No likelihoods come into play as our approach relies on the stable tail dependence function, which is related to the upper tail of the underlying cumulative distribution function. The method is threshold-based in the sense that a data point is considered to be extreme if the rank of at least one component is sufficiently high. The only assumption is that the copula corresponding to the underlying distribution is attracted to a parametrically specified multivariate extreme-value distribution, see \eqref{eq:ell2} below.

By reducing the data to their ranks, the tails of the univariate marginal distributions need not be estimated. Indeed, the marginal distributions are not even required to be attracted to an extreme-value distribution. Another advantage of the rank-based approach is that the estimator is invariant under monotone transformations of the margins, notably for Box--Cox type of transformations.

Our starting point is \citet{einmahl2012}, where an M-estimator for a parametrically modelled tail dependence function in dimension $\dm$ is derived. However, that method crucially relies on $\dm$-dimensional integration, which becomes intractable in high dimensions. This is why we consider tail dependence functions of pairs of variables only. Our estimator is constructed as the minimizer of the distance between a vector of integrals of parametric pairwise tail dependence functions and the vector of their empirical counterparts. The asymptotic variance of the estimator can be minimized by replacing the Euclidean distance by a quadratic form based on a weight matrix estimated from the data. In the simulation studies we will compute estimates in dimensions up to 100.

We show that our estimator is consistent under minimal assumptions and asymptotically normal under an additional condition controlling the growth of the threshold. In our analysis, we take into account the variability stemming from the rank transformation, the randomness of the threshold, the random weight matrix and, in particular, the fact  that the max-stable model is only an approximation in the tail.

A point worth noticing is the generality of our methodology. Where many studies focus on a specific parametric (tail) model, ours is generic and makes weak assumptions only. Also, the field of application of extreme-value analysis in high dimensions is not restricted to environmental studies: see for example \citet{dematteo2013}, where a spectral clustering approach is introduced and applied to gas pressure data in the shipping industry.

For our approach a common, continuous distribution is required. The method does not apply to count data, for instance, and care must be taken with environmental variables that exhibit yearly seasonality or a trend, for instance due to global warming. In our case study, we study data on wind speeds in the Netherlands over a relatively short time period and limited to the summer months only.

The paper is organized as follows. Section~\ref{sec:background} presents the necessary background on multivariate extreme-value theory and extremes of stochastic processes. Section~\ref{sec:Mestimator} contains the definition of the pairwise M-estimator and the main theoretical results on consistency and asymptotic normality, as well as the practical aspects of the choice of the weight matrix. In Section~\ref{sec:spatial} the tail dependence functions of the anisotropic Brown--Resnick process and the Smith model are presented, as well as several simulation studies: two for a large number of locations, illustrating the computational feasibility of the estimator in high dimensions, and one for a smaller number of locations, presenting the benefits of the weight matrix. In addition, we compare the
performance of our estimator to the one proposed in \citet{engelke2014}. Section~\ref{sec:EKS12comp} contains comparisons between our pairwise M-estimator and the estimator proposed in \citet{einmahl2012}. Finally, in Section~\ref{sec:applic} we present an application to wind speed data from the Netherlands. Proofs are deferred to Appendix~A. The wind speed data and the programs that were used for the simulation studies are implemented in the \textsf{R} package \textsf{spatialTailDep} \citep{spatialTailDep}.

\section{Background}
\label{sec:background}

\subsection{Multivariate extreme-value theory}
\label{sec:background:mevt}

Let $\vc{X}_i = (X_{i1},\ldots,X_{i\dm})$, $i \in \{1,\ldots,n\}$, be independent random vectors in $\RR^{\dm}$
with common continuous distribution function $F$ and marginal distribution functions $F_1,\ldots,F_\dm$. Write $M_{nj} = \max_{i=1,\ldots,n} X_{ij}$ for $j=1,\ldots,d$. We say that $F$ is in the max-domain of attraction of an extreme-value distribution $G$ if there exist sequences $a_{nj} >0$ and $b_{nj} \in \RR$ for $j=1,\ldots,\dm$ such that
\begin{equation}\label{eq:ell1}
\lim_{n \rightarrow \infty} \Pr \bigg[ \frac{M_{n1} - b_{n1}}{a_{n1}} \leq x_1 , \ldots , \frac{M_{n\dm} - b_{n\dm}}{a_{n\dm}} \leq x_\dm \bigg] = G(\vc{x}), \qquad \vc{x} \in \RR^{\dm}.
\end{equation}
The margins, $G_1, \ldots, G_\dm$, of $G$ are univariate extreme-value distributions and the function $G$ is determined by
\begin{equation*}
G(\vc{x}) = \exp{\{- \ell(- \log G_1 (x_1), \ldots , - \log G_\dm (x_\dm))\}},
\end{equation*}
where $\ell: [0,\infty)^\dm \rightarrow [0,\infty)$ is called the stable tail dependence function. The distribution function of $(1/ \{1 - F_j (X_{1j})\} )_{j=1,\ldots,\dm}$ is in the max-domain of attraction of the extreme-value distribution $G_0 (\vc{z}) = \exp{\{-\ell(1/z_1,\ldots,1/z_\dm)\}}$, $\vc{z} \in (0, \infty)^\dm$, and we can retrieve the function $\ell$ via
\begin{equation}\label{eq:ell2}
  \ell(\vc{x})
  = \lim_{t \downarrow 0} t^{-1} \,
  \Pr[
    1 - F_1(X_{11}) \leq tx_1
    \text{ or } \ldots \text{ or }
    1 - F_\dm (X_{1\dm}) \leq t x_\dm
  ],
  \qquad \vc{x} \in [0, \infty)^\dm.
\end{equation}
Note that $G_0$ has unit Fr\'{e}chet margins, $G_{0,j} (z_j) = \exp{(-1/z_j)}$ for $z_j > 0$ and $j=1,\ldots,d$.

Relation \eqref{eq:ell1} is equivalent to relation \eqref{eq:ell2} \emph{and} convergence of the $d$ marginal distributions in \eqref{eq:ell1}. As a consequence, \eqref{eq:ell2} is substantially weaker than \eqref{eq:ell1}: it only concerns the distribution function of $(F_1(X_{11}), \ldots, F_d (X_{1d}))$, which does not depend on the marginal distributions $F_1,\ldots,F_d$.
That is, condition \eqref{eq:ell2} only concerns the dependence structure of $F$, described by, for example, the copula $C$ corresponding to $F$. Since $F(x_1,\dots,x_d) = C(F_1(x_1),\dots,F_d(x_d))$, we have
\begin{equation*}
\ell(\vc{x})
= \lim_{t \downarrow 0} t^{-1} \,
\left(1-C( 1 - tx_1,\dots,1 - t x_d )\right),
\qquad \vc{x} \in [0, \infty)^d.
\end{equation*}
The class of distribution functions satisfying \eqref{eq:ell2} is hence much larger than the class of functions satisfying the multivariate max-domain of attraction condition \eqref{eq:ell1}. 
It contains, for instance, all distributions of the form $F(x) = F_1 (x_1) \cdots F_d (x_d)$ with continuous margins, even if some of those margins do not belong to the max-domain of attraction of a univariate extreme value distribution.
Note also that if $F$ is already an extreme-value distribution, then it is attracted by itself.

From now on we will only assume relation \eqref{eq:ell2}, making no assumptions on the marginal distributions $F_1, \ldots, F_d$ except for continuity.
The function $\ell$ is convex, homogeneous of order one and satisfies $\ell(0,\ldots,0,x_j,0,\ldots,0) = x_j$ for $j=1,\ldots,\dm$. We assume that $\ell$ belongs to some parametric family $\{ \ell(\cdot\, ;\theta) : \theta \in \Theta \}$, with $\Theta \subset \mathbb{R}^p$.
There are numerous such parametric models and new families of models continue to be invented. We will see some examples of parametric stable tail dependence functions in Section~\ref{sec:spatial}. For more examples and background on multivariate extreme-value theory, see \citet{coles2001}, \citet{beirlant2004}, or \citet{dehaanferreira2006}.

\subsection{Extremes of stochastic processes}
\label{sec:background:spatial}
Max-stable processes arise in the study of component-wise maxima of random processes rather than of random vectors.
Let $\mathcal{S}$ be a compact subset of $\mathbb{R}^2$ and let $\mathbb{C}(\mathcal{S})$ denote the space of continuous, real-valued functions on $\mathcal{S}$, equipped with the supremum norm $\|f\|_{\infty} = \sup_{\vc{s} \in \mathcal{S}} | f(\vc{s})|$ for $f \in \mathbb{C}(\mathcal{S})$. The restriction to $\mathbb{R}^2$ is for convenience only. In the applications to spatial data that we have in mind, $\mathcal{S}$ will represent the region of interest.

Consider independent copies $\{X_i(\vc{s})\}_{\vc{s} \in \mathcal{S}}$ for $i \in \{1,\ldots,n\}$ of a process $\{X(\vc{s})\}_{\vc{s} \in \mathcal{S}}$ in $\mathbb{C}(\mathcal{S})$. Then $X$ is in the max-domain of attraction of the max-stable process $Z$ if there exist sequences of continuous functions $a_n (\vc{s}) >0$ and $b_n (\vc{s})$ such that
\begin{equation*}
\bigg\{ \frac{ \max_{i=1,\ldots,n} X_i (\vc{s}) - b_n(\vc{s})}{a_n(\vc{s})} \bigg\}_{\vc{s} \in \mathcal{S}} \overset{w}{\rightarrow} \{Z(\vc{s})\}_{\vc{s} \in \mathcal{S}}, \qquad \text{ as } n \rightarrow \infty,
\end{equation*}
where $\overset{w}{\rightarrow}$ denotes weak convergence in $\mathbb{C}(\mathcal{S})$;
see \citet{dehaanlin2001} for a full characterization of max-domain of attraction conditions for the case $\mathcal{S} = [0,1]$. A max-stable process $Z$ is called simple if its marginal distribution functions are all unit Fr\'{e}chet.

Although our interest lies in the underlying stochastic processes $X_i$, data are always obtained on a finite subset of $\mathcal{S}$ only, i.e., at fixed locations $\vc{s}_1,\ldots,\vc{s}_\dm$. As a consequence, statistical inference is based on a sample of $\dm$-dimensional random vectors. The finite-dimensional distributions of 
$Z$ are multivariate extreme-value distributions. This brings us back to the ordinary, multivariate setting.

\section{M-estimator}
\label{sec:Mestimator}

\subsection{Estimation}
As in Section~\ref{sec:background:mevt}, let $\vc{X}_1, \ldots, \vc{X}_n$ be an independent random sample from a $\dm$-variate distribution $F$ with continuous margins and with stable tail dependence function $\ell$, see equation~\eqref{eq:ell2}. Assume that $\ell$ belongs to a parametric family, $\{\ell(\cdot \, ; \theta) : \theta \in \Theta \}$, with $\Theta \subset \RR^p$. Let $\theta_0$ denote the true parameter vector, that is, $\theta_0$ is the unique point in $\Theta$ such that $\ell(x) = \ell(x; \theta_0)$ for all $\vc{x} \in [0, \infty)^\dm$. The goal is to estimate the parameter vector $\theta_0$.

To this end, we first define a nonparametric estimator of $\ell$. Let $R_{ij}^n$ denote the rank of $X_{ij}$ among $X_{1j}, \ldots , X_{nj}$ for $j=1,\ldots,\dm$. Replacing $F$ and $F_1, \ldots, F_d$ in~\eqref{eq:ell2} by their empirical counterparts and replacing $t$ by $k/n$ yields
\begin{equation}
\label{eq:npstdf}
  \widehat{\ell}_{n,k} (\vc{x})
  \coloneqq
  \frac{1}{k} \sum_{i=1}^n \mathbbm{1}
  \left\{
    R_{i1}^n > n + \frac{1}{2} - kx_1
    \text{ or } \ldots \text{ or }
    R_{id}^n > n + \frac{1}{2} - kx_\dm
  \right\}.
\end{equation}
For the estimator to be consistent, we need $k = k_n \in \{1, \ldots, n\}$ to depend on $n$ in such a way that $k \rightarrow \infty$ and $k/n \rightarrow 0$ as $n \rightarrow \infty$. The estimator was originally defined in the bivariate case in \citet{huang1992} and \citet{dreeshuang1998}.

Let $\ell = \ell(\cdot \, ; \theta_0)$, and let $g=(g_1, \ldots, g_q)^T : [0,1]^d \rightarrow \mathbb{R}^q$ with $q \geq p$ denote a column vector of integrable functions. In \citet{einmahl2012}, an M-estimator of $\theta_0$ is defined by
\begin{equation}\label{eq:theta}
\widehat{\theta}'_n \coloneqq \argmin_{\theta \in \Theta} \sum_{m=1}^q \left( \int_{[0,1]^\dm} g_m (\vc{x}) \left\{ \widehat{\ell}_{n,k} (\vc{x}) - \ell(\vc{x} ; \theta) \right\} \, \diff \vc{x} \right)^2.
\end{equation}
Under suitable conditions, the estimator $\widehat{\theta}'_n$ is consistent and asymptotically normal. The use of ranks via the nonparametric estimator in~\eqref{eq:npstdf} permits to avoid having to fit a model to the (tails of the) marginal distributions. In fact, the only assumption on $F$, the existence of the stable tail dependence function~$\ell$ in \eqref{eq:ell2}, is even weaker than the assumption that $F$ belongs to the maximal domain of attraction of a max-stable distribution.

However, the approach is ill-adapted to the spatial setting, where data are gathered from dozens of locations. In high dimensions, the computation of $\widehat{\theta}'_n$ becomes infeasible due to the presence of $\dm$-dimensional integrals in the objective function in~\eqref{eq:theta}.

Akin to composite likelihood methods, we opt for a pairwise approach, minimizing over quadratic forms of vectors of two-dimensional integrals. Let $q$ represent the number of pairs of locations that we wish to take into account, so that $p \leq q \leq d(d-1)/2$.
Let $\pi$ be the function from
$\{1, \dots, q\}$ to $\{1, \dots, d\}^2$ that describes these
pairs, that is, for $m \in \{1,\ldots,q\}$, we have $\pi(m)=(\pi_1(m),\pi_2(m))=(u,v)$ with $1 \leq u < v \leq d$.
In the spatial setting (cf.~Section~\ref{sec:background:spatial}), the indices $u$ and $v$ correspond to locations $\vc{s}_u$ and $\vc{s}_v$ respectively.

The bivariate margins of the stable tail dependence function $\ell(\cdot\,;\theta)$ and the nonparametric estimator in \eqref{eq:npstdf} are given by
\begin{align*}
\ell_{\pi(m)} (x_{\pi_1(m)},x_{\pi_2(m)} ; \theta) & = \ell_{uv} (x_u,x_v ; \theta) \coloneqq \ell(0,\ldots,0,x_u,0,\ldots,0,x_v,0,\ldots,0 ;  \theta), \\
\widehat{\ell}_{n,k,\pi(m)} (x_{\pi_1(m)},x_{\pi_2(m)})& = \widehat{\ell}_{n,k,uv} (x_u,x_v ) \coloneqq
\widehat{\ell}_{n,k}(0,\ldots,0,x_u,0,\ldots,0,x_v,0,\ldots,0),
\end{align*}
respectively.
Consider the random $q \times 1$ column vector
\begin{equation*}
L_{n,k} (\theta) \coloneqq \left( \int_{[0,1]^2} \left\{ \widehat{\ell}_{n, k, \pi(m)} (x_{\pi_1(m)},x_{\pi_2(m)}) - \ell_{\pi(m)} (x_{\pi_1(m)},x_{\pi_2(m)} ; \theta) \right\} \, \diff x_{\pi_1(m)} \diff x_{\pi_2(m)} \right)_{m=1}^q.
\end{equation*}
Let $\Omegan  \in \mathbb{R}^{q \times q}$ be a symmetric, positive definite, possibly random  matrix. Define
\begin{equation*}
  f_{n,k,\Omegan} (\theta)
  \coloneqq L_{n,k} (\theta)^T \, \Omegan \, L_{n,k} (\theta),
  \qquad \theta \in \Theta.
\end{equation*}
The pairwise M-estimator of $\theta_0$ is defined as
\begin{equation}\label{eq:thetafinal}
\widehat{\theta}_n \coloneqq \argmin_{\theta \in \Theta} f_{n,k,\Omegan} (\theta) = \argmin_{\theta \in \Theta} \left\{ L_{n,k} (\theta)^T \, \Omegan \, L_{n,k} (\theta) \right\}.
\end{equation}
The simplest choice for $\Omegan$ is just the $q \times q$ identity matrix $I_q$, yielding
\begin{equation}
\label{eq:fnkIq}
  f_{n,k,I_q}(\theta)
  = \sum_{(u,v)}
  \left(
    \int_{[0,1]^2}
      \left\{
	\widehat{\ell}_{n, k, uv} (x_{u},x_{v}) - \ell_{uv} (x_{u},x_{v} ; \theta)
      \right\} \,
    \diff x_u \, \diff x_v
  \right)^2.
\end{equation}
Note the similarity of this objective function with the one for the original M-estimator in equation~\eqref{eq:theta}. The role of the matrix $\Omegan$ is to be able to assign data-driven weights to quantify the size of the vector of discrepancies $L_{n,k} (\theta)$ via a generalized Euclidian norm. As we will see in Section~\ref{asymptotic}, a judicious choice of this matrix will allow to minimize the asymptotic variance.

\subsection{Asymptotic results and choice of the weight matrix}
\label{asymptotic}

We show consistency and asymptotic normality of the rank-based pairwise M-estimator. Moreover, we provide a data-driven choice for $\Omegan$ which minimizes the asymptotic covariance matrix of the limiting normal distribution. Results for the construction of confidence regions and hypothesis tests are presented as well.

A quantity related to the stable tail dependence function $\ell$ is the exponent measure $\Lambda$, which is a measure on $[0,\infty]^d \setminus \{ (\infty,\ldots,\infty)\}$ determined by
\begin{equation*}
\Lambda(\{\vc{w} \in [0,\infty]^\dm : w_1 \leq x_1 \text{ or } \ldots \text{ or } w_\dm \leq x_\dm \}) = \ell(\vc{x}), \qquad \vc{x} \in [0,\infty)^d.
\end{equation*}
Let $W_{\Lambda}$ be a mean-zero Gaussian process, indexed by the Borel sets of $[0,\infty]^d \setminus \{(\infty,\ldots,\infty)\}$ and with covariance function
\begin{equation*}
\mathbb{E}[W_{\Lambda} (A_1) \, W_{\Lambda} (A_2)] = \Lambda(A_1 \cap A_2),
\end{equation*}
where $A_1$, $A_2$ are Borel sets in $[0,\infty]^d \setminus \{(\infty,\ldots,\infty)\}$. For $\vc{x} \in [0, \infty)^d$, define
\begin{align*}
W_{\ell} (\vc{x}) & = W_{\Lambda} (\{\vc{w} \in [0,\infty]^d \setminus \{(\infty,\ldots,\infty)\} : w_1 \leq x_1 \text{ or } \ldots \text{ or } w_d \leq x_d \}), \\
W_{\ell,j} (x_j) & = W_{\ell} (0,\ldots,0,x_j,0,\ldots,0), \qquad j=1,\ldots,d.
\end{align*}
Let $\dot{\ell}_j$ be the partial derivative of $\ell$ with respect to $x_j$, and define
\begin{equation*}
  B(\vc{x})
  \coloneqq
  W_{\ell} (\vc{x}) - \sum_{j=1}^d \dot{\ell}_j (\vc{x}) \, W_{\ell,j} (x_j),
  \qquad \vc{x} \in [0,\infty)^d.
\end{equation*}
For $m \in \{1,\ldots,q\}$ with $\pi(m) = (\pi_1 (m),\pi_2 (m)) = (u,v)$, put
\begin{equation*}
B_{\pi(m)} (x_{\pi_1 (m)}, x_{\pi_2 (m)}) = B_{uv} (x_u, x_v) \coloneqq B(0,\ldots,0,x_u,0,\ldots,0,x_v,0,\ldots,0).
\end{equation*}
Also define the mean-zero random column vector
\begin{equation*}
\widetilde{B} \coloneqq \left( \int_{[0,1]^2} B_{\pi(m)} (x_{\pi_1 (m)}, x_{\pi_2 (m)}) \, \diff x_{\pi_1 (m)} \diff  x_{\pi_2 (m)} \right)_{m=1}^q.
\end{equation*}
The law of $\widetilde{B}$ is zero-mean Gaussian and its covariance matrix $\Gamma(\theta_0) \in \RR^{q \times q}$ depends on $\theta_0$ via the model assumption $\ell = \ell(\cdot\,;\theta_0)$.
For pairs $\pi(m) = (u,v)$ and $\pi(m') = (u',v')$, we can obtain the $(m,m')$-th entry of $\Gamma (\theta)$ by
\begin{equation}\label{eq:gamma}
\Gamma_{(m,m')} (\theta)  = \mathbb{E}[\widetilde{B}_{m} \widetilde{B}_{m'}] =  \int_{[0,1]^4} \mathbb{E} \left[  B_{u v} (x_{u},x_{v}) \, B_{u' v'} (x_{u'},x_{v'}) \right] \, \diff x_{u} \diff x_{v} \diff x_{u'} \diff x_{v'}.
\end{equation}
Define $\psi: \Theta \rightarrow \mathbb{R}^{q}$ by
\begin{equation}\label{eq:psi}
\psi (\theta)   \coloneqq  \left( \int_{[0,1]^2} \ell_{\pi(m)} (x_{\pi_1 (m)}, x_{\pi_2 (m)} ; \theta) \, \diff x_{\pi_1 (m)} \, \diff x_{\pi_2 (m)}
\right)_{m=1}^q.
\end{equation}
Assuming $\theta$ is an interior point of $\Theta$ and $\psi$ is differentiable in $\theta$, let $\dot{\psi} (\theta) \in \mathbb{R}^{q \times p}$ denote the total derivative of $\psi$ at $\theta$.

\begin{thm}[Existence, uniqueness and consistency]\label{resultmain1}
Let $\{ \ell(\cdot\,;\theta): \theta \in \Theta \}$, $\Theta \in \RR^p$, be a parametric family of $d$-variate stable tail dependence functions and let $(\pi(m))_{m=1}^q$, with $p \le q \le d(d-1)/2$, be $q$ distinct pairs in $\{1, \ldots, d\}$ such that the map $\psi$ in \eqref{eq:psi} is a homeomorphism from $\Theta$ to $\psi(\Theta)$. Let the $d$-variate distribution function $F$ have continuous margins and stable tail dependence function $\ell(\,\cdot\,;\theta_0)$ for some interior point $\theta_0 \in \Theta$. Let $\vc{X}_1, \ldots, \vc{X}_n$ be an iid sample from $F$.
Let $k = k_n \in \{1, \ldots, n\}$ satisfy $k \rightarrow \infty$ and $k/n \rightarrow 0$, as $n \rightarrow \infty$.
Assume also that
\begin{itemize}
\item[(C1)]
$\psi$ is twice continuously differentiable on a neighbourhood of $\theta_0$ and $\dot{\psi} (\theta_0)$ is of full rank;
\item[(C2)]
there exists a symmetric, positive definite matrix $\Omega$ such that $\Omegan \pto \Omega$ entry-wise.
\end{itemize}
Then with probability tending to one, the minimizer $\widehat{\theta}_n$ of $f_{n,k,\Omegan}$ exists and is unique. Moreover,
\begin{equation*}
\widehat{\theta}_n \overset{\mathbb{P}}{\rightarrow} \theta_0, \qquad \text{ as } n \rightarrow \infty.
\end{equation*}
\end{thm}
Let $\Delta_{d-1} = \{\vc{w} \in [0,1]^d : w_1 + \cdots + w_d =1\}$ denote the unit simplex in $\mathbb{R}^d$.
\begin{thm}[Asymptotic normality]\label{resultmain2}
If in addition to the assumptions of Theorem~\ref{resultmain1}
\begin{itemize}
\item[(C3)] $t^{-1} \Pr[1 - F_1(X_{11}) \leq t x_1 \textnormal{ or } \ldots \textnormal{ or } 1 - F_d (X_{1d}) \leq t x_d] - \ell(\vc{x} ; \theta_0) = O(t^{\alpha})$ uniformly in $\vc{x} \in \Delta_{d-1}$ as $t \downarrow 0$ for some $\alpha >0$;
\item[(C4)] $k = o(n^{2 \alpha/(1 + 2 \alpha)})$ and $k \rightarrow \infty$ as $n \rightarrow \infty$,
\end{itemize}
then
\begin{equation*}
\sqrt{k} \, (\widehat{\theta}_n - \theta_0) \dto \mathcal{N}_p(0,M(\theta_0))
\end{equation*}
where, for $\theta \in \Theta$ such that $\dot{\psi}(\theta)$ is of full rank,
\begin{equation}\label{eq:asym}
  M(\theta)
  \coloneqq
  \bigl( \dot{\psi} (\theta)^T \, \Omega \, \dot{\psi} (\theta) \bigr)^{-1} \,
  \dot{\psi} (\theta)^T \, \Omega \, \Gamma (\theta) \, \Omega \, \dot{\psi}(\theta) \,
  \bigl( \dot{\psi} (\theta)^T \, \Omega \, \dot{\psi} (\theta) \bigr)^{-1}.
\end{equation}
\end{thm}
The proofs of Theorems~\ref{resultmain1} and~\ref{resultmain2} are deferred to Appendix~\ref{sec:proofs}.

An asymptotically optimal choice for the random weight matrix $\Omegan$ would be one for which the limit $\Omega$ minimizes the asymptotic covariance matrix $M(\theta_0)$ with respect to the positive semi-definite partial ordering on the set of symmetric matrices. This minimization problem shows up in other contexts as well, and its solution is well-known: provided $\Gamma(\theta)$ is invertible, the minimum is attained at $\Omega = \Gamma(\theta)^{-1}$, the matrix $M(\theta)$ simplifying to
\begin{equation}
\label{eq:Mopt}
  M_{\text{opt}}(\theta) =
  \bigl( \dot\psi(\theta)^T \, \Gamma(\theta)^{-1} \, \dot\psi(\theta) \bigr)^{-1},
\end{equation}
see for instance \citet[page 339]{abadir2005}. However, this choice of the weight matrix requires the knowledge of $\theta_0$, which is unknown. One possible solution consists of computing the optimal weight matrix evaluated at a preliminary estimator of $\theta_0$.

For $\theta \in \Theta$, let $H_{\theta}$ be the spectral measure related to $\ell (\cdot\,;\theta)$ \citep{dehaan1977,resnick1987}: it is a finite measure defined on the unit simplex  $\Delta_{d-1}$ and it satisfies
\begin{equation*}
  \ell(\vc{x};\theta)
  =
  \int_{\Delta_{d-1}}
    \max_{j=1,\ldots,d} \left\{ w_j x_j \right\}
  H_{\theta} (\diff \vc{w}),
  \qquad \vc{x} \in [0, \infty)^d.
\end{equation*}

\begin{cor}[Optimal weight matrix]
\label{cor1}
In addition to the assumptions of Theorem~\ref{resultmain2}, assume the following:
\begin{itemize}
\item[(C5)]
for all $\theta$ in the interior of $\Theta$, the matrix $\Gamma(\theta)$ in \eqref{eq:gamma} has full rank;
\item[(C6)]
the mapping $\theta \mapsto H_{\theta}$ is weakly continuous at $\theta_0$.
\end{itemize}
Assume $\widehat{\theta}_n^{(0)}$ converges in probability to $\theta_0$ and let $\widehat{\theta}_n$ be the pairwise M-estimator with weight matrix $\Omegan = \Gamma(\widehat{\theta}_n^{(0)})^{-1}$. Then, with $M_{\textnormal{opt}}$ as in \eqref{eq:Mopt}, we have
\[
  \sqrt{k} ( \widehat{\theta}_n - \theta_0 )
  \dto
  \mathcal{N}_p(0, M_{\textnormal{opt}}(\theta_0)),
  \qquad n \to \infty.
\]
For any choice of the positive definite matrix $\Omega$ in \eqref{eq:asym}, the difference $M(\theta_0) - M_{\textnormal{opt}}(\theta_0)$ is positive semi-definite.
\end{cor}

In view of Corollary~\ref{cor1}, we propose the following two-step procedure:
\begin{enumerate}
\item  Compute the pairwise M-estimator $\widehat{\theta}_n^{(0)}$ with the weight matrix equal to the identity matrix, i.e., by minimizing $f_{n,k,I_q}$ in \eqref{eq:fnkIq}.
\item Calculate the pairwise M-estimator $\widehat{\theta}_n$ by minimizing $f_{n,k,\Omegan}$ with $\Omegan = \Gamma (\widehat{\theta}^{(0)}_n)^{-1}$.
\end{enumerate}
We will see in Section~\ref{applic} that this choice of $\Omegan$ indeed reduces the estimation error.

Calculating $M (\theta)$ can be a challenging task. The matrix $\Gamma (\theta)$ can become quite large since for a $d$-dimensional model, the maximal number of pairs is $d(d-1)/2$. In practice we will choose a smaller number of pairs: we will see in Section~\ref{applic} that this may even have a positive influence on the quality of our estimator. The entries of $\Gamma (\theta)$ are four-dimensional integrals of  $\mathbb{E}[B_{uv} (x_u,x_v) B_{u'v'} (x_{u'},x_{v'})]$ for $\pi(m) = (u,v)$ and $\pi(m') = (u',v')$, $m=1,\ldots,q$. The online supplementary material for this paper contains details on the calculation and implementation of the matrix $\Gamma (\theta)$. 

A natural competitor of the two-step procedure could be a one-step procedure where the weight matrix $\Gamma(\theta)^{-1}$ is recalculated within the minimisation routine. This resembles, but is substantially different from a continuously updating generalised method of moments \citep{hansen:heaton:yaron:1996}. Rather than as in equation~\eqref{eq:thetafinal}, the pairwise M-estimator of $\theta_0$ would be defined as the minimizer of the function
\[
  \theta \mapsto L_{n,k}(\theta)^T \, \Gamma(\theta)^{-1} \, L_{n,k}(\theta).
\]
Calculation of  $\Gamma(\theta)$ being time-consuming however, such an approach would be computationally unwieldy.

Finally, we present results that can be used for the construction of confidence regions and hypothesis tests.
\begin{cor}\label{cor2}
If the assumptions from Corollary \ref{cor1} are satisfied, then
\begin{equation*}
k (\widehat{\theta}_n - \theta_0)^T M(\widehat{\theta}_n)^{-1} (\widehat{\theta}_n - \theta_0) \overset{d}{\rightarrow} \chi_p^2, \qquad \text{ as } n \rightarrow \infty.
\end{equation*}
\end{cor}
Let $r < p$ and $\theta = (\theta_1,\theta_2) \in \Theta$ with $\theta_1 \in \mathbb{R}^{p-r}$ and $\theta_2 \in \mathbb{R}^r$. Suppose we want to test $\theta_2 = \theta^*_2$ against $\theta_2 \neq \theta^*_2$.
Write
$\widehat{\theta}_n = (\widehat{\theta}_{1n},\widehat{\theta}_{2n})$ and let $M_2 (\theta)$ be the $r \times r$ matrix corresponding to the lower right corner of $M(\theta)$.
\begin{cor}\label{cor3}
If the assumptions from Corollary \ref{cor1} are satisfied and if $\theta_0 = (\theta_1,\theta^*_2) \in \Theta$ for some $\theta_1$, then
\begin{equation*}
k (\widehat{\theta}_{2n} - \theta^*_2)^T M_2 (\widehat{\theta}_{1n},\theta^*_2)^{-1} (\widehat{\theta}_{2n} - \theta^*_2) \overset{d}{\rightarrow} \chi_r^2.
\end{equation*}
\end{cor}
We will not prove these corollaries here, since their proofs are straightforward extensions of those in \citet[Corollary 4.3; Corollary 4.4]{einmahl2012}.

\section{Spatial models}
\label{sec:spatial}

\subsection{Theory and definitions}\label{sec:theory}
The isotropic Brown--Resnick process on $\mathcal{S} \subset \mathbb{R}^2$ is given by
\begin{equation*}
Z(\vc{s}) =  \max_{i \in \mathbb{N}} \xi_i \exp{ \left\{ \epsilon_i (\vc{s}) - \gamma(\vc{s}) \right\}}, \qquad \vc{s} \in \mathcal{S},
\end{equation*}
where $\{\xi_i\}_{i \geq 1}$ is a Poisson process on $(0,\infty]$ with intensity measure $\xi^{-2} \,\diff \xi$ and $\{\epsilon_i (\cdot) \}_{i \geq 1}$ are independent copies of a Gaussian process with stationary increments, $\epsilon (0) = 0$, variance $2 \gamma (\cdot)$, and semi-variogram $\gamma (\cdot)$.
The process with $\gamma (\vc{s}) = (||\vc{s}|| / \rho)^{\alpha}$ appears as the only limit of (rescaled) maxima of stationary and isotropic Gaussian random fields \citep{kabluchko2009}; here $\rho > 0$ and $0 < \alpha  \leq 2$. Since isotropy may not be a reasonable assumption for many spatial applications, we follow \citet{blanchet2011} and \citet{engelke2014} and introduce a transformation matrix $V$ defined by
 \begin{equation*}
V \coloneqq V(\beta,c) \coloneqq
\begin{bmatrix}
\cos{\beta} & -\sin{\beta} \\
c \sin{\beta} & c \cos{\beta}
\end{bmatrix},
\qquad \beta \in [0,\pi/2), \, c > 0,
\end{equation*}
and a transformed space $\mathcal{S}' = \{V^{-1}\vc{s} : \vc{s} \in \mathcal{S}\}$, so that an isotropic process on $\mathcal{S}$ is transformed to an anisotropic process on $\mathcal{S}'$. For $\vc{s}' \in \mathcal{S}'$ we focus on the anisotropic Brown--Resnick process
\begin{equation}\label{eq:BR}
Z_V (\vc{s}') \coloneqq Z(V \vc{s}') = \max_{i \in \mathbb{N}} \xi_i \exp{ \left\{ \epsilon_i (V \vc{s}') - \gamma(V \vc{s}') \right\}},
\end{equation}
whose semi-variogram is defined by
 \begin{equation*}
\gamma_V (\vc{s}') \coloneqq \gamma (V \vc{s}') = \left[ {\vc{s}'}^T \frac{V^T V}{\rho^2} \vc{s}' \right]^{\alpha/2}.
 \end{equation*}
The pairwise stable tail dependence function for a pair $(u,v)$, corresponding to locations $(\vc{s}'_u,\vc{s}'_v)$,  is given by
\begin{equation*}
\ell_{uv}(x_u,x_v) = x_u \Phi \bigg( \frac{a_{uv}}{2} + \frac{1}{a_{uv}} \log{\frac{x_u}{x_v}} \bigg)+ x_v \Phi \bigg( \frac{a_{uv}}{2} + \frac{1}{a_{uv}} \log{\frac{x_v}{x_u}} \bigg),
\end{equation*}
where $a_{uv} \coloneqq \sqrt{2 \gamma_V(\vc{s}'_u-\vc{s}'_v)}$ and $\Phi$ is the standard normal distribution function. Observe that the choice $\alpha = 2$ leads to
\begin{equation*}
a^2_{uv} = 2 \gamma(V(\vc{s}'_u-\vc{s}'_v)) = (\vc{s}'_u-\vc{s}'_v)^T  \Sigma^{-1} (\vc{s}'_u-\vc{s}'_v) , \quad \text{ for some } \Sigma =
\begin{bmatrix}
\sigma_{11} & \sigma_{12} \\
\sigma_{12} & \sigma_{22}
\end{bmatrix},
\end{equation*}
where $\Sigma$ represents any valid $2 \times 2$ covariance matrix.
This submodel is known as the Gaussian extreme value process or simply the Smith model \citep{smith1990}. We will present simulation studies for processes in the domain of attraction, in the sense of \eqref{eq:ell2}, of both the Smith model and the anisotropic Brown--Resnick process.
To calculate the weight matrix $\Gamma(\theta)^{-1}$, we will need to compute integrals over the four-dimensional margins of the stable tail dependence function, see~\eqref{eq:gamma} and the online supplementary material for this paper. 

In \cite{huser2013} the following representation is given for $\ell(x_1,\ldots,x_d ; \theta)$ for general $d$.
If $Z_V $ is defined as in \eqref{eq:BR} then for $\vc{s}_{u_1},\ldots,\vc{s}_{u_d} \in \mathbb{R}^2$
\begin{equation*}
\ell_{u_1,\ldots,u_d} (x_1,\ldots,x_d) =  \sum_{i=1}^d  x_i \Phi_{d-1} ( \eta^{(i)} (1/x) ; R^{(i)} ) ,
\end{equation*}
where
\begin{align*}
\eta^{(i)} (x) & = (\eta_1^{(i)} (x_1,x_i), \ldots, \eta_{i-1}^{(i)} (x_{i-1},x_i), \eta_{i+1}^{(i)} (x_{i+1},x_i) , \ldots , \eta_d^{(i)} (x_d,x_i) ) \,\,\, \in \mathbb{R}^{d-1}, \\
\eta_j^{(i)} (x_j,x_i) & = \sqrt{\frac{\gamma_V (s_{u_i} - s_{u_j})}{2}} + \frac{\log{(x_j/x_i)}}{\sqrt{2 \gamma_V (s_{u_i} - s_{u_j})}} \,\,\, \in \mathbb{R},
\end{align*}
and $R^{(i)} \in \mathbb{R}^{(d-1) \times (d-1)}$ is the correlation matrix with entries
\begin{equation*}
R^{(i)}_{jk} = \frac{\gamma_V (s_{u_i} - s_{u_j}) + \gamma_V(s_{u_i} - s_{u_k}) - \gamma_V(s_{u_j} - s_{u_k})}{2 \sqrt{\gamma_V(s_{u_i}-s_{u_j}) \gamma_V(s_{u_i} - s_{u_k})}}, \qquad j,k=1,\ldots,d; \, j,k \neq i.
\end{equation*}

\subsection{Simulation studies}\label{applic}
In order to study the performance of the pairwise M-estimator when the underlying distribution function $F$ satisfies \eqref{eq:ell2} for a function $\ell$ corresponding to the max-stable models described before, we generate random samples from Brown--Resnick processes and Smith models perturbed with additive noise. If $\vc{Z}= (Z_1,\ldots,Z_{d})$ is a max-stable process observed at $d$ locations, then we consider
\begin{equation*}
X_j = Z_j + \epsilon_j, \qquad j=1,\ldots,d,
\end{equation*}
where $\epsilon_j$ are independent half normally distributed random variables, corresponding to the absolute value of a normally distributed random variable with standard deviation $1/2$. All simulations are done in \textsf{R} \citep{Rmanual}. Realizations of $Z$ are simulated using the \textsf{SpatialExtremes} package \citep{ribatet}.

\paragraph{Perturbed max-stable processes on a large grid.\\}
Assume that we have $d=100$ locations on a $10 \times 10$ unit distance grid. We simulate 500 samples of size $n=500$ from the perturbed Smith model with parameters
\begin{equation*}
\Sigma =
\begin{bmatrix}
1.0 & 0.5 \\
0.5 & 1.5
\end{bmatrix},
\end{equation*}
and from a perturbed anisotropic Brown--Resnick process with parameters $\alpha = 1$, $\rho = 3$, $\beta=0.5$ and $c = 0.5$. Instead of estimating $\rho$, $\beta$, and $c$ directly, we estimate the three parameters of the matrix
\begin{equation*}
\mathcal{T} = \begin{bmatrix}
\tau_{11} & \tau_{12} \\
\tau_{12} & \tau_{22}
\end{bmatrix}
= \rho^{-2} \, V (\beta,c)^T V (\beta,c).
\end{equation*}
In practice, this parametrization, which is in line with the one of the Smith model, often yields better results.
We study the bias and root mean squared error (RMSE) for $k \in \{25,50,75,100\}$.
We compare the estimators for two sets of pairs: one containing all pairs ($q = 4950$) and one containing only pairs of neighbouring locations ($q = 342$). Although the first option may sound like a time-consuming procedure, estimation of the parameters for one sample takes about 20 seconds for the Smith model and less than two minutes for the anisotropic Brown--Resnick process. We let the weight matrix $\Omega$ be the $q \times q$ identity matrix, since for so many pairs a data-driven computation of the optimal weight matrix is too time-consuming.
Figure \ref{fig:spatsmith} shows the bias and RMSE of $(\sigma_{11},\sigma_{22},\sigma_{12})$ for the Smith model.
We see that great improvements are achieved by using only pairs of neighbouring locations and that the thus obtained estimator performs well. Using all pairs causes the parameters to have a large positive bias, which translates into a high RMSE. In general, distant pairs often lead to less dependence and hence less information about $\ell$ and its parameters.  Observe that small values of $k$ are preferable, i.e. $k = 25$ or $k=50$.

Figure \ref{fig:brerr} shows the bias and RMSE of the pairwise M-estimators of $(\alpha,\rho,\beta,c)$ for the anisotropic Brown--Resnick process. We see again that using only pairs of neighbouring locations improves the quality of estimation.
The corresponding estimators perform well for the estimation of $\alpha$, $\beta$, and $c$. The lesser performance when estimating $\rho$ seems to be inherent to the Brown--Resnick process and appears regardless of the estimation procedure: see for example \citet{engelke2014} or \citet{wadsworth2014}, who both report a positive bias of $\rho$ for small sample sizes. 
Compared to those for the Smith model, the values of $k$ for which the estimation error is smallest are higher, i.e., $k = 50$ or $k = 75$.

\begin{figure}[!ht]
\centering
\subfloat{\includegraphics[width=0.3\textwidth]{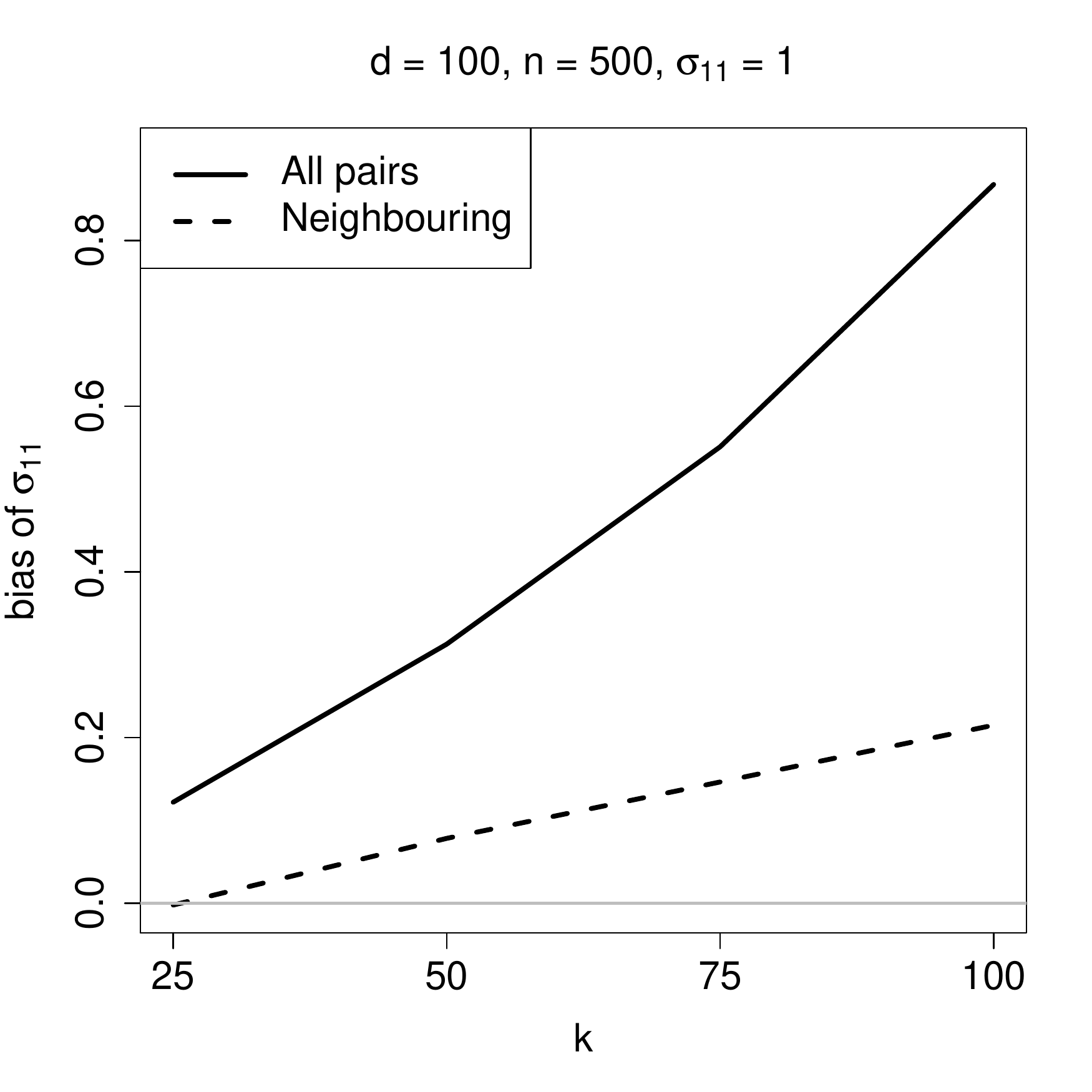}}
\subfloat{\includegraphics[width=0.3\textwidth]{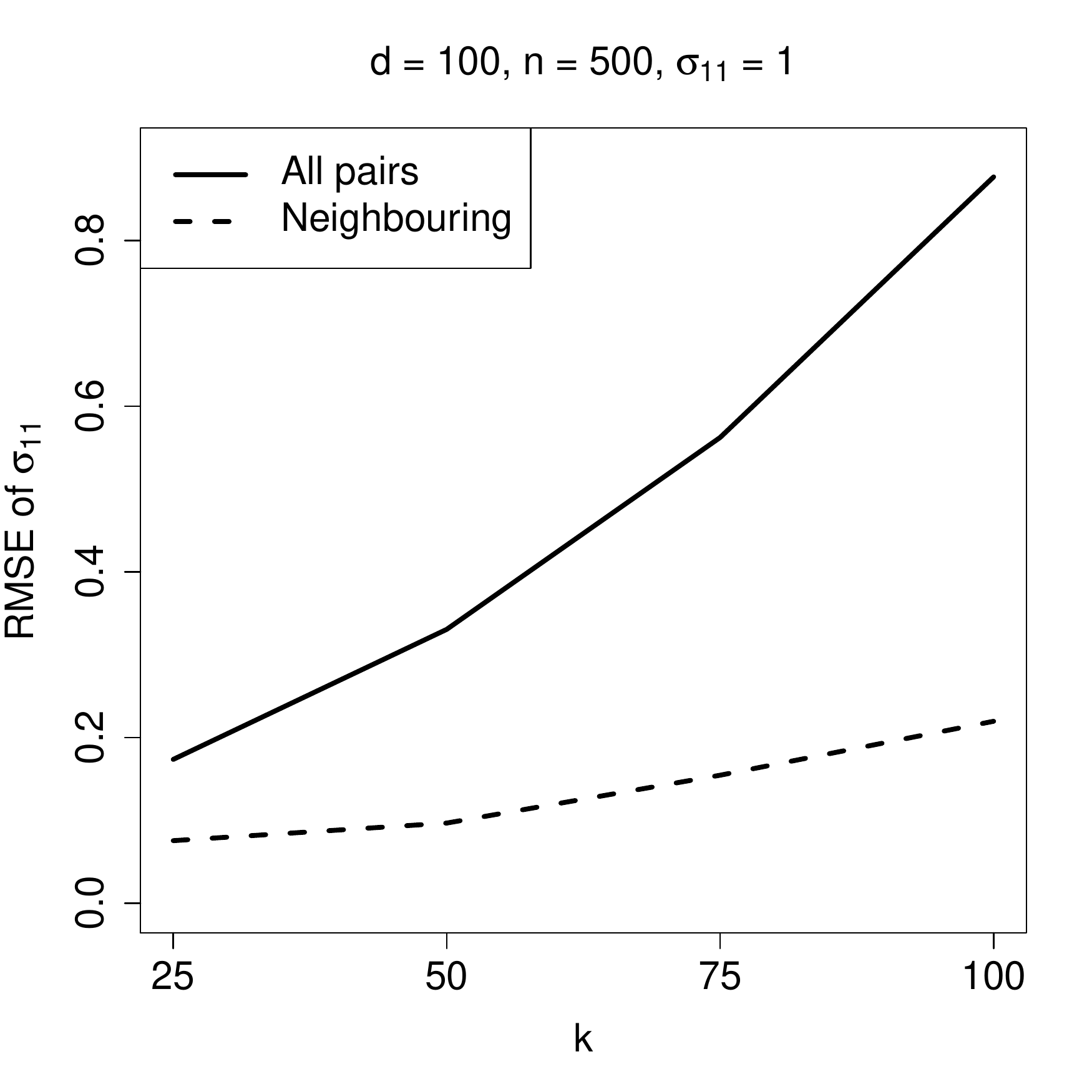}} \\
\subfloat{\includegraphics[width=0.3\textwidth]{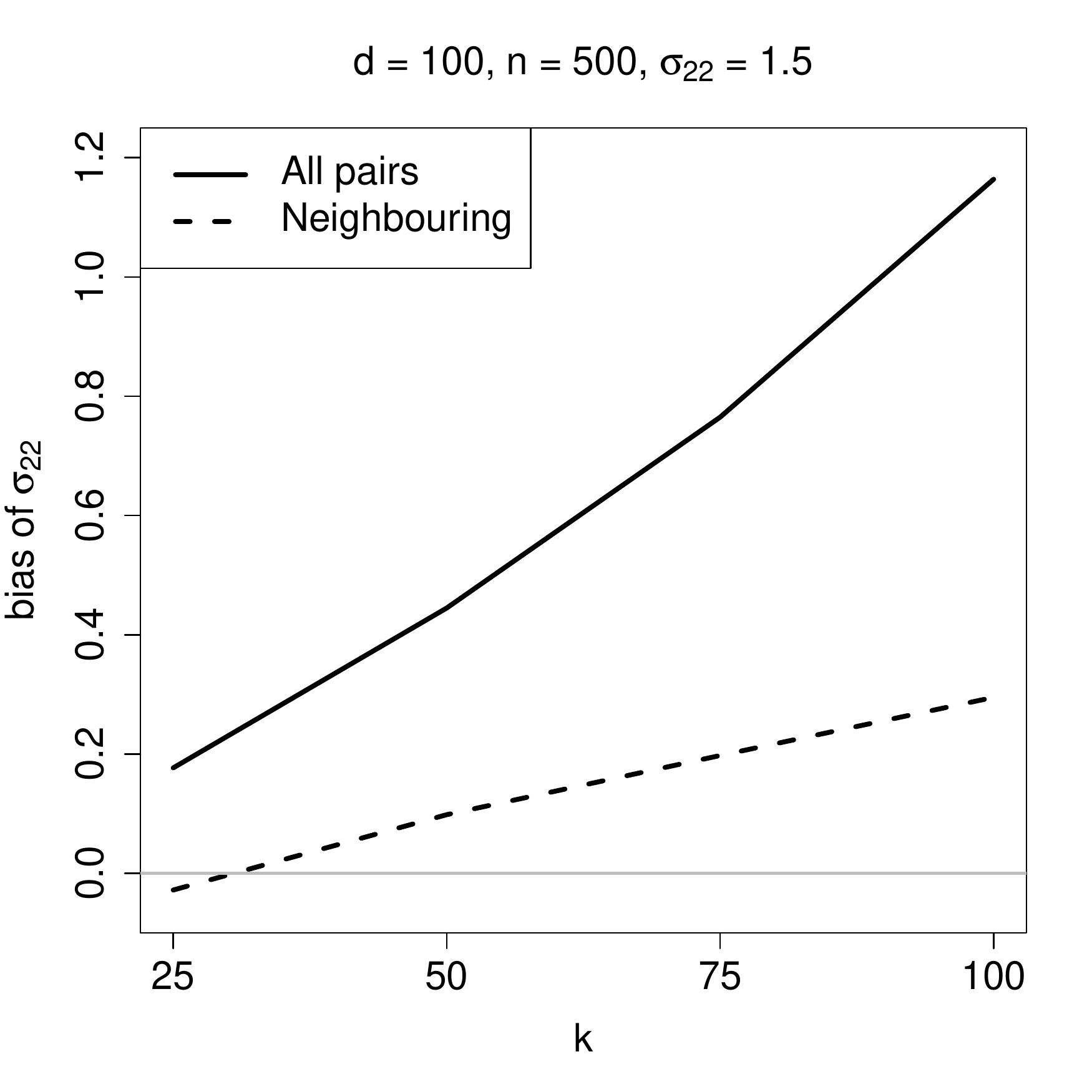}}
\subfloat{\includegraphics[width=0.3\textwidth]{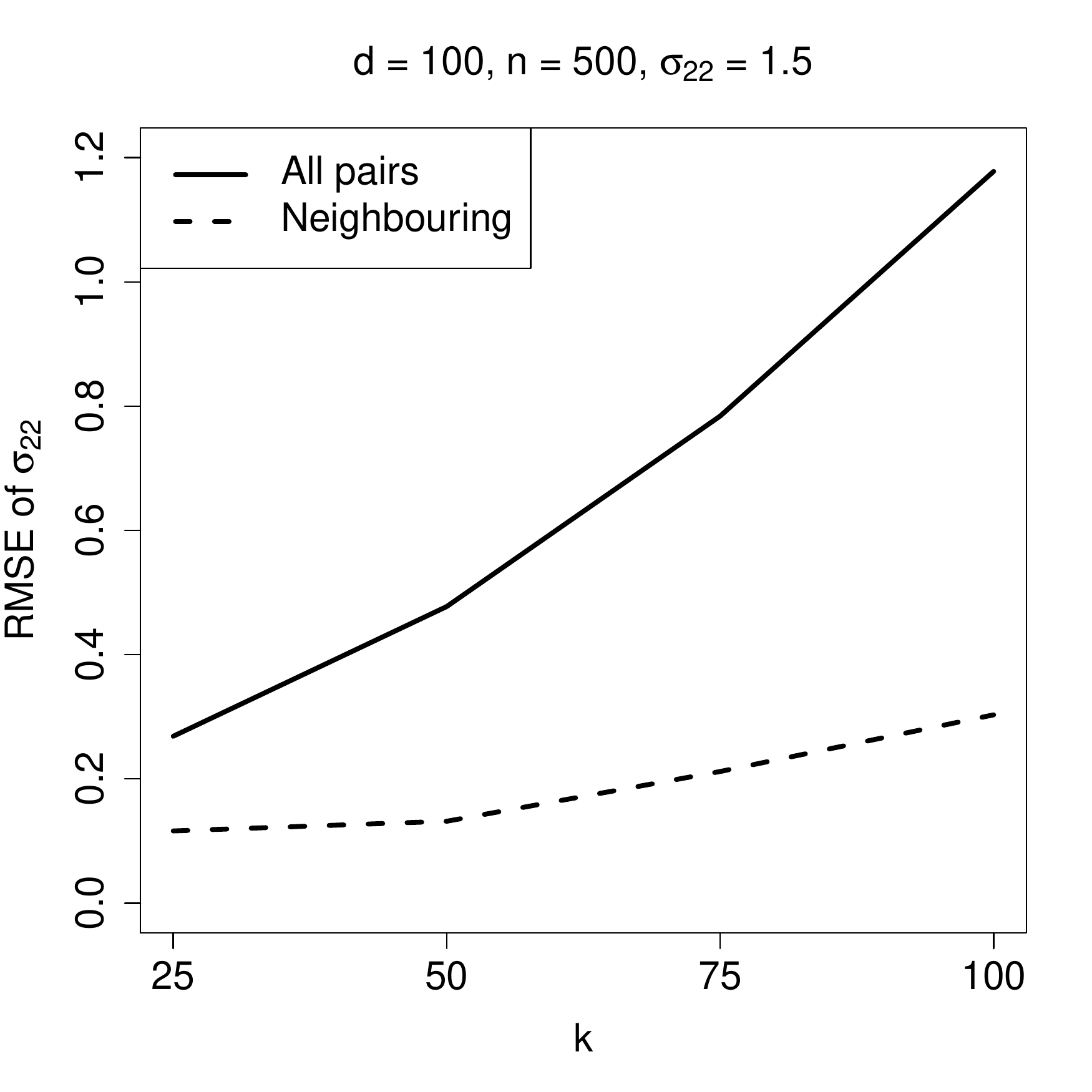}} \\
\subfloat{\includegraphics[width=0.3\textwidth]{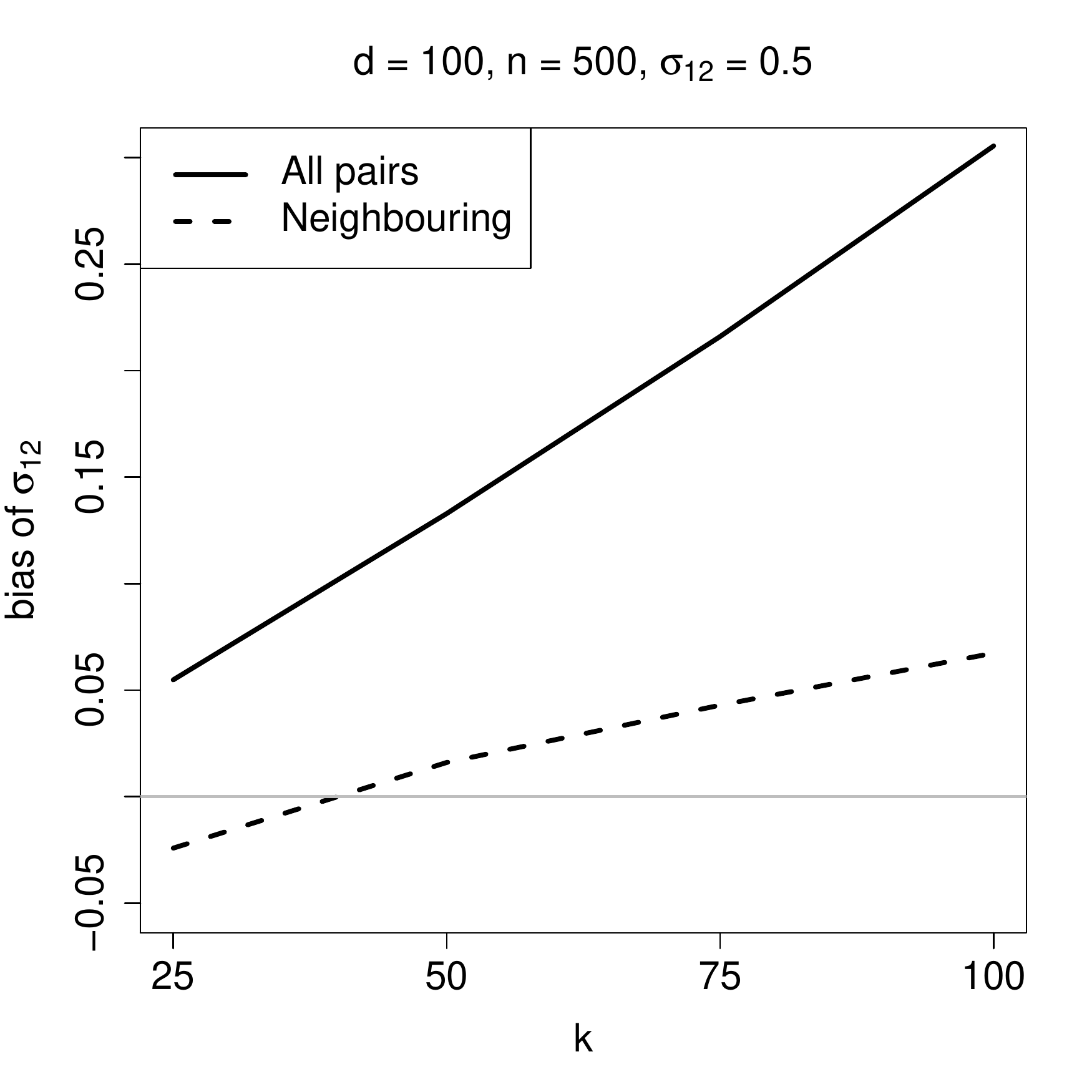}}
\subfloat{\includegraphics[width=0.3\textwidth]{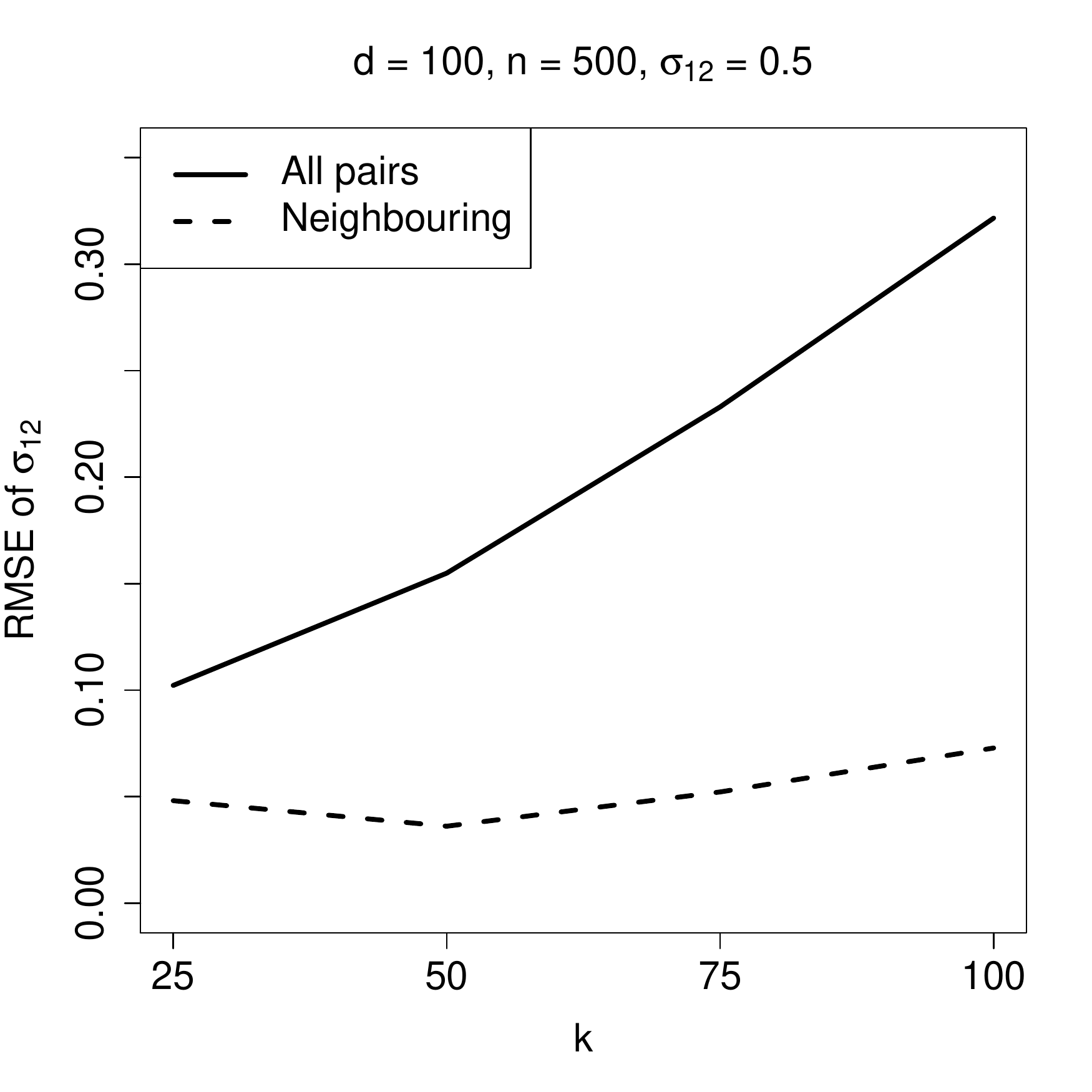}}
\caption{Bias and RMSE for estimators of $\sigma_{11} = 1$ (top), $\sigma_{22} = 1.5$ (middle), and $\sigma_{12} = 0.5$ (bottom) for the perturbed 100-dimensional Smith model with identity weight matrix; 500 samples of $n=500$.}
\label{fig:spatsmith}
\end{figure}

\begin{figure}[p]
\centering
\subfloat{\includegraphics[width=0.3\textwidth]{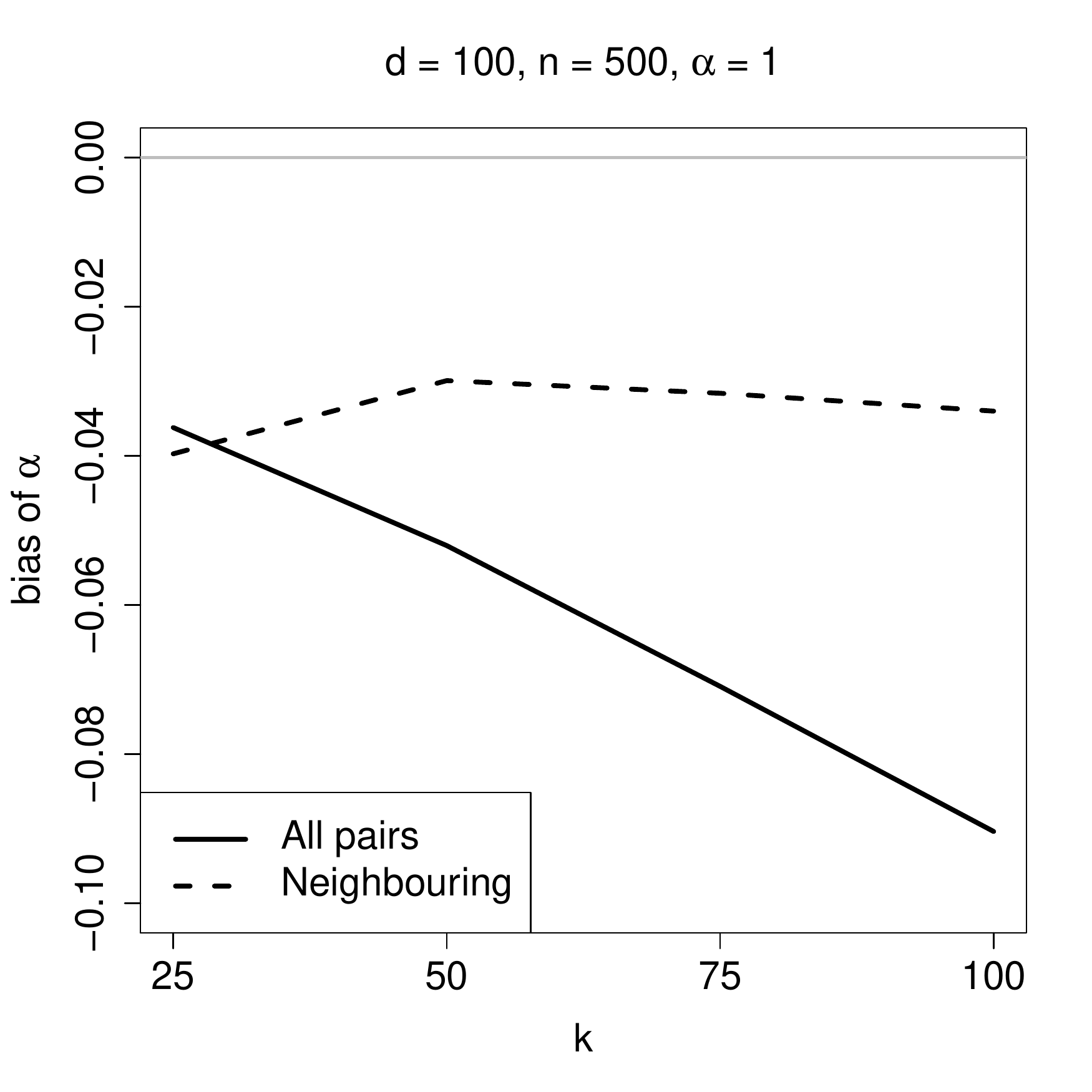}}
\subfloat{\includegraphics[width=0.3\textwidth]{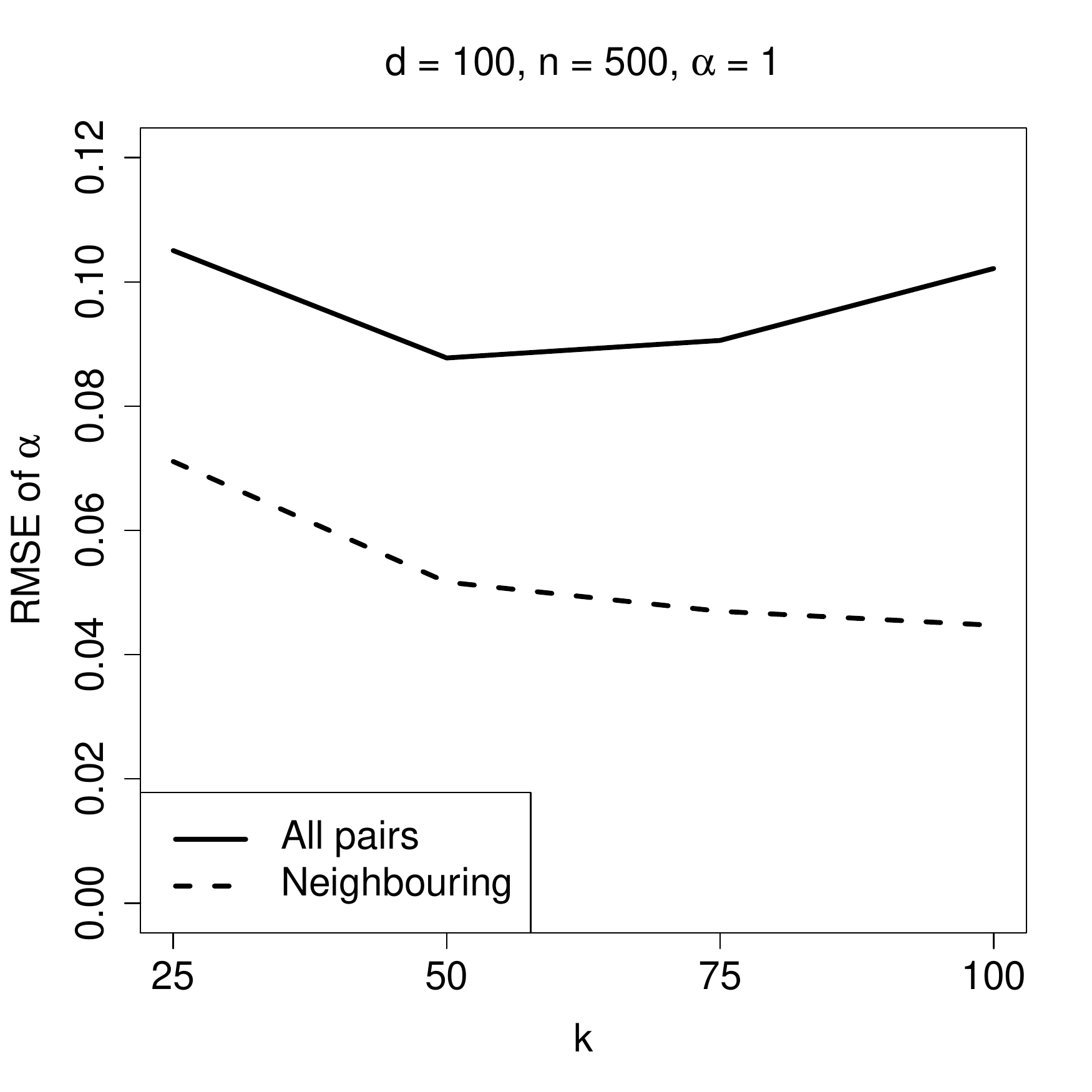}} \\
\subfloat{\includegraphics[width=0.3\textwidth]{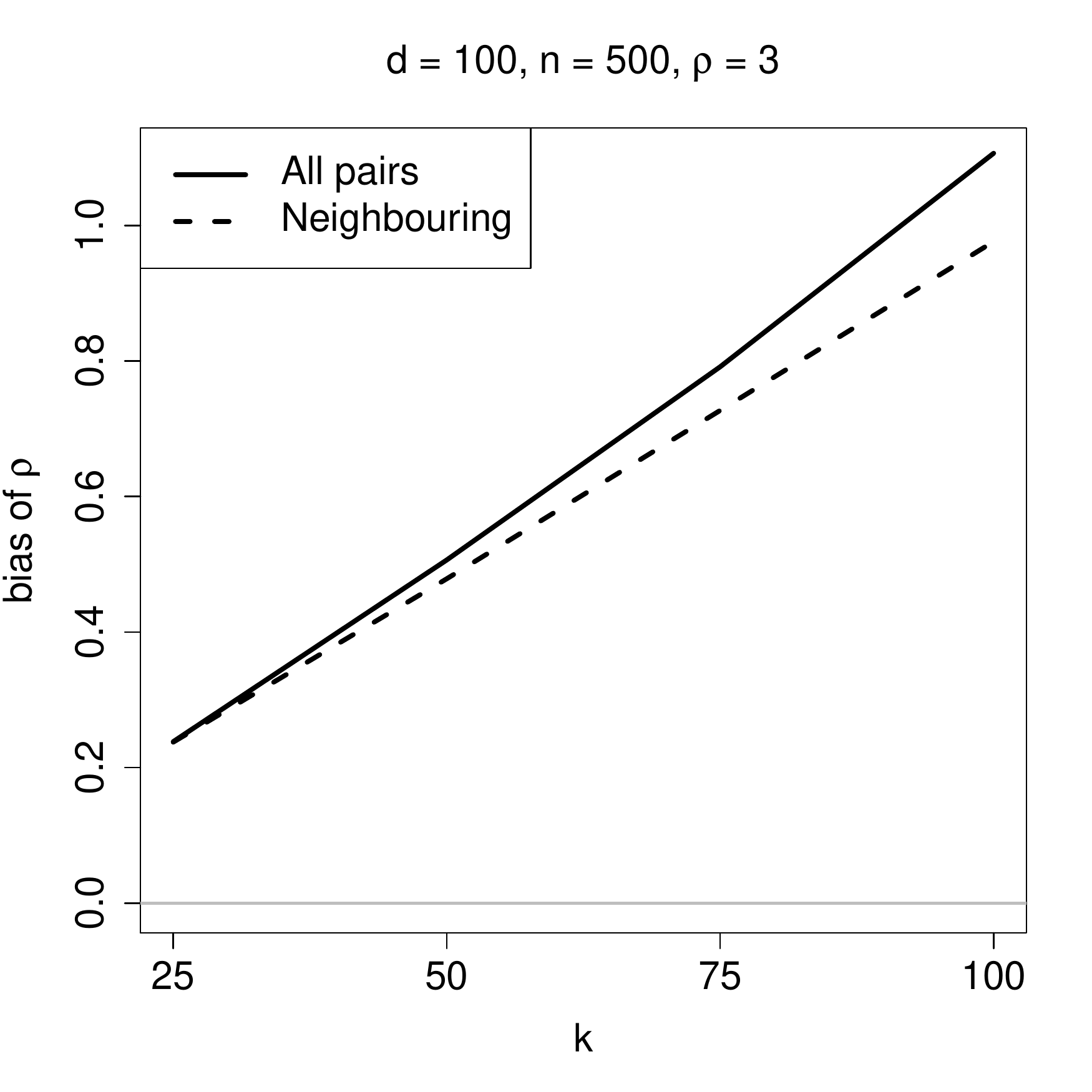}}
\subfloat{\includegraphics[width=0.3\textwidth]{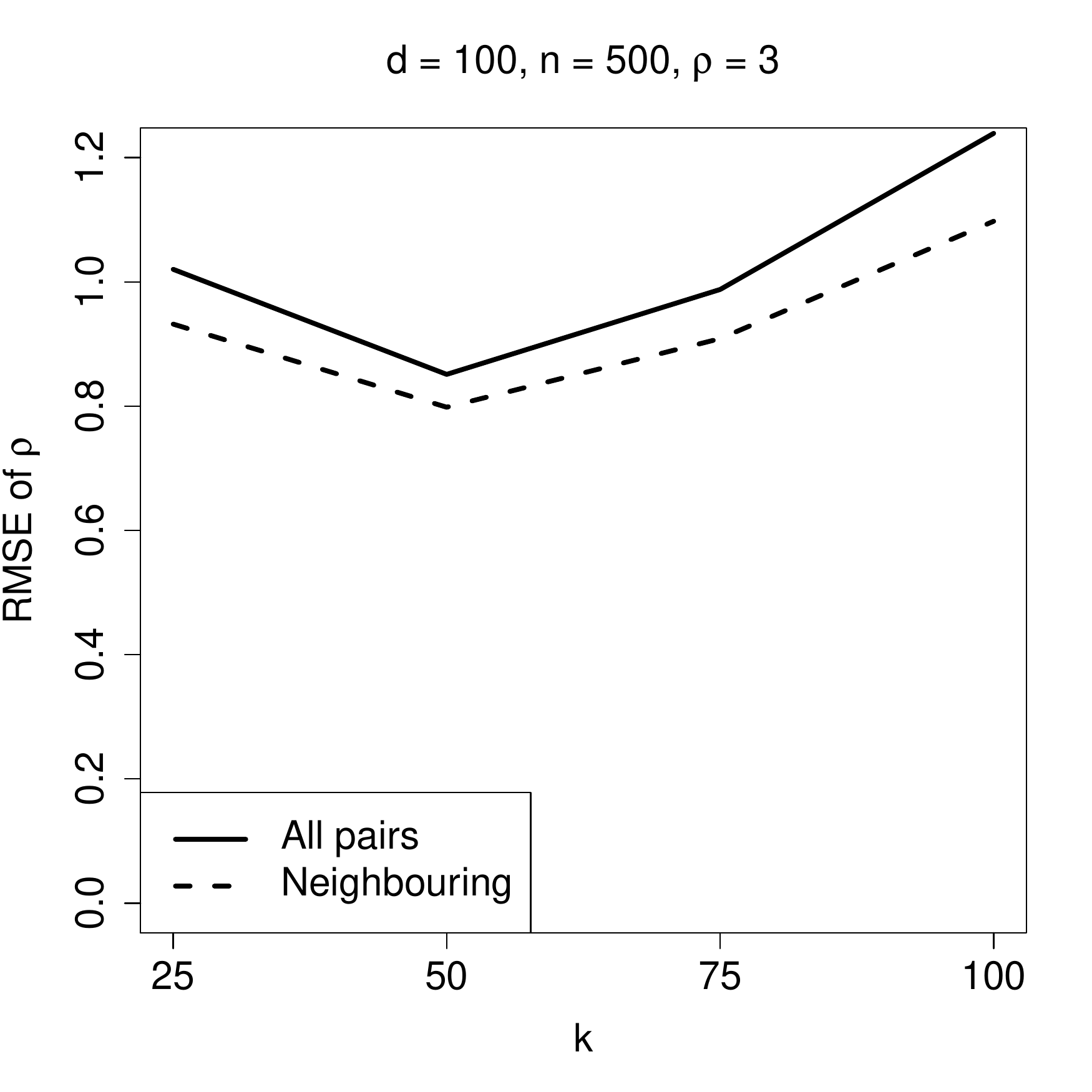}} \\
\subfloat{\includegraphics[width=0.3\textwidth]{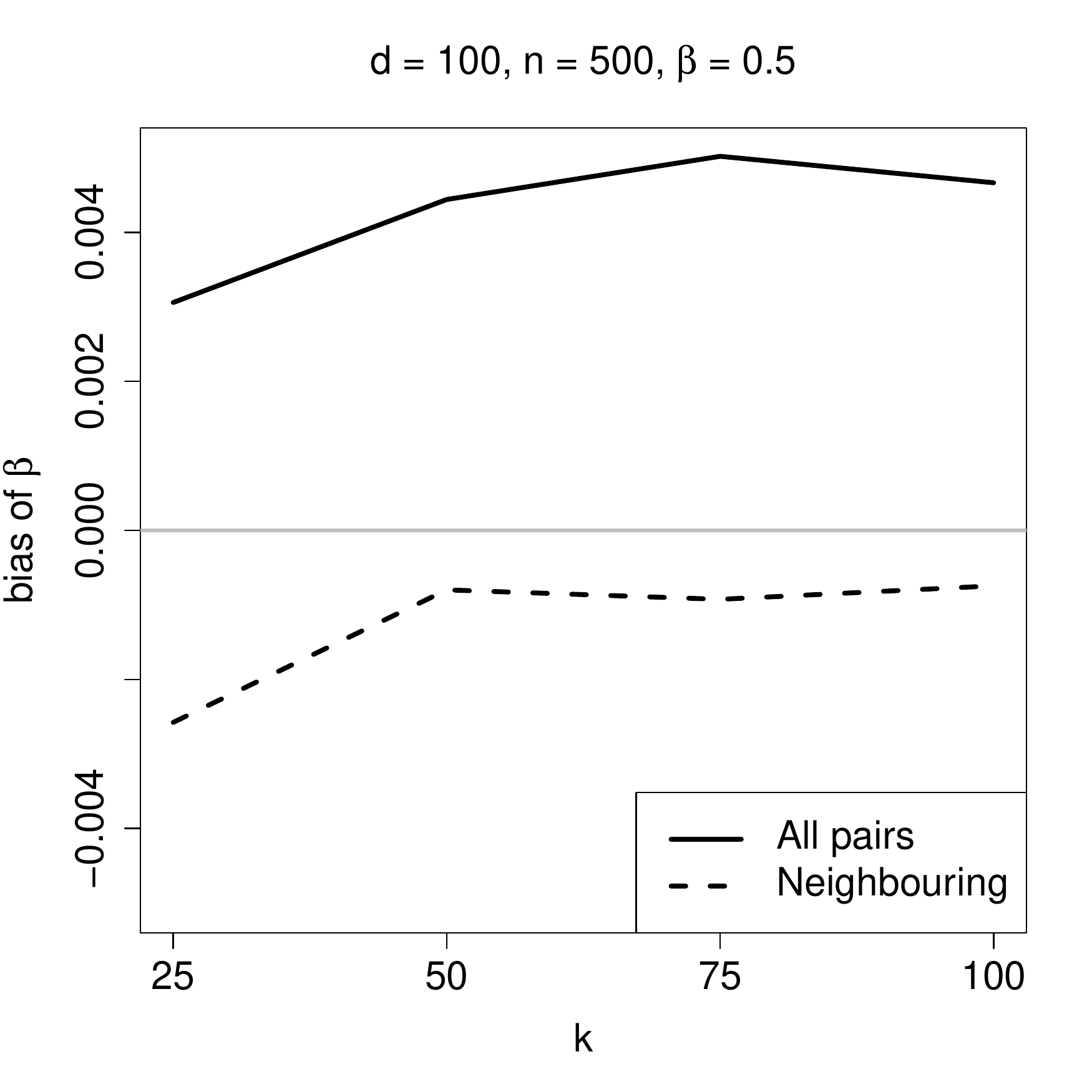}}
\subfloat{\includegraphics[width=0.3\textwidth]{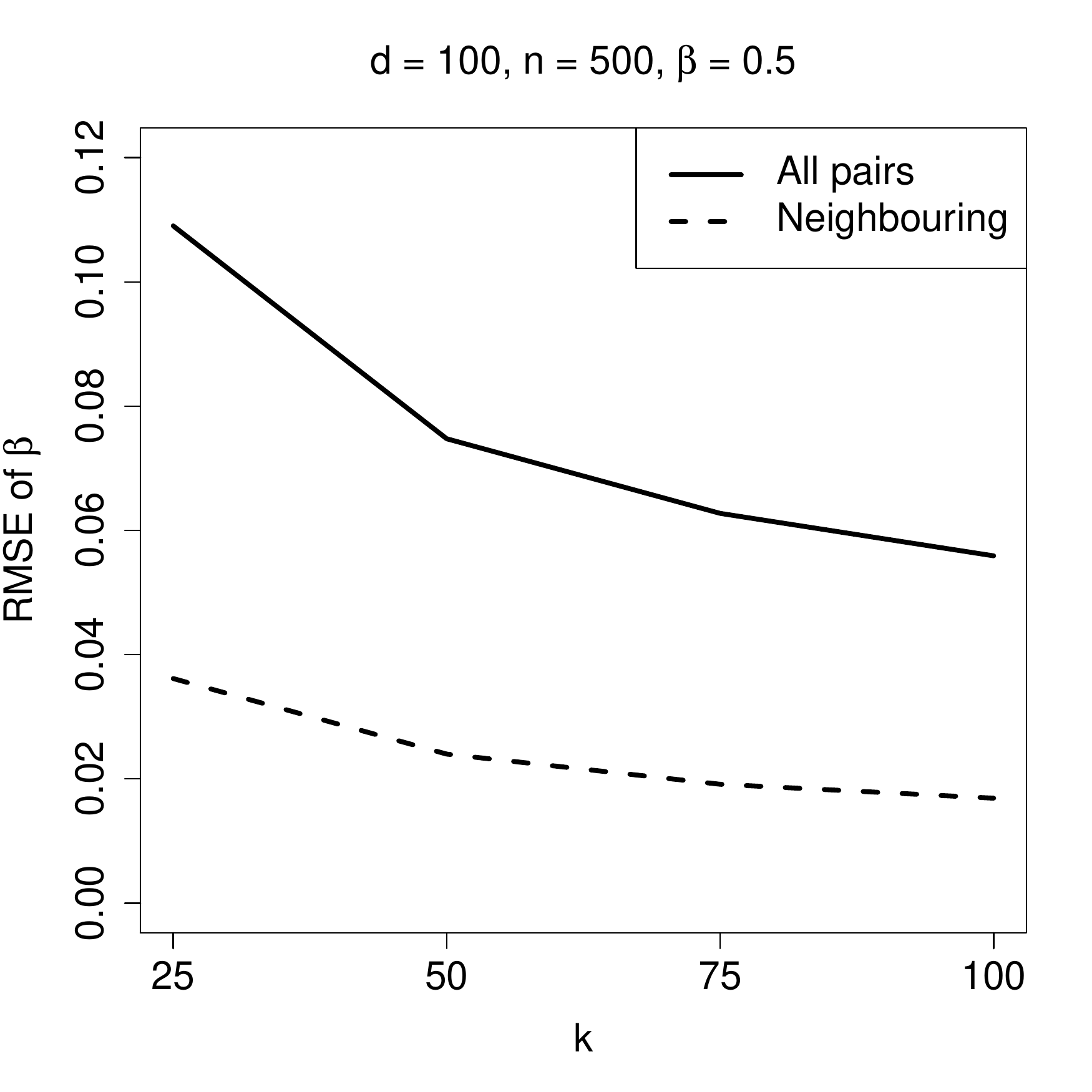}} \\
\subfloat{\includegraphics[width=0.3\textwidth]{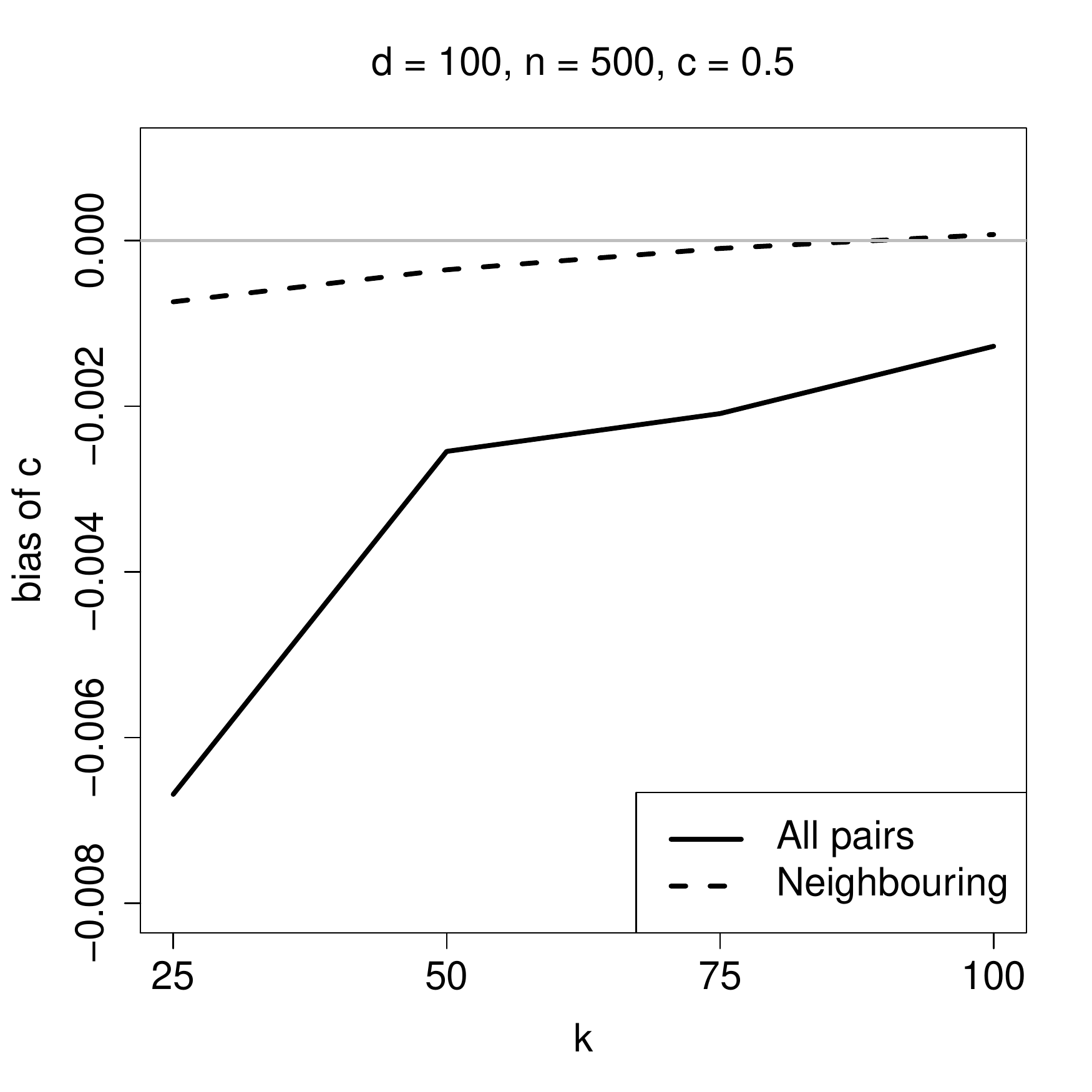}}
\subfloat{\includegraphics[width=0.3\textwidth]{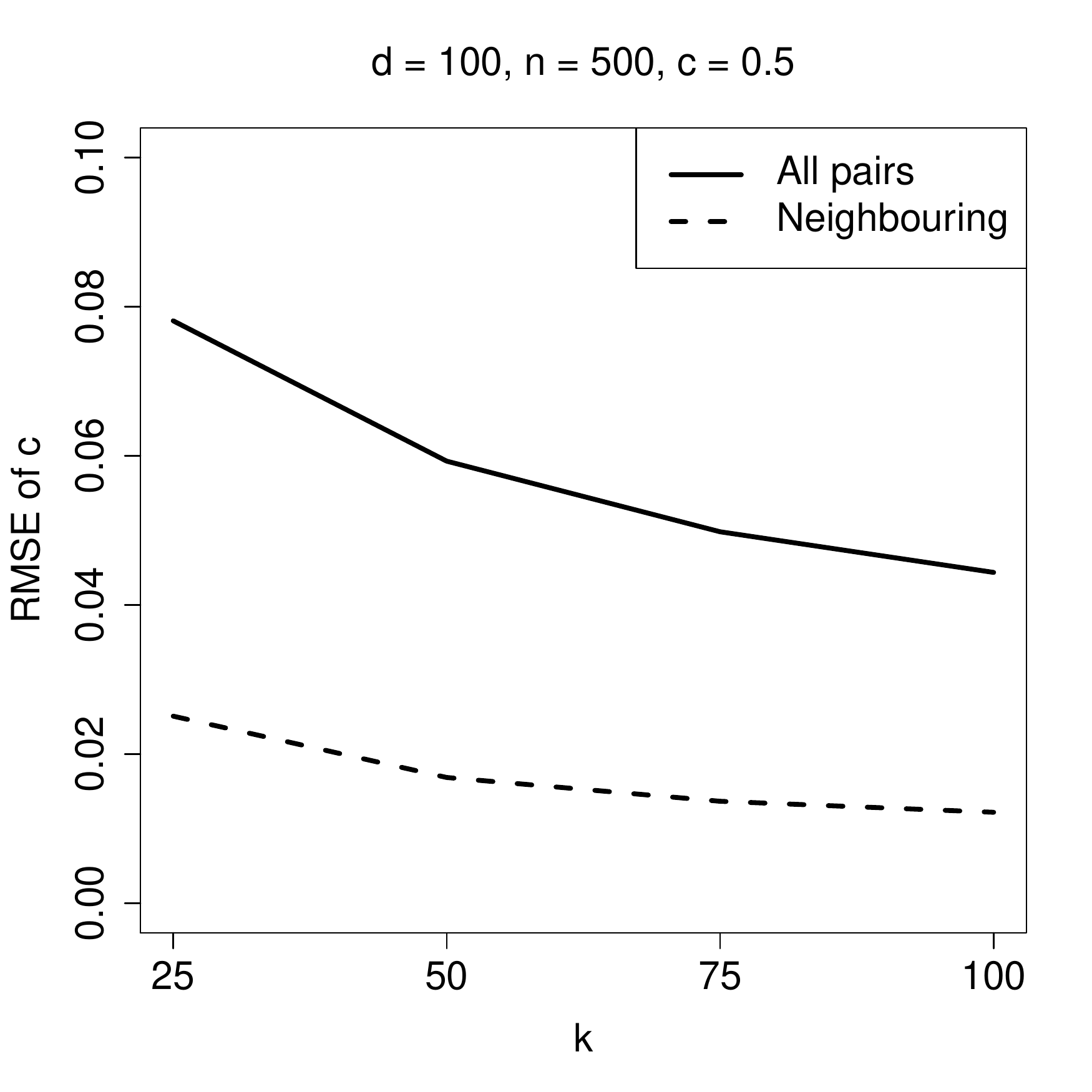}}
\caption{Bias and RMSE for estimators of $\alpha = 1$ (top), $\rho=3$ (top middle), $\beta = 0.5$ (bottom middle) and $c = 0.5$ (bottom) for the perturbed 100-dimensional Brown--Resnick process with identity weight matrix; 500 samples of $n=500$.}
\label{fig:brerr}
\end{figure}

\paragraph{A perturbed Brown--Resnick process on a small grid with optimal weight matrix.\\}
We consider $d=12$ locations on an equally spaced unit distance $4 \times 3$ grid. We simulate 500 samples of size $n=1000$ from an anisotropic Brown--Resnick process with parameters $\alpha = 1.5$, $\rho = 1$, $\beta= 0.25$ and $c = 1.5$. We study the bias, standard deviation, and RMSE for $k \in \{25,75,125\}$. In Figure~\ref{fig:brerr2}, three estimation methods are compared: one involving all pairs ($q = 66$), one involving only pairs of neighbouring locations ($q=29$), and one using optimal weight matrices chosen according to the two-step procedure described after Corollary~3.3, based on the $29$ pairs of neighbouring locations. In line with Corollary~3.3, the weighted estimators have lower (or equal) standard deviation (and RMSE) than the unweighted estimators. The difference is clearest for low $k$ for $\alpha$ and $\rho$.

\begin{figure}[p]
\centering
\subfloat{\includegraphics[width=0.3\textwidth]{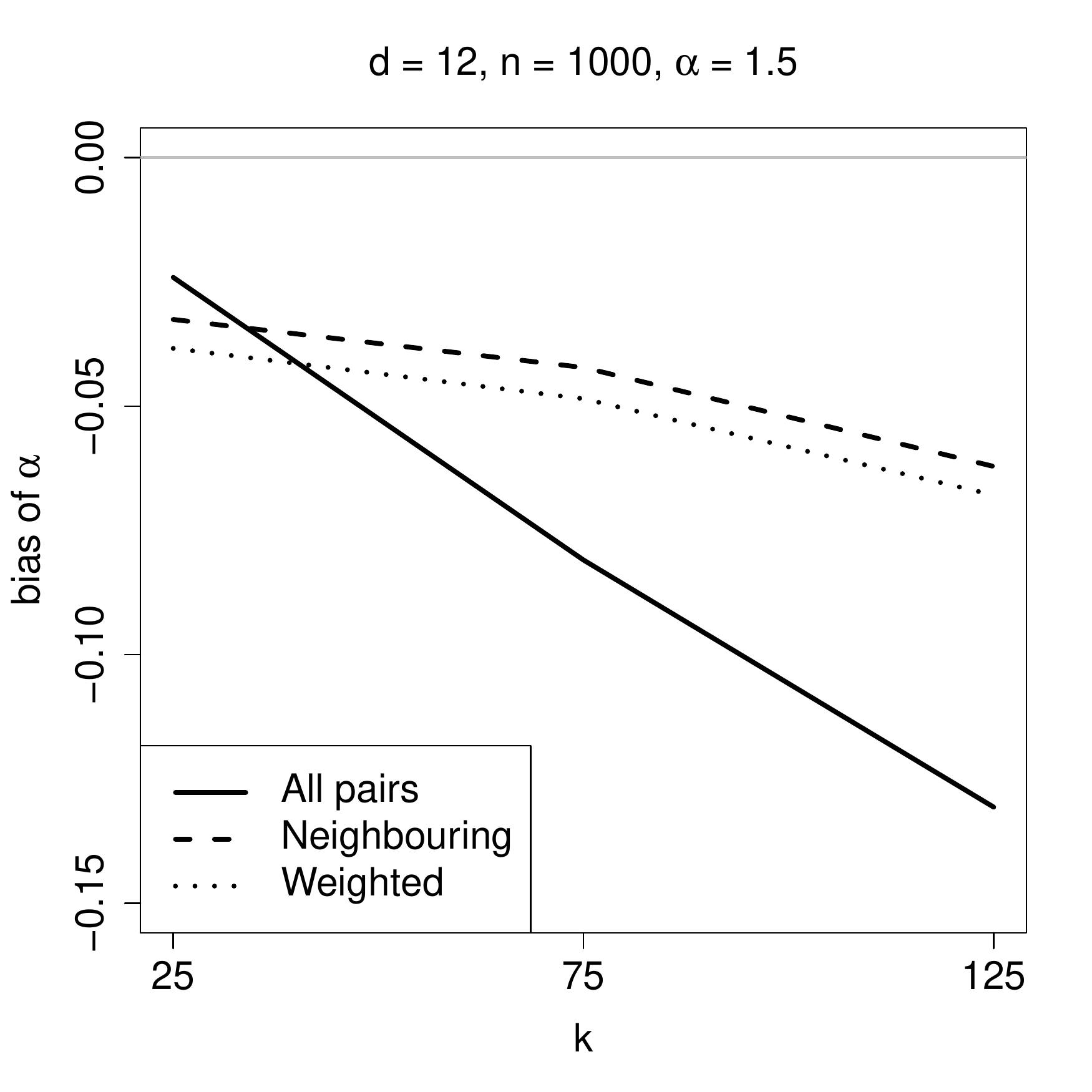}}
\subfloat{\includegraphics[width=0.3\textwidth]{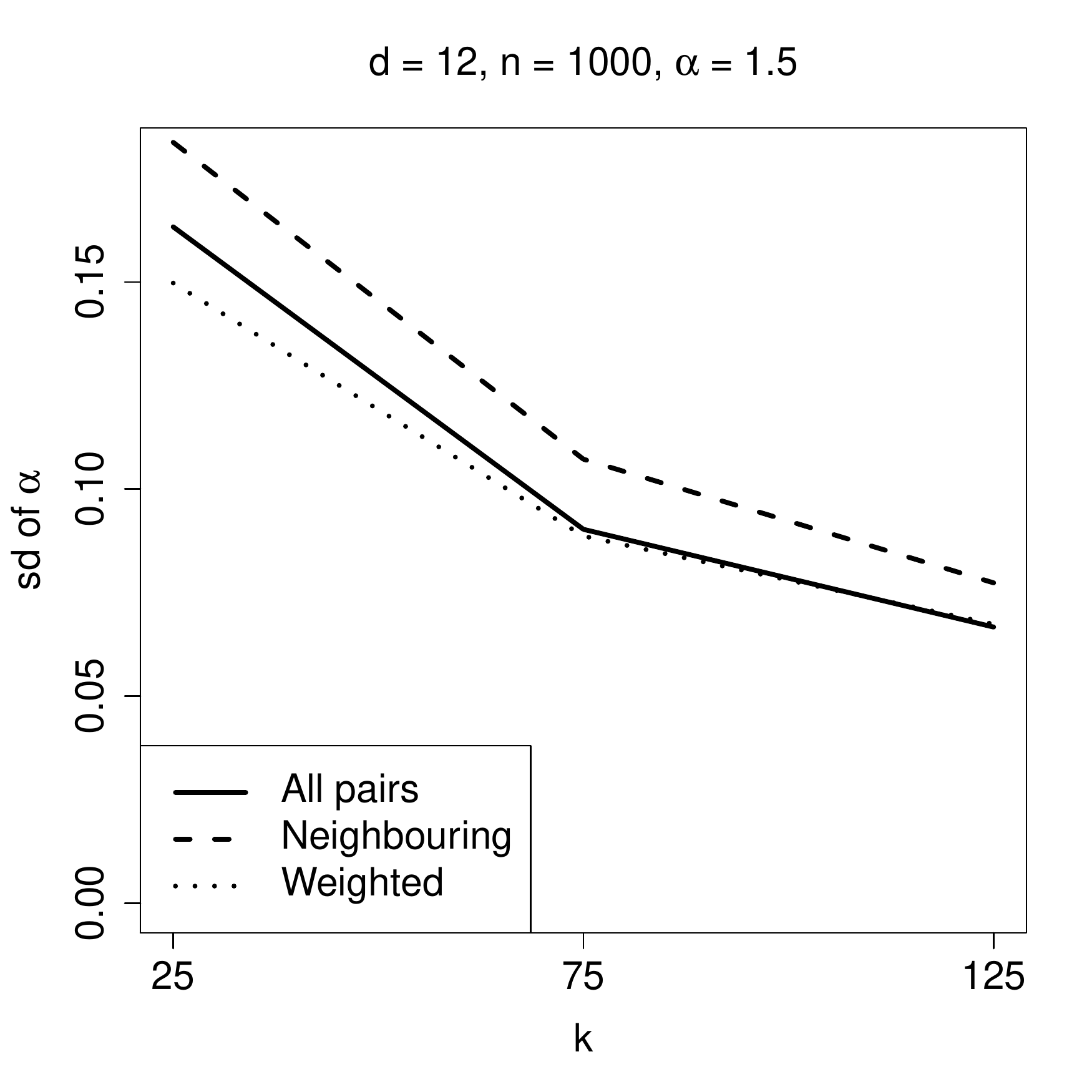}}
\subfloat{\includegraphics[width=0.3\textwidth]{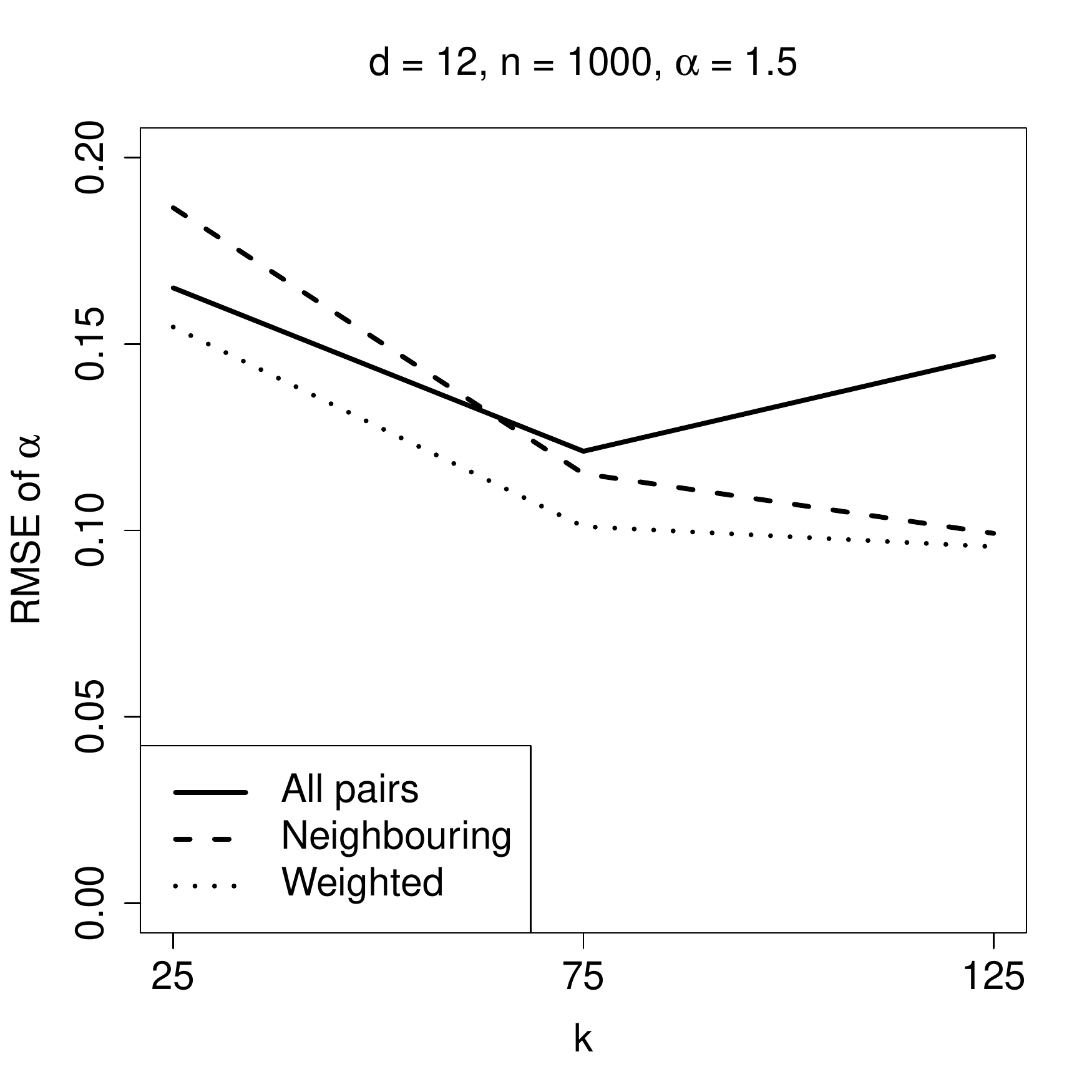}} \\
\subfloat{\includegraphics[width=0.3\textwidth]{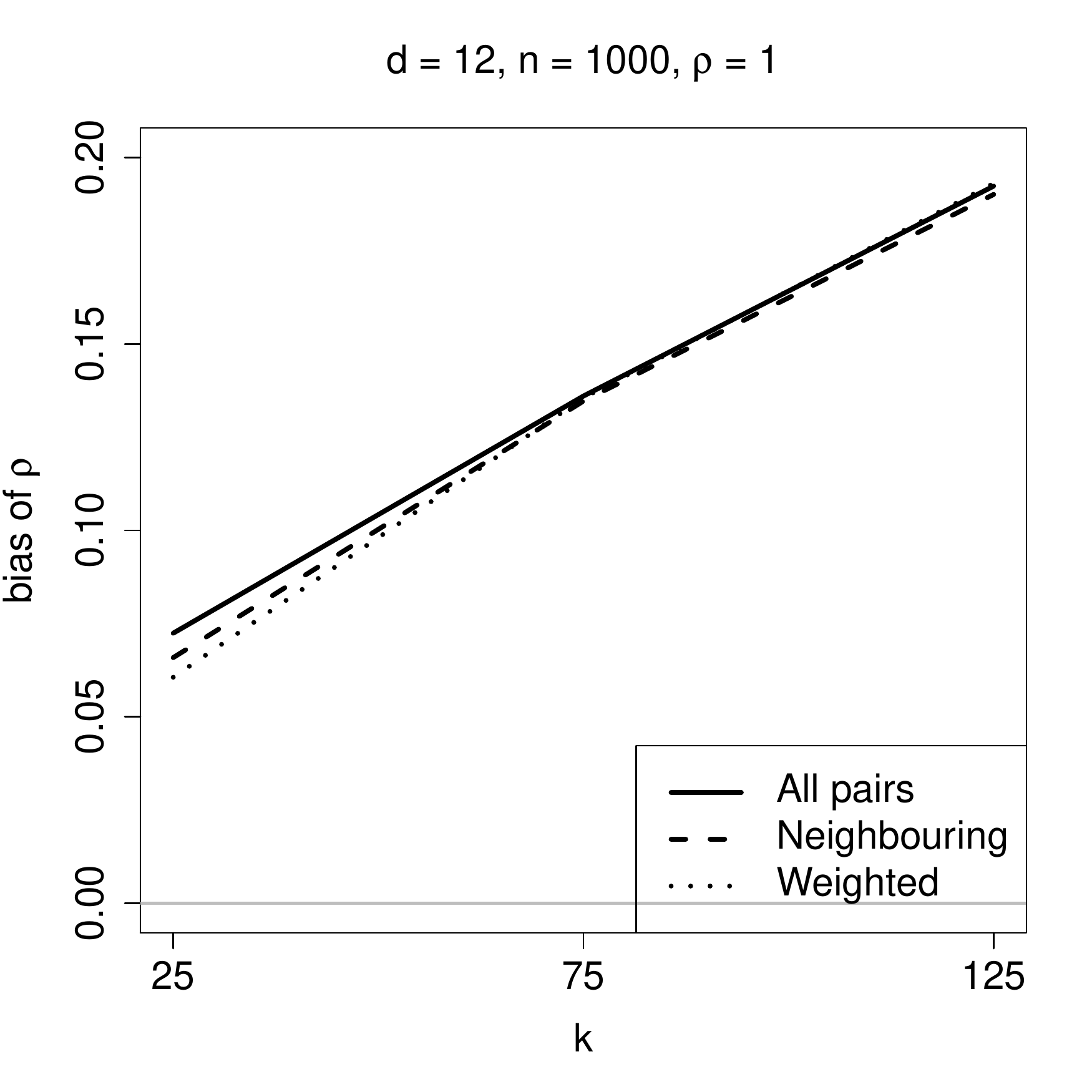}}
\subfloat{\includegraphics[width=0.3\textwidth]{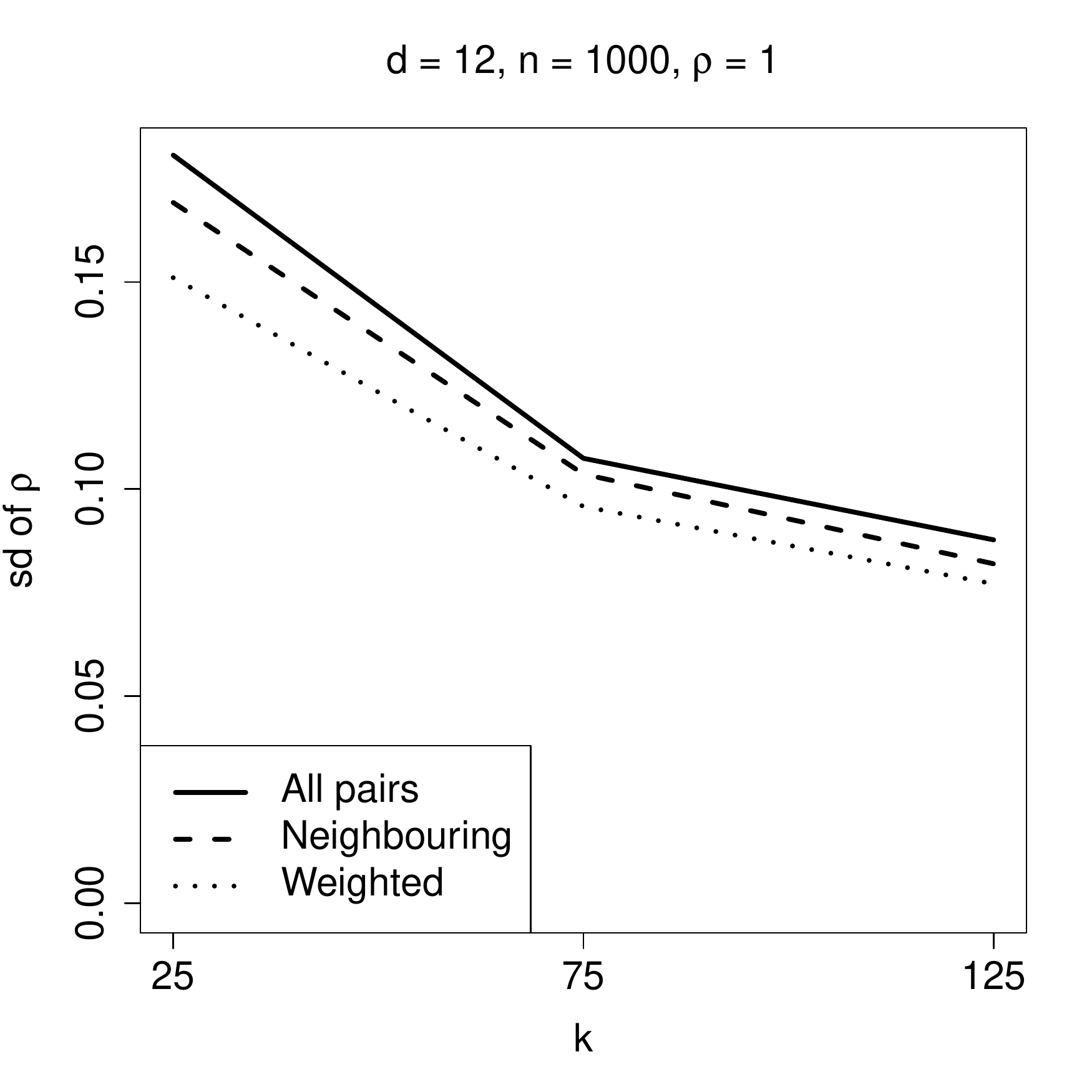}}
\subfloat{\includegraphics[width=0.3\textwidth]{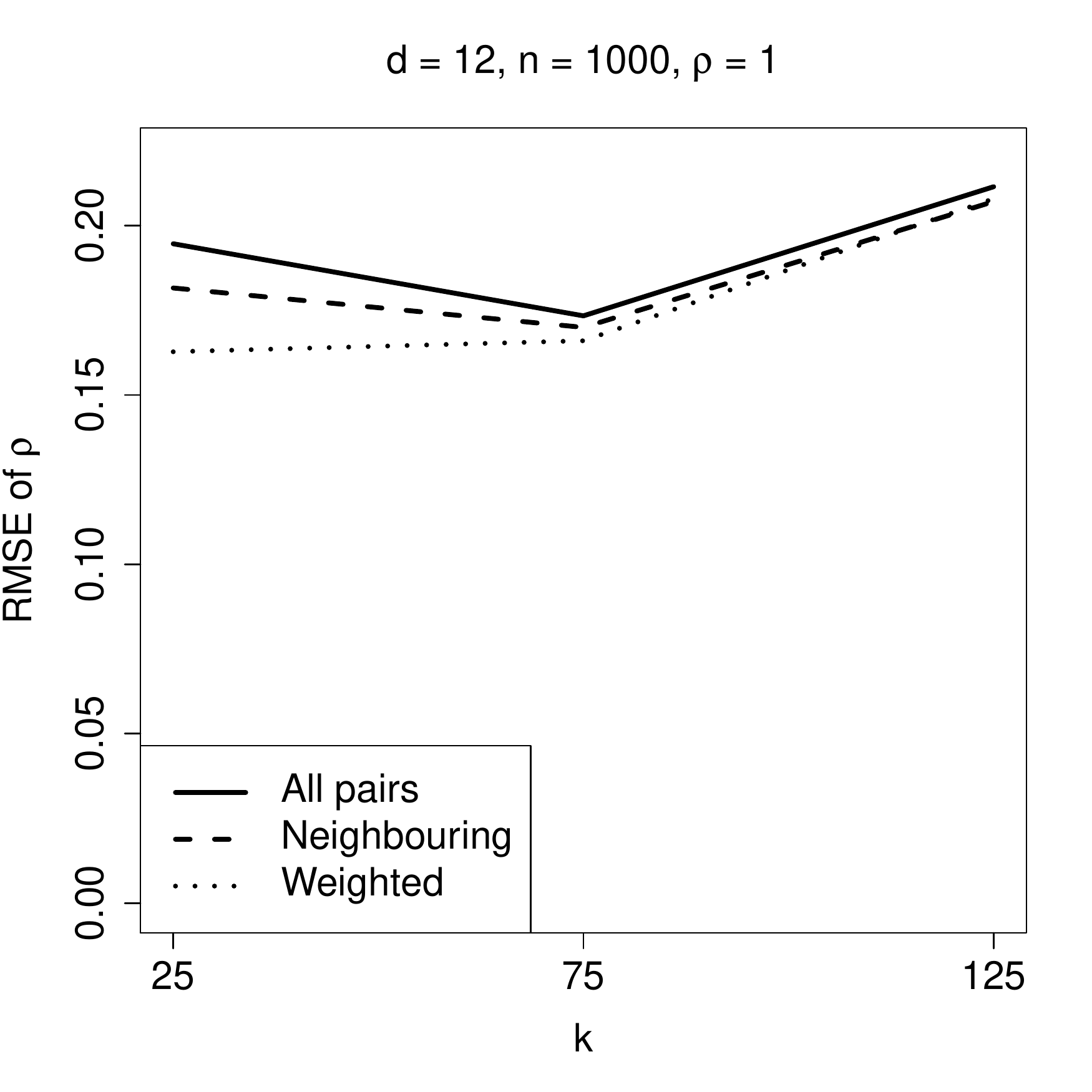}} \\
\subfloat{\includegraphics[width=0.3\textwidth]{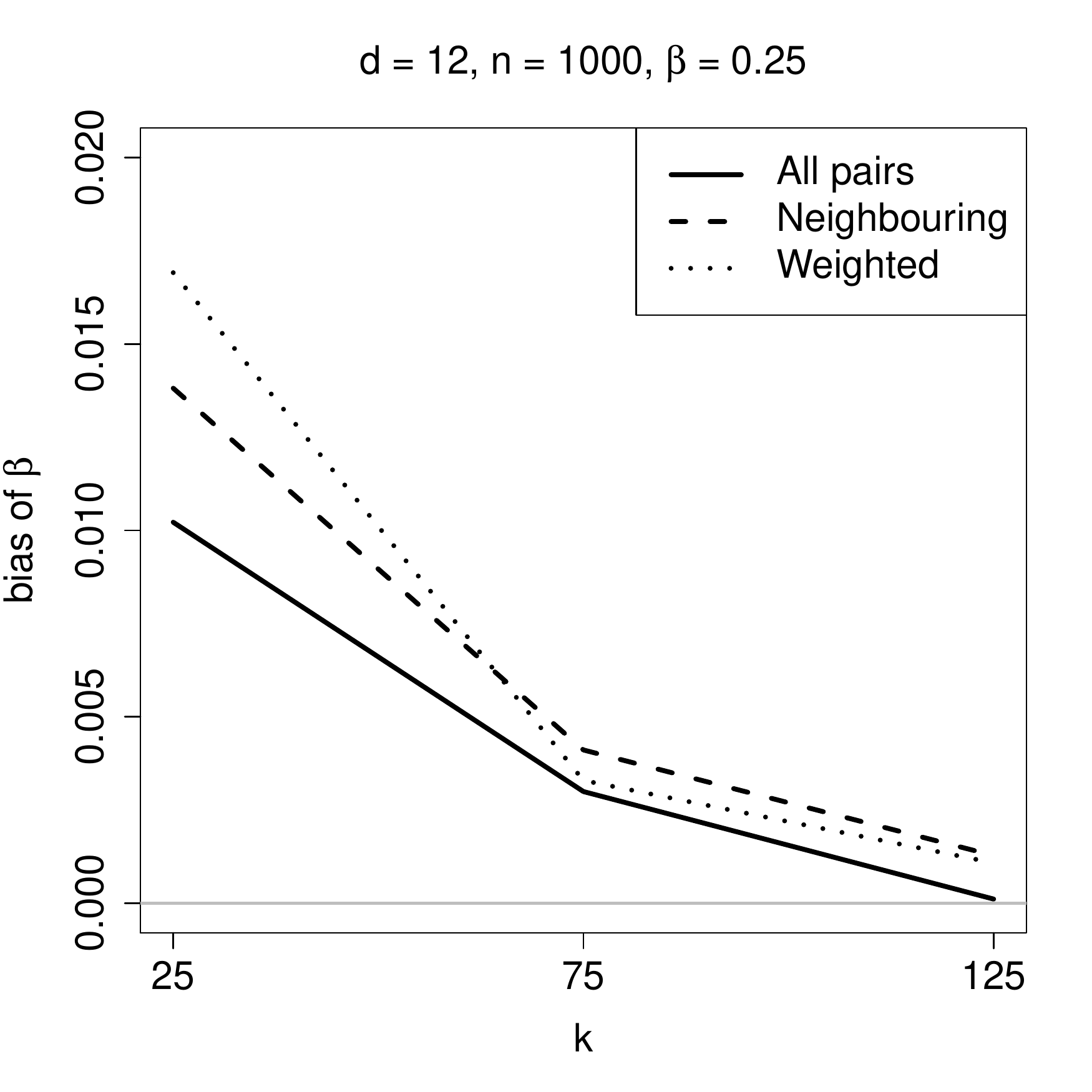}}
\subfloat{\includegraphics[width=0.3\textwidth]{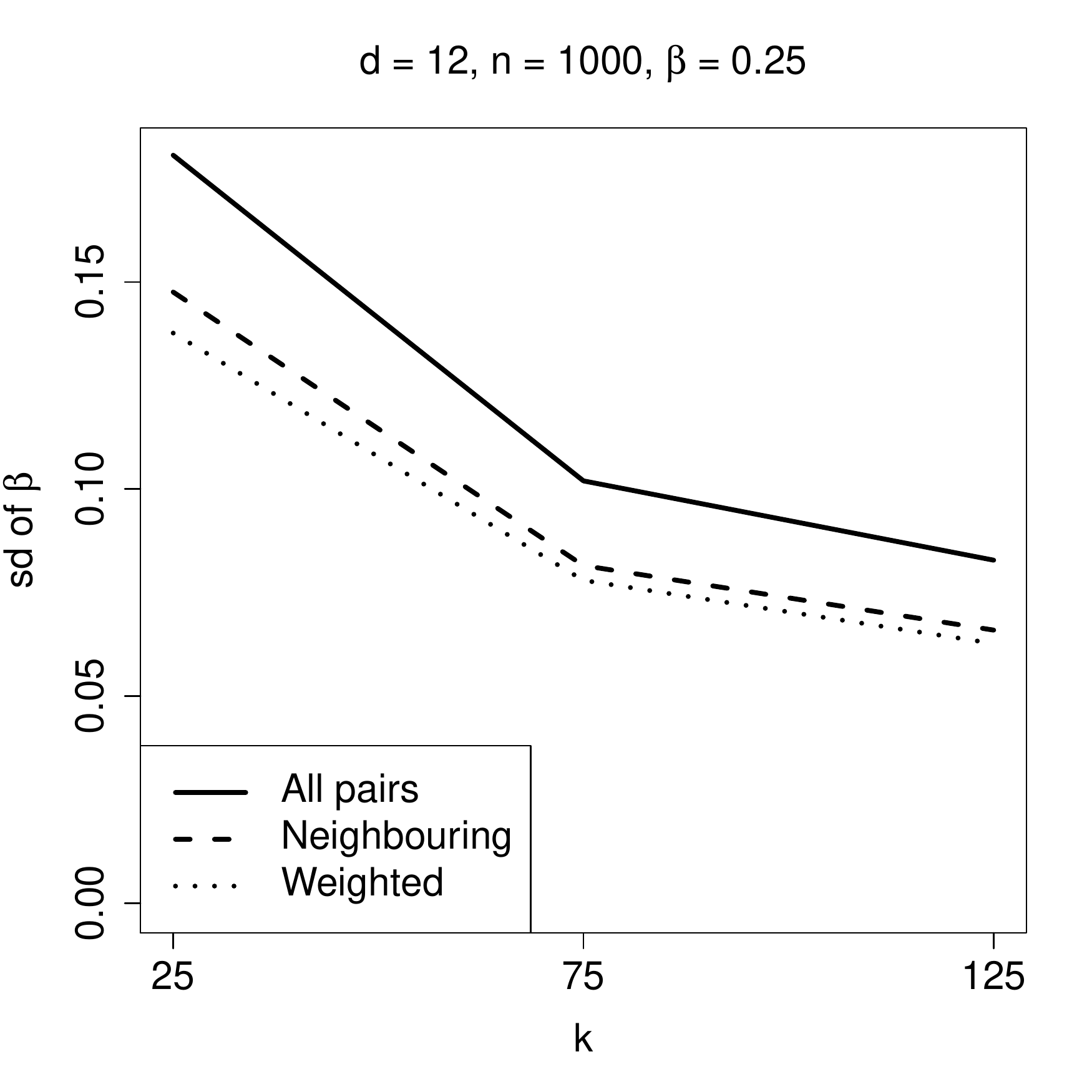}}
\subfloat{\includegraphics[width=0.3\textwidth]{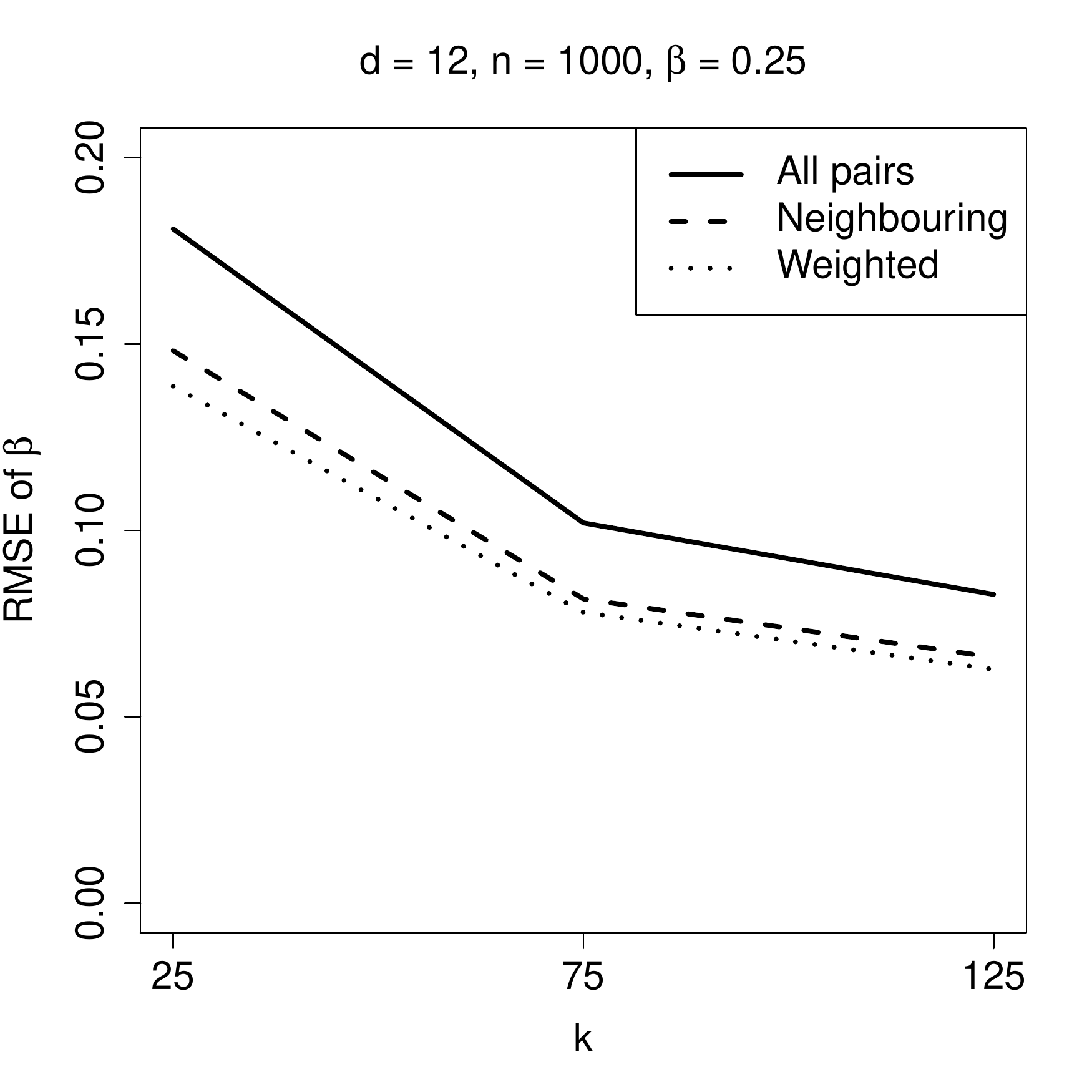}} \\
\subfloat{\includegraphics[width=0.3\textwidth]{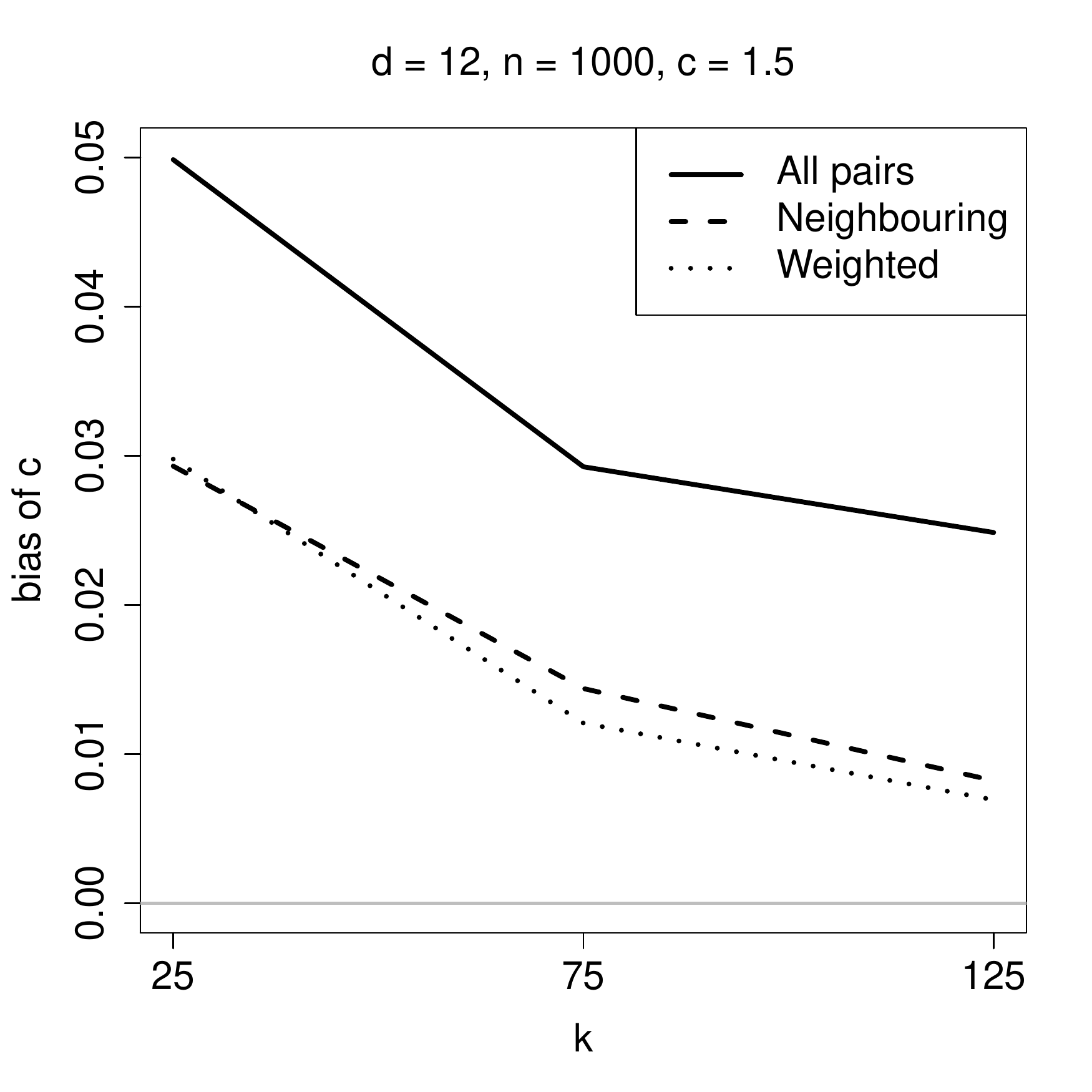}}
\subfloat{\includegraphics[width=0.3\textwidth]{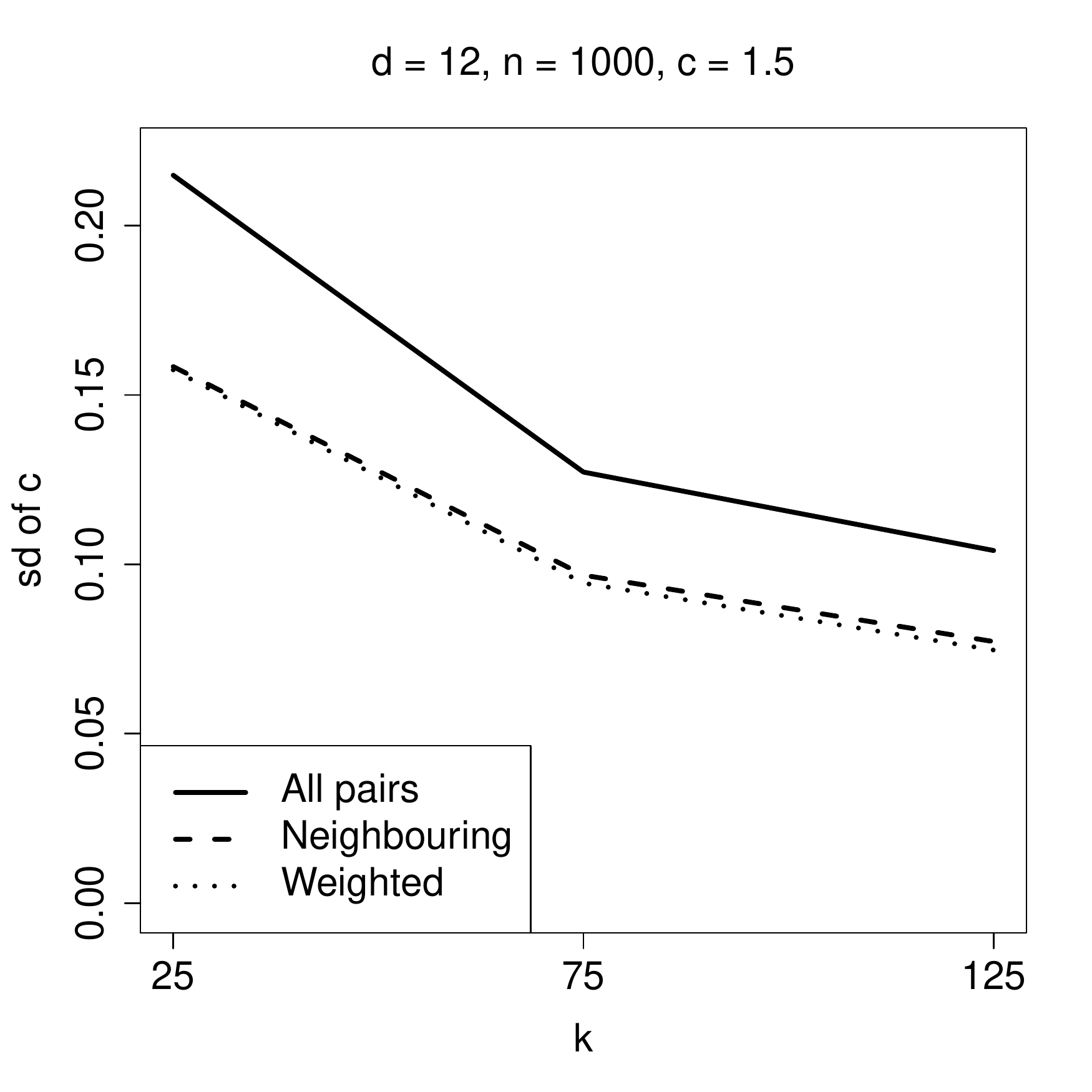}}
\subfloat{\includegraphics[width=0.3\textwidth]{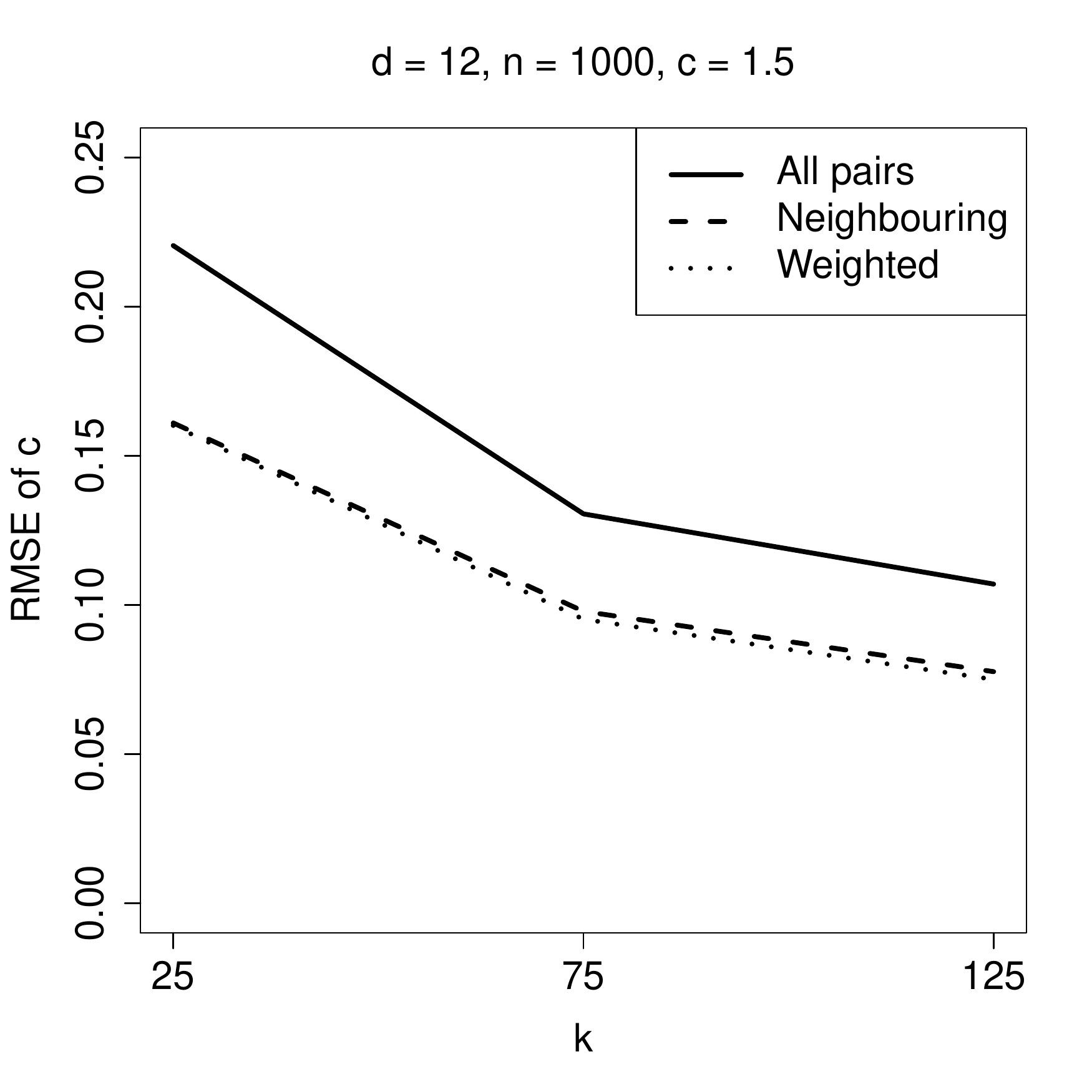}}
\caption{Bias, standard deviation and RMSE for estimators of $\alpha = 1.5$ (top), $\rho=1$ (top middle), $\beta = 0.25$ (bottom middle) and $c = 1.5$ (bottom) for the perturbed 12-dimensional Brown--Resnick process; 500 samples of $n=1000$.}
\label{fig:brerr2}
\end{figure}

\paragraph{Comparison with \citet{engelke2014}.\\}
To compare the pairwise M-estimator with the one from \citet{engelke2014}, we consider the setting used in the simulation study of the latter paper: we simulate 500 samples of size $8000$ of the \textit{univariate} Brown--Resnick process on an equidistant grid on the interval $[0,3]$  with step size $0.1$. The parameters of the model are $(\alpha, \rho)=(1,1)$. We estimate the unknown parameters for $k=500$ and 140 pairs, so that the locations of the selected pairs are at most a distance $0.5$ apart. We use the identity weight matrix, since in this particular setting the weight matrix is very large and, as far as we could tell from some preliminary experiments, it leads to only a small reduction in estimation error. Asymptotically we see a reduction of the standard deviations
of about $ 13 \%$ for $\alpha$ and $3 \%$ for $\rho$.
In Figure~\ref{fig:boxplots} below, the results are presented in the form of boxplots, to facilitate comparison with Figure~$2$ in \citet{engelke2014}.
Our procedure turns out to perform equally well for the estimation of $\alpha$ and only slightly worse when estimating $\rho$.
It is to be kept in mind that, whereas the method in \citet{engelke2014} is tailor-made for the Brown--Resnick process, our method is designed to work for general parametric models.

\begin{figure}[ht]
\centering
\subfloat{\includegraphics[width=0.3\textwidth]{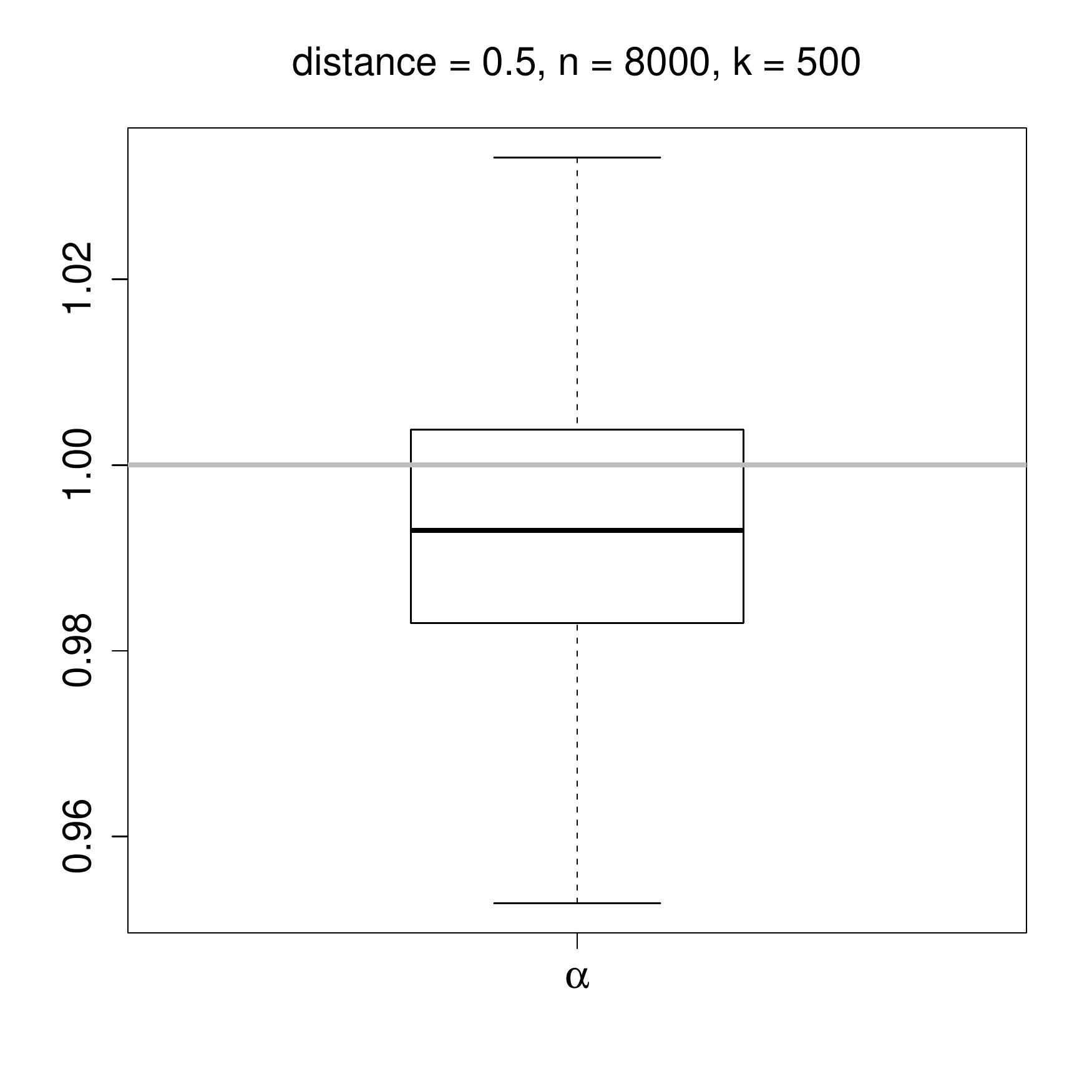}}
\subfloat{\includegraphics[width=0.3\textwidth]{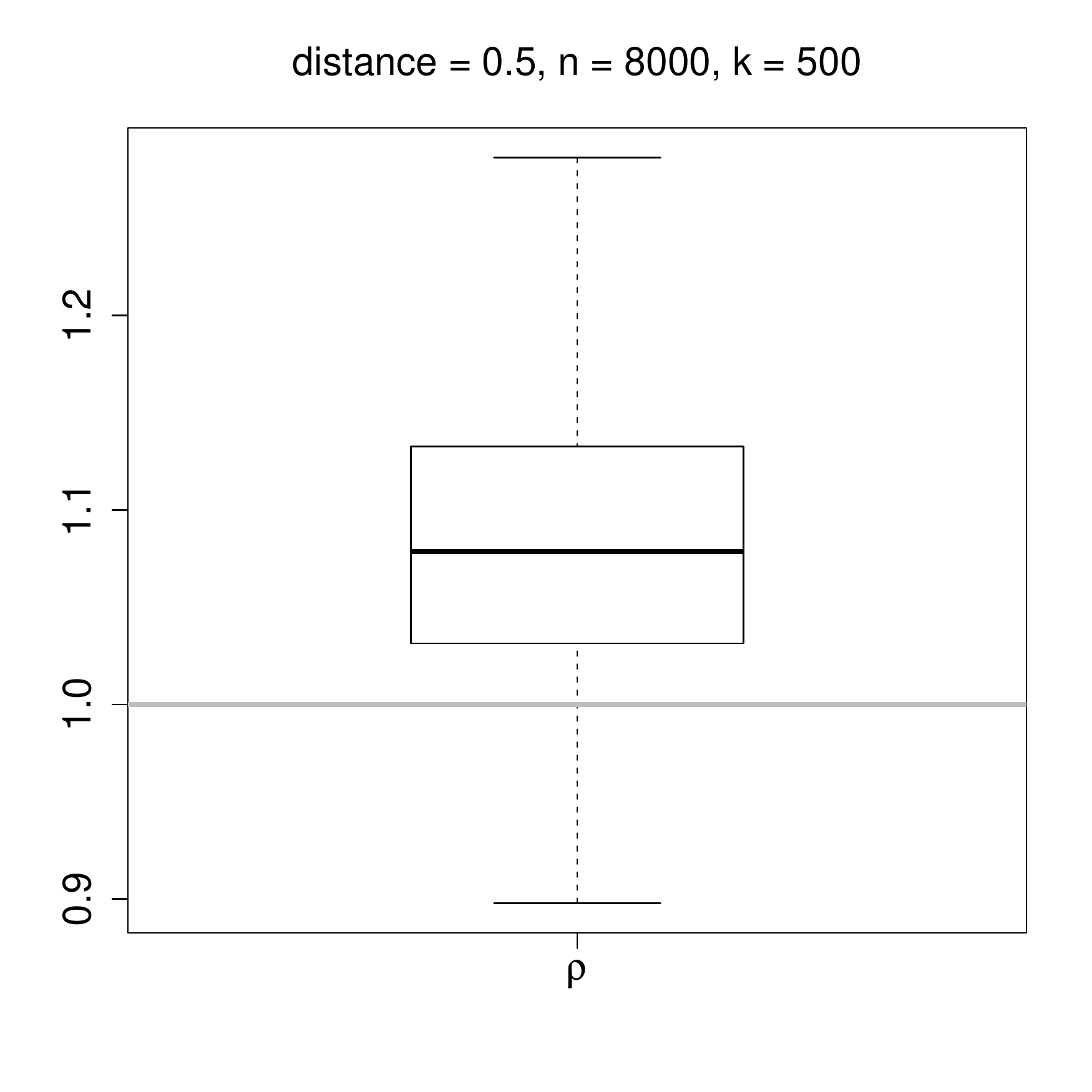}}
\caption{Boxplots of estimators of $\alpha = 1$ (left) and $\rho=1$ (right) for a univariate Brown--Resnick process on the interval $[0,3]$ with $d=31$; 500 replications of $n = 8000$, $k=500$.}
\label{fig:boxplots}
\end{figure}

\paragraph{Discussion.\\}
We have seen in the 100-dimensional simulation study that the computation of the unweighted pairwise estimator is fast even for a large number of pairs. However, calculating an entry for the optimal weight matrix takes about 15 seconds on a standard computer. Since we have to calculate $q(q+1)/2$ entries of the weight matrix, this method gets more time-consuming when the number of pairs $q$ is large.

We also noticed that for large dimensions, a relatively small sample size of $n=500$ is sufficient to obtain good results.
However, the smaller the dimension, the larger the sample size needs to be, i.e., a decrease of information in space must be compensated by an increase of information in ``time''.
We have observed that the choice of the starting value hardly affects the outcome of the optimisation procedure, unless the dimension is less than five.
More guidelines and rules-of-thumb for practical use of the estimator can be found in the reference manual of the \textsf{spatialTailDep} package.

Another interesting feature is that, for both the Smith model and the Brown--Resnick process, considering only neighbouring pairs leads to better results than considering all pairs. As the distance between two locations increases, they become tail independent, so that including pairs of distant locations adds little information about the model parameters.

Finally, to assess the quality of the normal approximation to the sampling distribution of the estimator, we have conducted simulation experiments for the Smith model. For sample sizes $n=5000$ and $n= 10\,000$, multivariate normality was not rejected for any of the values of $k$ we considered.

\section{Efficiency comparisons}\label{sec:EKS12comp}

\subsection{Finite-sample comparisons}

A natural question that arises is whether the quality of estimation decreases when making
the step from the $d$-dimensional estimator $\widehat{\theta}'_n$ in \eqref{eq:theta}
to the pairwise estimator $\widehat{\theta}_n$ in \eqref{eq:thetafinal}. We
will demonstrate for the multivariate logistic model and the Smith model that this is not the case, necessarily in a dimension where $\widehat{\theta}'_n$ can be computed. The $d$-dimensional logistic model has stable tail dependence function
\begin{equation}\label{eq:logistic}
\ell(\vc{x}; \theta) = \left( x_1^{1/\theta} + \cdots + x_d^{1/\theta} \right)^\theta, \qquad \theta \in [0,1].
\end{equation}
We simulate 200 samples of size $n=1500$ from the logistic model in dimension $d=5$ with parameter value $\theta_0=0.5$ and we assess the quality of our estimates via the bias and root mean squared error (RMSE) for $k \in \{40,80,\ldots,320\}$. The top panels of Figure~\ref{fig:smith} show the bias and RMSE for the M-estimator of \citet{einmahl2012} with the function $g \equiv 1$ (dashed lines). 
\begin{figure}[ht]
\centering
\subfloat{\includegraphics[width=0.3\textwidth]{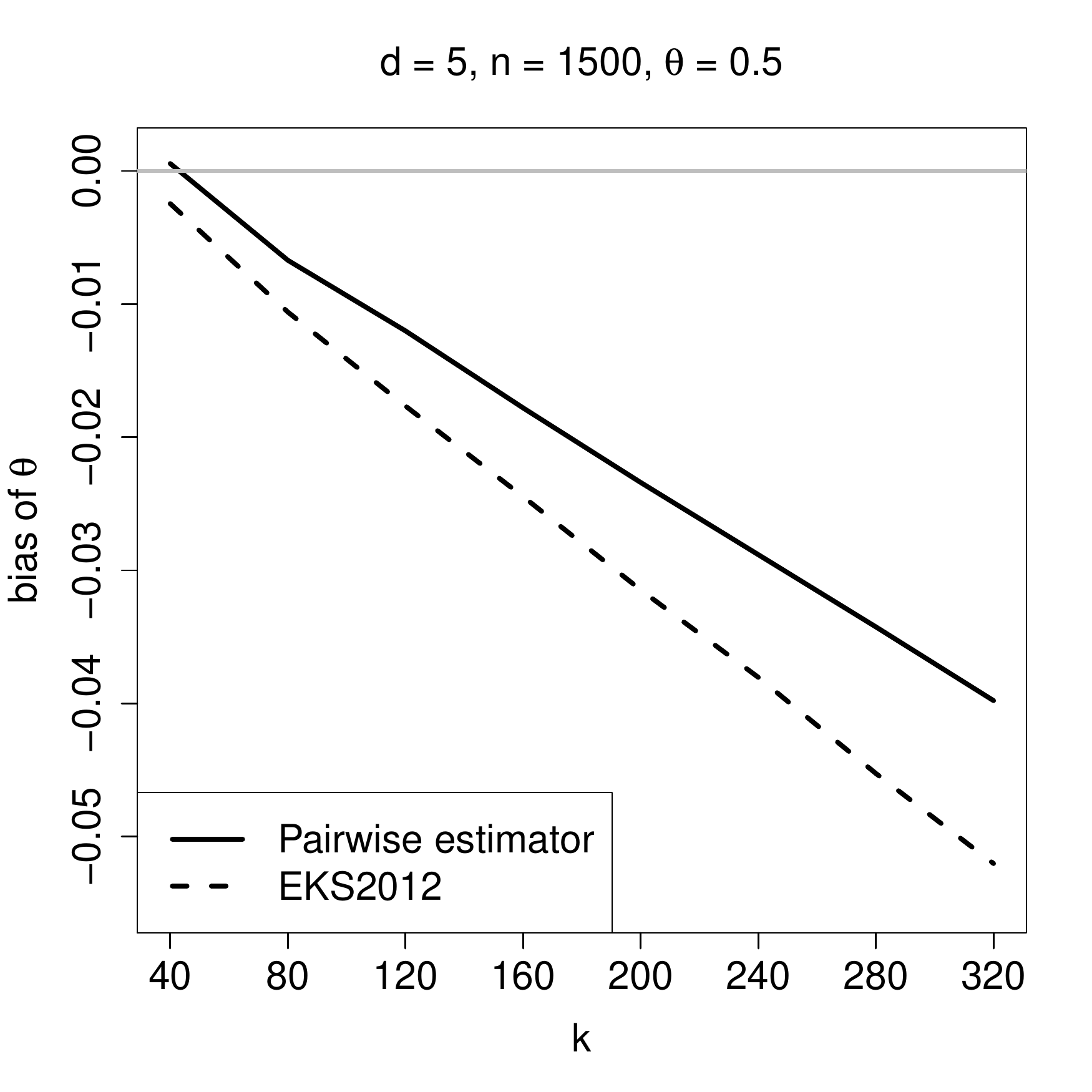}}
\subfloat{\includegraphics[width=0.3\textwidth]{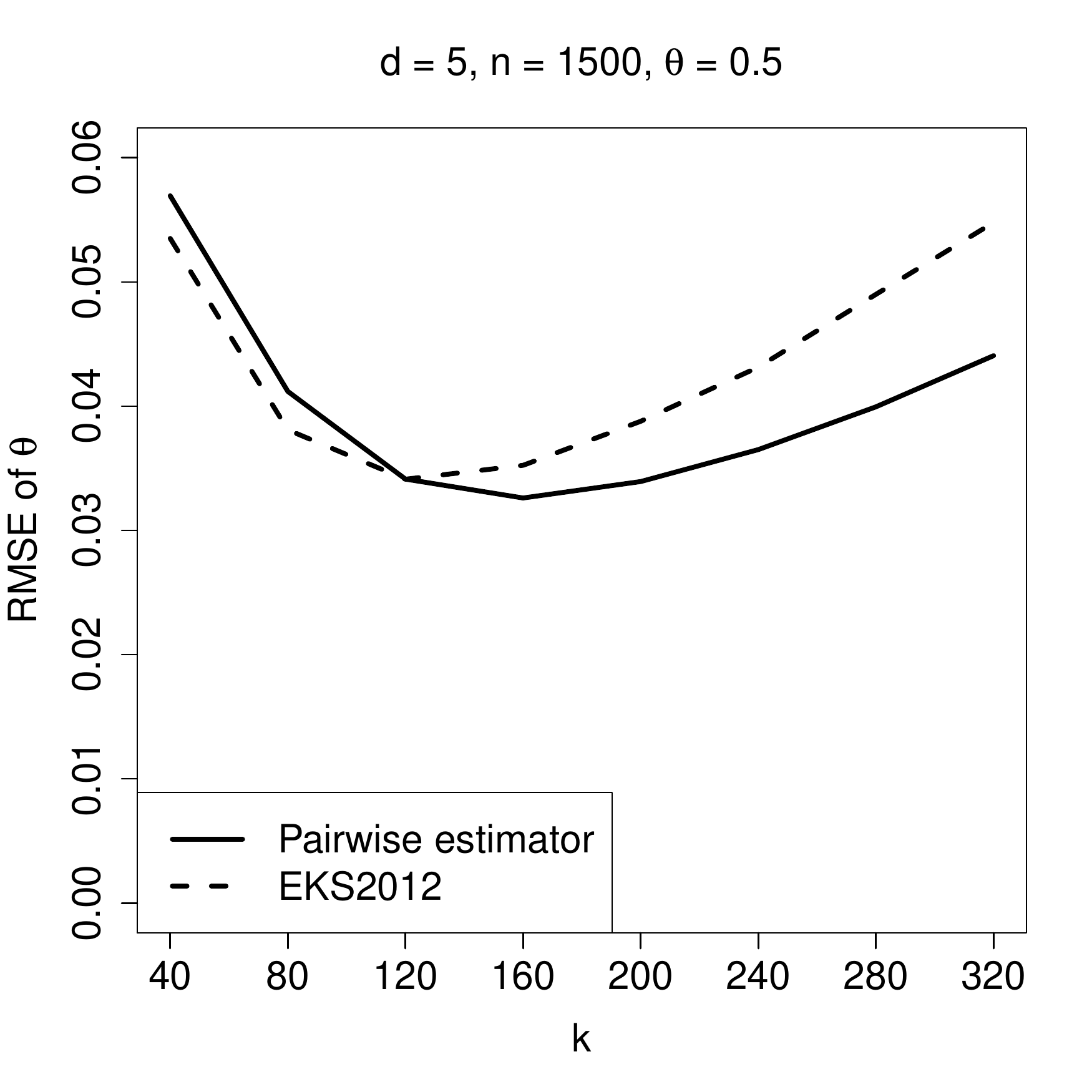}} \\
\subfloat{\includegraphics[width=0.3\textwidth]{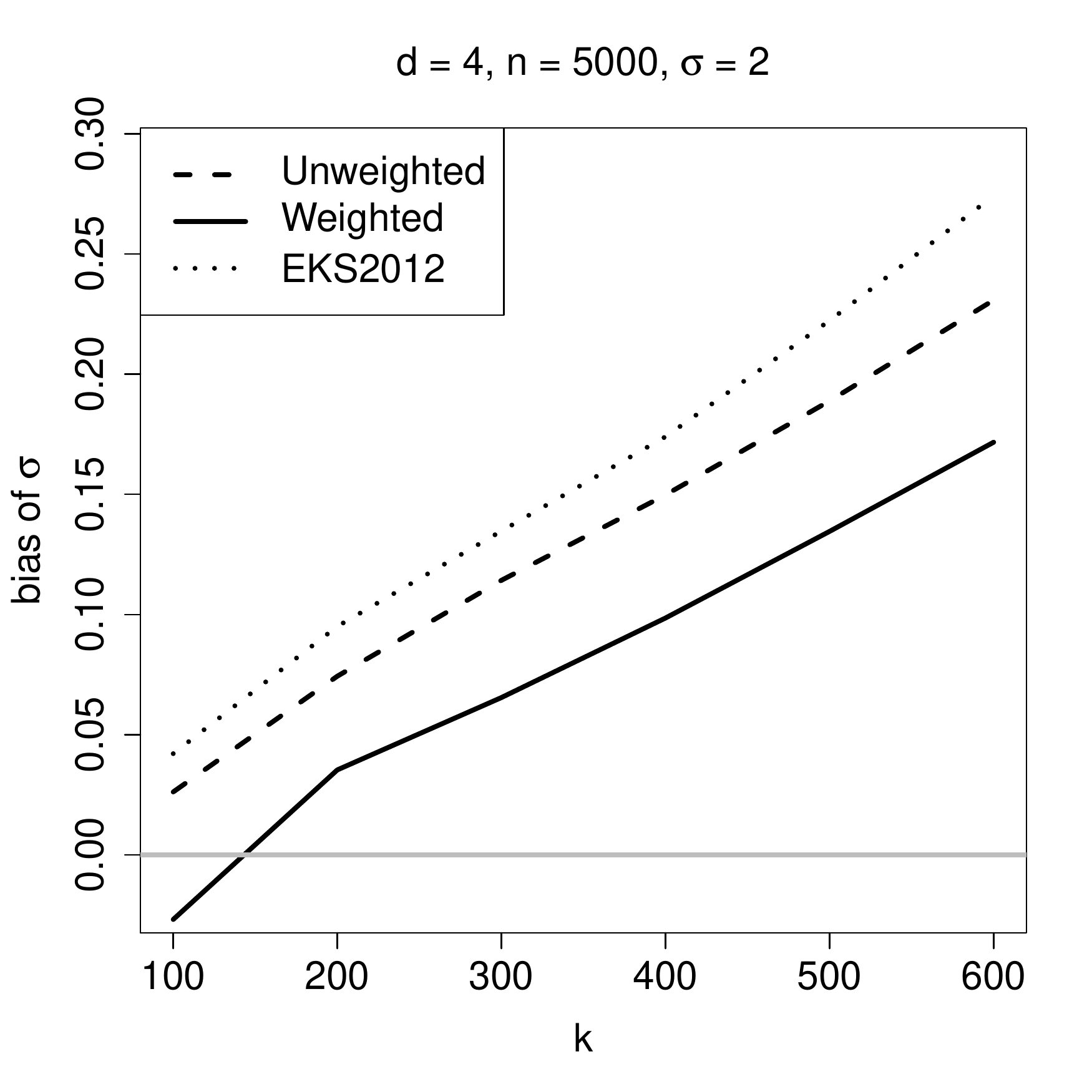}}
\subfloat{\includegraphics[width=0.3\textwidth]{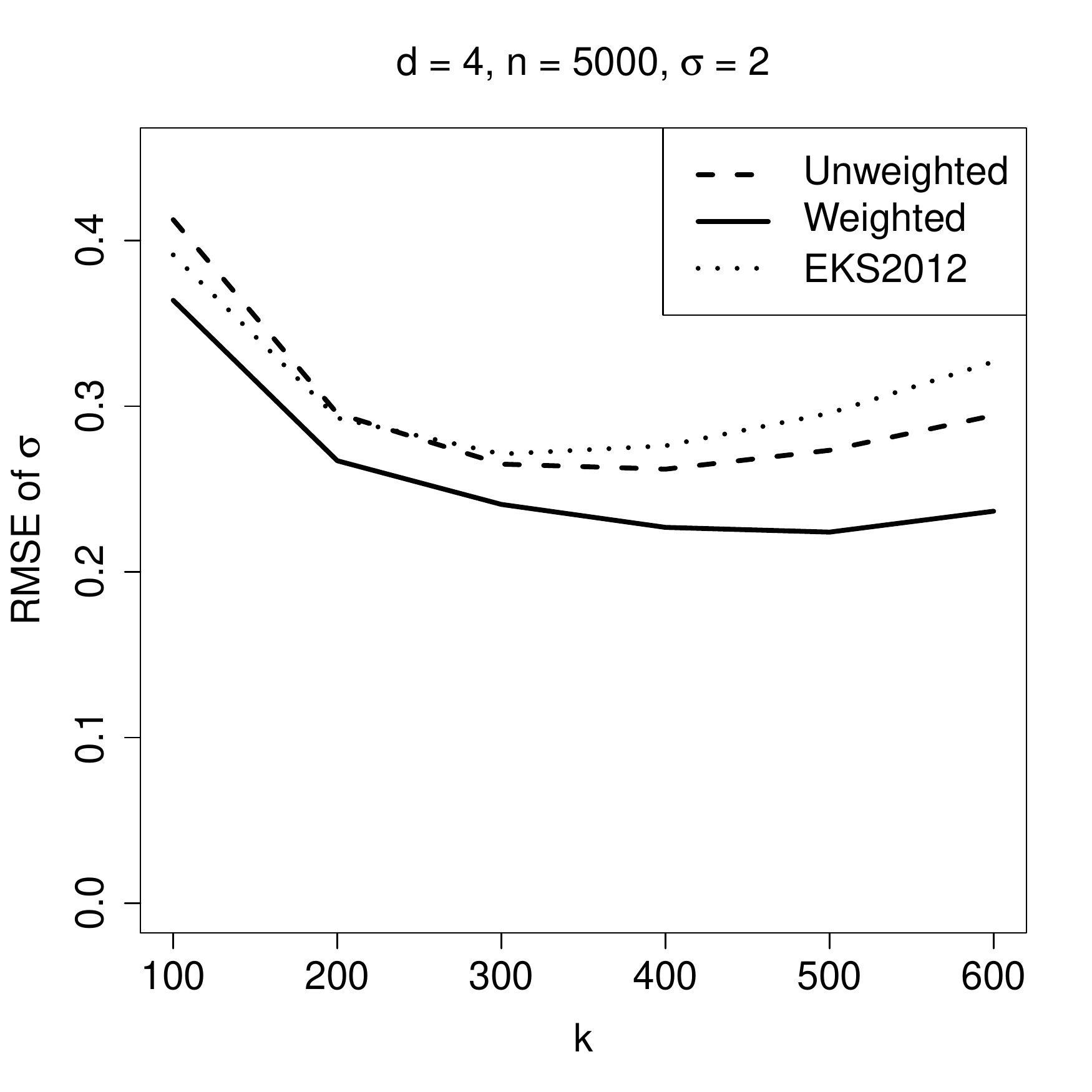}}
\caption{Top: bias and RMSE for estimators of $\theta_0 = 0.5$ for the logistic model; the pairwise estimator and the $5$-dimensional M-estimator from \citet{einmahl2012}; 200 samples of $n=1500$. Bottom: bias and RMSE for estimators of $\sigma = 2$ for the Smith model; the pairwise estimator with identity weight matrix (unweighted), the pairwise estimator with optimal weight matrix (weighted) and the $4$-dimensional M-estimator from \citet{einmahl2012}; 200 samples of $n=5000$.}
\label{fig:smith}
\end{figure}
The results are the same as in \citet[Figure 1]{einmahl2012}. The solid lines show the bias and RMSE for the pairwise M-estimator with $q=10$ and $\widehat{\Omega}_n = I_q$. We see that the pairwise estimator performs somewhat better in terms of bias and also has the lower minimal RMSE, for $k=160$. Note that we only show results for the pairwise estimator with identity weight matrix since using the optimal weight matrix has no effect on the estimator.

Next, consider the Smith model with $d=4$ locations on an equally spaced unit distance $2 \times 2$ grid. We simulate 200 samples of size $n=5000$ from an isotropic Smith model with parameter value $\theta_0 = \sigma = 2$, i.e., $\Sigma = \sigma I_2$.
The bottom panels of Figure~\ref{fig:smith} show the bias and RMSE for $k \in \{100,\ldots,600\}$ for the four-dimensional M-estimator with $q=5$ weight functions, given by $g_m (\vc{x}) = x_m$ for $m=1,\ldots,4$ and $g_m \equiv 1$ for $m=5$,
the pairwise M-estimator with identity weight matrix, and the pairwise M-estimator with optimal weight matrix. We see clearly that the pairwise weighted method is the best one in terms of both bias and RMSE.

\subsection{Asymptotic variances}

Another question is whether the asymptotic variance increases when switching to the pairwise estimator. First, we consider the Smith model on the line with $d$ equidistant locations, i.e.,
\begin{equation*}
a^2_{uv} = \frac{(s_u - s_v)^2}{\sigma}, \qquad s_u, s_v \in \{1,\ldots,d\}.
\end{equation*} 
The left and middle panels of Figure~\ref{fig:smith2} show values for the asymptotic variances of a number of estimators when $\sigma \in \{0.5,1,1.5,2\}$ and $d \in \{4,6\}$. For the $d$-dimensional estimator $\widehat{\theta}'_n$, we used $g \equiv 1$ as before, and thus $q = 1$; the formula for the asymptotic variance is given in (4.6) in \citet{einmahl2012}. For the pairwise estimator, we computed the asymptotic variance in \eqref{eq:asym} in two cases: first, neighbouring pairs only and identity weight matrix, and second, all pairs and the optimal weight matrix. Throughout, both pairwise estimators have a slightly lower asymptotic variance than the $d$-dimensional estimator.

When the dimension, $d$, is large, say 100, the method from \citet{einmahl2012} involves intractable, high-dimensional integrals. For the sake of comparison, we construct a computationally tractable variant of the logistic model that mimics the property of the Smith model that tail dependence vanishes as the distance between locations increases. 

Consider $d$ locations in $r$ ``regions", every region containing $d/r$ locations. Within all regions, assume a logistic stable tail dependence function as in \eqref{eq:logistic}, with a common value of $\theta_0 \in [0,1]$ for all regions; locations in different regions are assumed to be tail independent. The right panel of Figure~\ref{fig:smith2} shows the asymptotic variances of a number of estimators for $\theta_0 \in \{0.1,0.2,\ldots,0.9\}$, $d=100$, and $r = 20$. For the $d$-dimensional estimator, we used again $g \equiv 1$ and $q=1$. For the pairwise estimator, we used all 10 pairs in each of the 20 regions, yielding $q = 200$ pairs in total; because of symmetry, the optimal weight matrix produces the same asymptotic variance as the identity weight matrix. For most of the parameter values, using the pairwise estimator entails only a modest increase in asymptotic variance. For some parameter values, it even leads to a small decrease.
\begin{figure}[ht]
\centering
\subfloat{\includegraphics[width=0.3\textwidth]{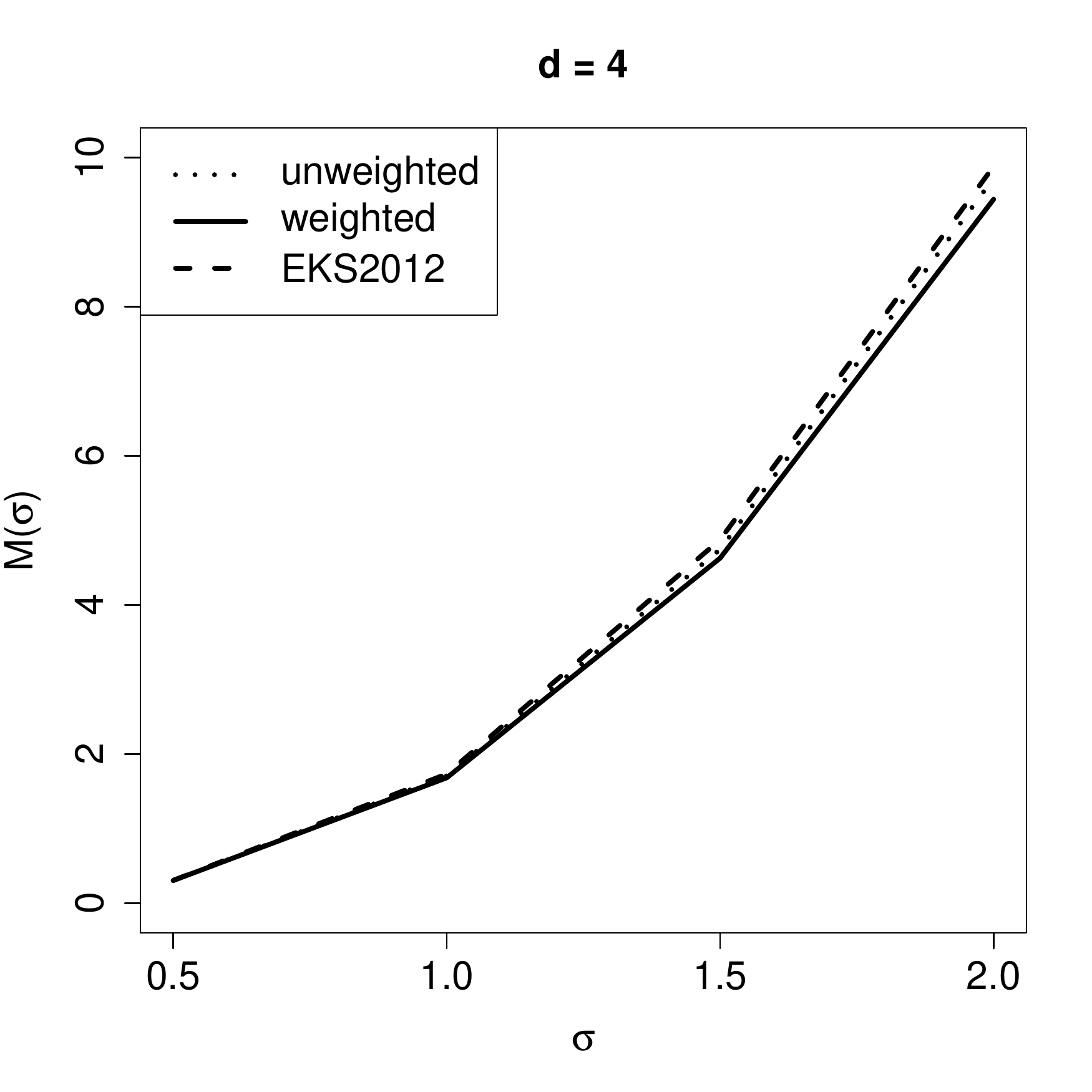}}
\subfloat{\includegraphics[width=0.3\textwidth]{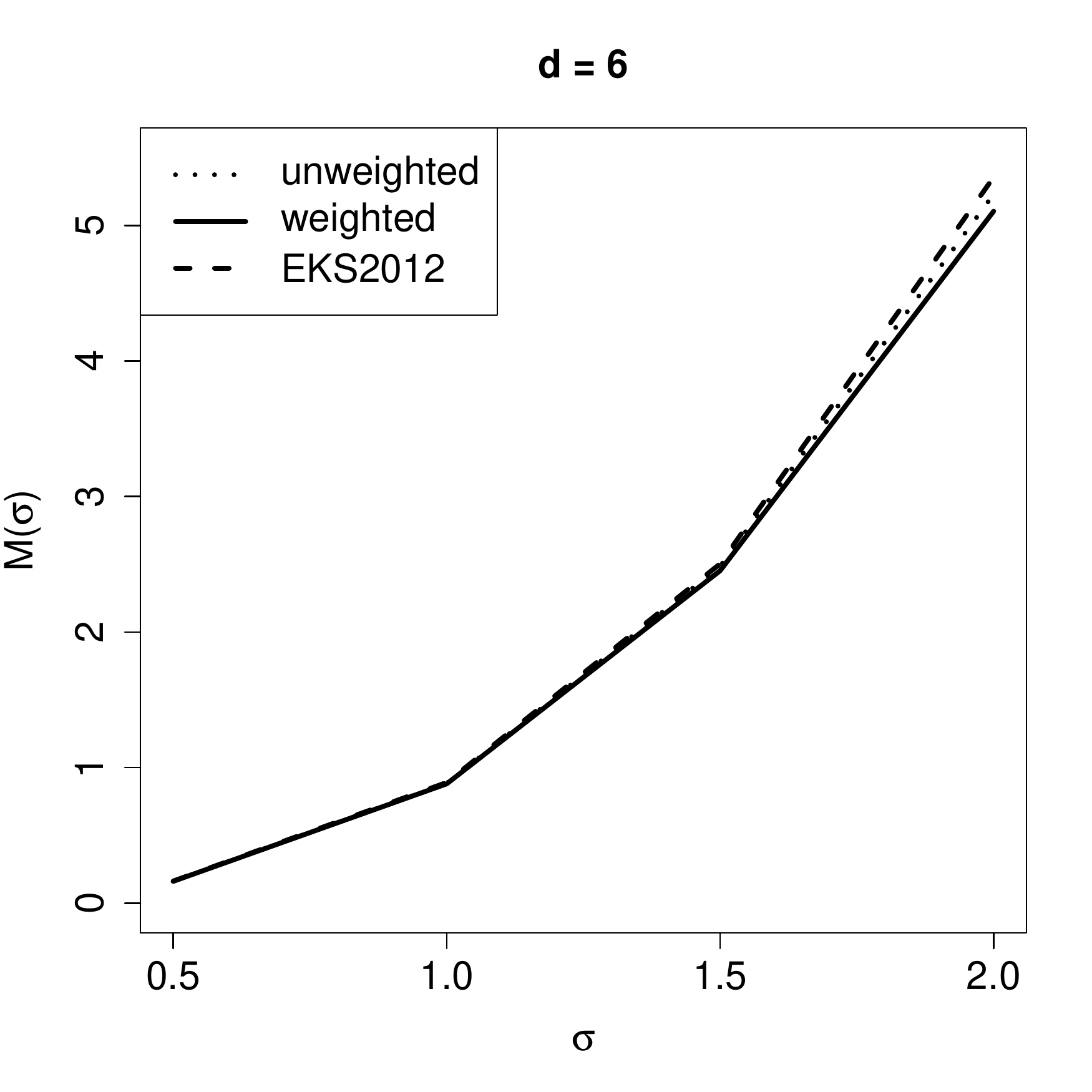}}
\subfloat{\includegraphics[width=0.3\textwidth]{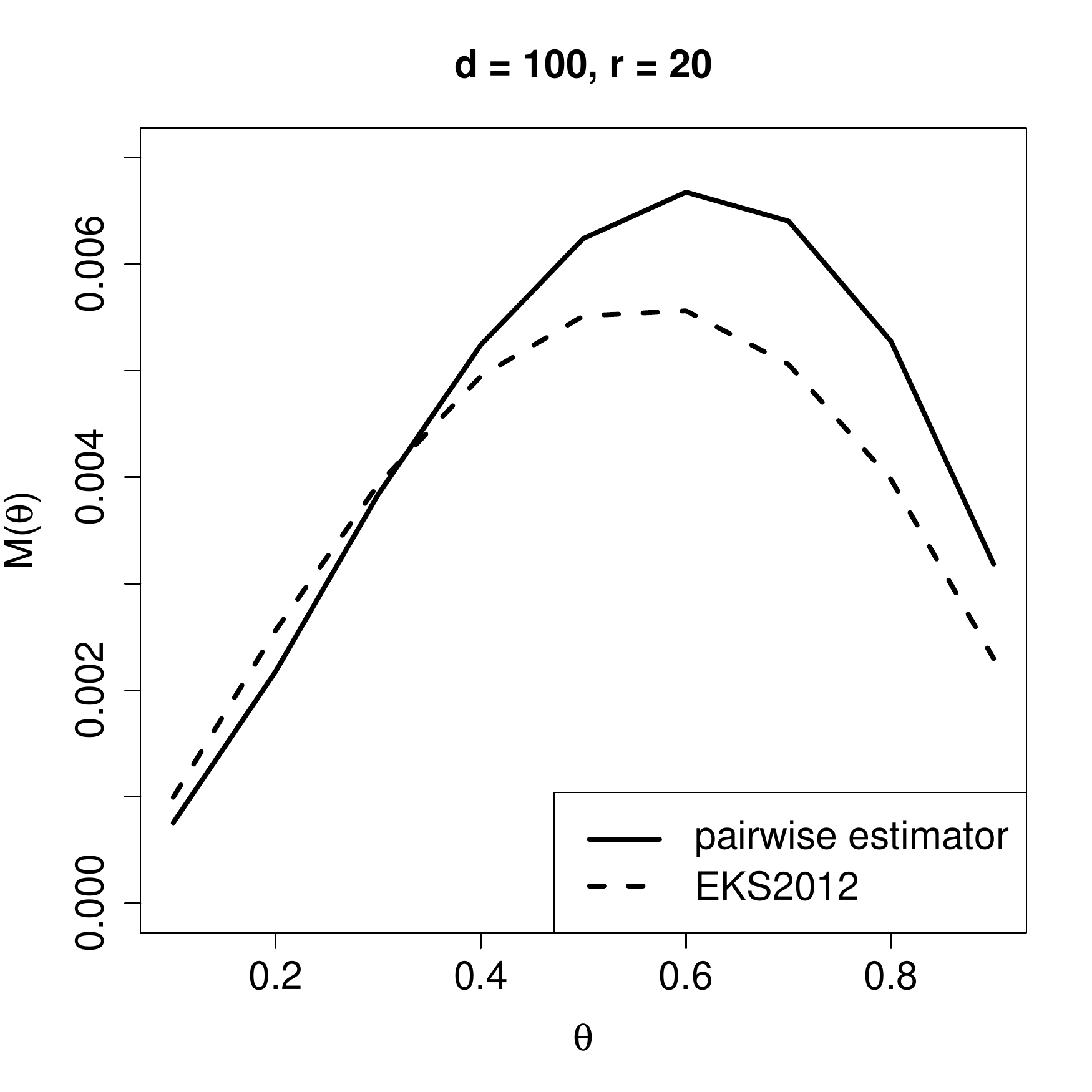}}
\caption{Left and middle: asymptotic variance $M(\sigma)$ for $\sigma \in \{0.5,1,1.5,2 \}$ and $d \in \{4,6\}$ for the $d$-dimensional Smith model on the line; the pairwise estimator with identity weight matrix (unweighted), the pairwise estimator with optimal weight matrix (weighted) and the $d$-dimensional M-estimator from \citet{einmahl2012}. Right: asymptotic variance $M(\theta_0)$ for $\theta_0 \in \{0.1,\ldots,0.9 \}$, $d = 100$ and $r = 20$ for the logistic model; the pairwise estimator with identity weight matrix and the $d$-dimensional M-estimator from \citet{einmahl2012}.}
\label{fig:smith2}
\end{figure}

\section{Application: speeds of wind gusts}
\label{sec:applic}

Using  extreme-value theory to estimate the frequency and magnitude of extreme wind events or to estimate the return levels for (extremely) long return periods is not a novelty in the fields of meteorology and climatology.
Numerous research papers published in the last 20--25 years are applying methods from extreme-value theory to treat those estimation problems, see, for example, \citet{karpa2013,ceppi2008,palutikof1999} and the references therein.
However, until very recently, all statistical approaches were univariate.
The scientific and computational advancements nowadays facilitate the usage of high-dimensional or spatial models.
In \citet{engelke2014} and \citet{oesting2013}, for instance, Brown--Resnick processes are used to model wind speed data.

We consider a data set from the Royal Netherlands Meteorological Institute (KNMI), consisting of the daily maximal speeds of wind gusts, which are measured in 0.1 m/s. The data are observed at 35 weather stations in the Netherlands, over the time period from January 1, 1990 to May 16, 2012. The data set is freely available from \url{http://www.knmi.nl/climatology/daily_data/selection.cgi}.
Due to the strong influence of the sea on the wind speeds in the coastal area, we only consider the inland stations, of which we removed three stations with more than 1000 missing observations. The thus obtained 22 stations and the remaining amount of missing data per station are shown in the left panel of Figure \ref{Fig:AllTogetherNow}.
 We aggregate the daily maxima to three-day maxima in order to minimize temporal dependence and we also restrict our observation period to the summer season (June, July and August) to obtain more or less equally distributed data. To treat the missing data, if at least one of the observations for the three-day maximum is present, we define this to be a valid three-day maximum, thus ignoring these missing observations. We consider a three-day maximum missing only if all three constituting daily maxima are missing. In this way only a few data are missing. We use the ``complete deletion approach" for these data and obtain a data set with $n=672$ observations. This data set is available from the \textsf{spatialTailDep} package.

We consider the stable tail dependence function corresponding to the Brown--Resnick process (see Section~\ref{sec:theory}). It is frequently argued, see e.g. \citet{engelke2014} or \citet{ribatet2013}, that an anisotropic model is needed to describe the spatial tail dependence of wind speeds.
Using Corollary~\ref{cor3} we first test, based on the $q = 29$ pairs of stations that are at most 50 kilometers apart,
if the isotropic process suffices for the above data. In the reparametrization introduced in Section \ref{applic}, the case  $\tau_{11}=\tau_{22}$ and $\tau_{12}=0$ corresponds to isotropy. The test statistic
    \begin{equation*}
	k\left(\widehat\tau_{11}-\widehat\tau_{22},\widehat\tau_{12}\right) M_2\left(\widehat{\alpha},\widehat\tau_{11}+\widehat\tau_{22},0,0\right)^{-1}\left(\widehat\tau_{11}-\widehat\tau_{22},\widehat\tau_{12}\right)^T
    \end{equation*}
is computed for $k=60$. We obtain a value of $0.180$, leading to a $p$-value of $0.914$ against the $\chi_2^2$-distribution (Corollary~\ref{cor3}), so we can not reject the null hypothesis. Although the stable tail dependence function corresponding to the more complicated anisotropic Brown--Resnick process is usually assumed for this type of data, the test result shows that the more simple isotropic Brown--Resnick process suffices for the Dutch inland summer season wind speeds.

The estimate of the parameter vector $(\alpha, \rho)$ corresponding to the isotropic Brown--Resnick process is obtained for $k=60$, with $q=29$ pairs and using the optimal weight matrix chosen according to the two-step procedure described after Corollary~\ref{cor1}. The estimates, with standard errors in parentheses, are $\widehat\alpha = 0.398$ $(0.020)$ and $\widehat\rho = 0.372$ $(0.810)$. We also see that the Smith model would not fit these data well since $\alpha$ is much smaller than $2$.

\begin{figure}[ht]
\centering
    \subfloat{\includegraphics[angle=0,width=.385\textwidth]{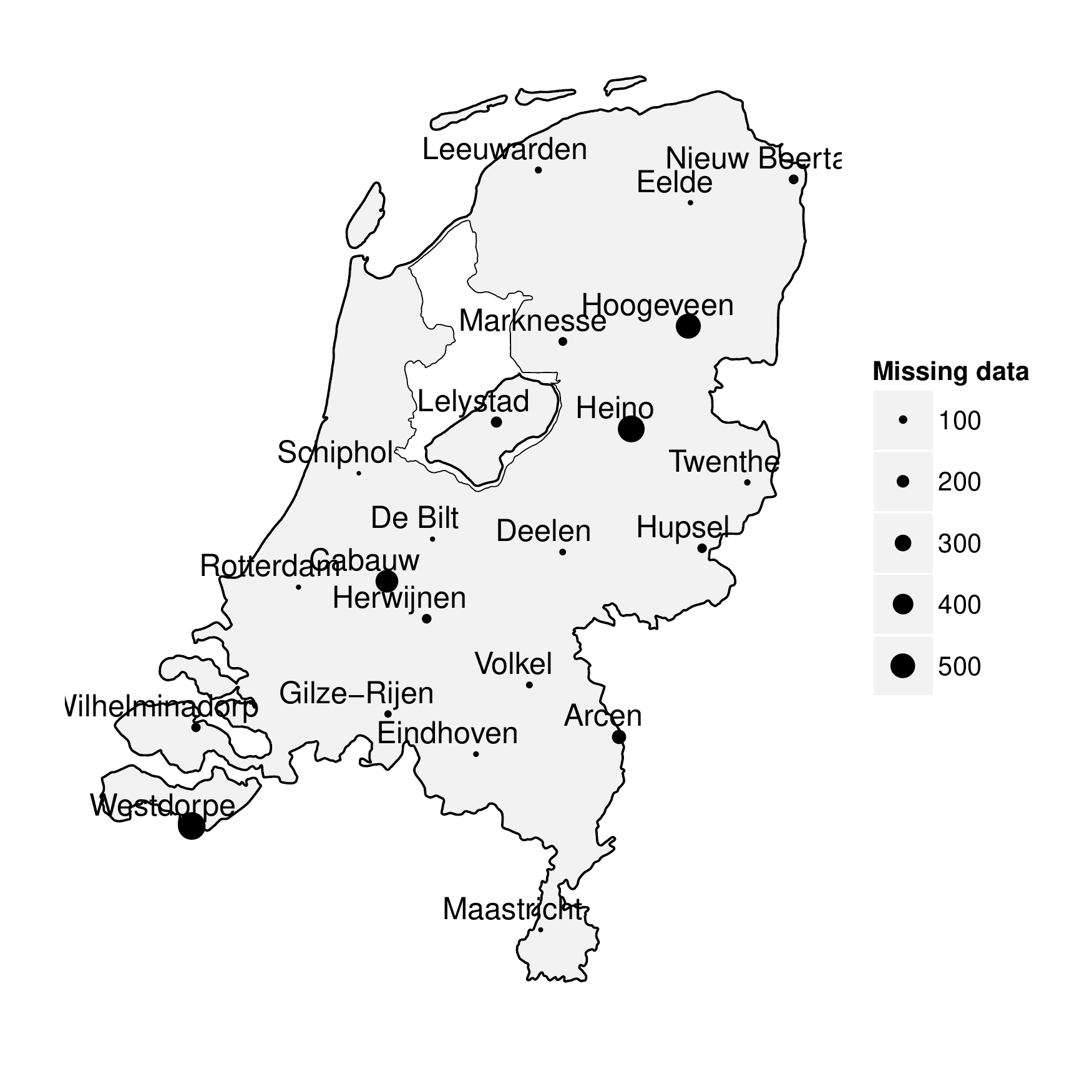}}
    \subfloat{\includegraphics[angle=0,width=.385\textwidth]{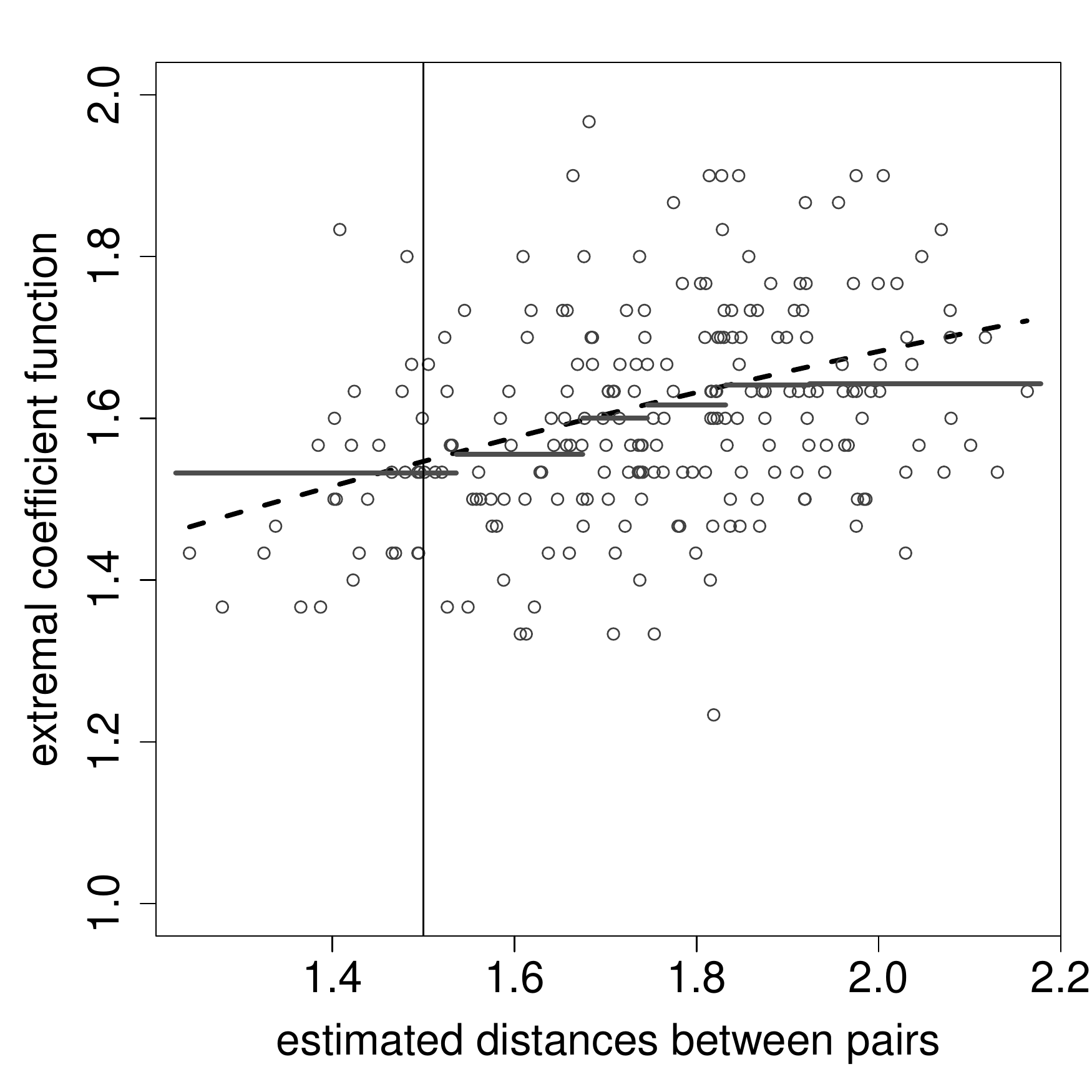}}
    \caption{KNMI weather stations (left). Estimates of the extremal coefficient function (right).}
    \label{Fig:AllTogetherNow}
\end{figure}

To visually assess the goodness-of-fit, we compare the nonparametric and the Brown--Resnick model based estimates of the extremal coefficient function, $\ell(1,1)$. Instead of presenting them as a function of the actual distance between stations, we exploit the simple expression $\ell(1,1) = 2\Phi\left(a_{uv}/2\right)$ for the extremal coefficient function of the Brown-Resnick process, see Section~\ref{sec:theory}.

In the right panel of Figure \ref{Fig:AllTogetherNow}, the following are depicted:
\begin{compactitem}
 \item the 231 nonparametric estimates of the extremal coefficient function $\ell(1,1)$, based on all pairs of stations (circles),
 \item 6 per-bin averages of the nonparametric estimates of $\ell(1,1)$ (solid line), and
 \item the extremal coefficient function values computed from the model (dashed line),
\end{compactitem}
against the estimated distances
  \[\widehat{a}_{uv} = \sqrt{2 \widehat\gamma (\mathbf{s}_u - \mathbf{s}_v)} = \sqrt{2} \left(\frac{\|\mathbf{s}_u - \mathbf{s}_v\|}		 {\widehat{\rho}}\right)^{\widehat{\alpha}/2}.
  \]

The vertical line in the plot represents the 50 km threshold. It is more in line with our M-estimator, which uses integration over $[0,1]^2$, to focus on the center  $(1/2, 1/2)$ instead of the vertex $(1,1)$ of the unit square. Hence, we use the homogeneity of $\ell$ to replace $\ell(1,1)$ with $2\ell(1/2, 1/2)$ and then estimate the latter with $2\widehat \ell_{n,k}(1/2, 1/2)$,  see \eqref{eq:npstdf}. The nonparametric estimates of $\ell(1,1)$ in the figure are obtained in this way. We see that the estimated $\ell(1,1)$ of the Brown--Resnick process is quite close to the average $2\widehat \ell_{n,k}(1/2, 1/2)$ per-bin, supporting the adequacy of the model.

\appendix

\section{Proofs}
\label{sec:proofs}

The notations are as in Section~\ref{sec:Mestimator}. Let $\widehat{\Theta}_n$ denote the (possibly empty) set of minimizers of the function
\begin{equation*}
f_{n,k,\Omegan} (\theta) =  L_{n,k} (\theta)^T \, \Omegan \, L_{n,k} (\theta) \eqqcolon \lVert L_{n,k} (\theta) \rVert^2_{\Omegan}.
\end{equation*}
Write $\delta_0$ for the Dirac measure concentrated at zero. Recall that to each $m \in \{1, \ldots, q\}$ there corresponds a pair of indices $\pi(m) = (u, v)$ with $1 \le u < v \le d$. Let $\mu = (\mu_1,\ldots,\mu_q)^T$ denote a column vector of measures on $\RR^\dm$ whose $m$-th element is defined as
\begin{equation*}
  \mu_m (\diff \vc{x})
  = \mu_m (\diff x_1 \times \ldots \times \diff x_d)
  = \mu_{m1} (\diff x_1) \times \ldots \times \mu_{md} (\diff x_d)
  \coloneqq \diff x_u \, \diff x_v \prod_{j \neq u,v} \delta_0 (\diff x_j),
\end{equation*}
so that $\mu_{mj}$ is the Lebesgue measure if $j=u$ or $j=v$, and $\mu_{mj}$ is the Dirac measure at zero for $j \neq u,v$. Using the measures $\mu_m$ allows us to write
\begin{equation*}
L_{n,k} (\theta) = \left( \int_{[0,1]^\dm} \left\{ \widehat{\ell}_{n, k} (\vc{x}) - \ell (\vc{x} ; \theta) \right\} \, \mu_m(\diff \vc{x}) \right)_{m=1}^q = \int \widehat{\ell}_{n,k} \, \mu - \psi(\theta).
\end{equation*}

\begin{lem}\label{firstlem}
If $0 < \lambda_{n,1} \le \ldots \le \lambda_{n,q}$ and $0 < \lambda_{1} \le \ldots \le \lambda_{q}$ denote the ordered eigenvalues of the symmetric matrices $\Omegan$ and $\Omega \in \mathbb{R}^{q \times q}$, respectively, then, as $n \to \infty$,
\begin{equation*}
  \Omegan \pto \Omega
  \qquad \text{implies} \qquad
  (\lambda_{n,1},\ldots,\lambda_{n,q})
  \pto
  (\lambda_{1},\ldots,\lambda_{q}).
\end{equation*}
\end{lem}

\begin{proof}[Proof of Lemma \ref{firstlem}]
The convergence $\Omegan \pto \Omega$ elementwise implies $\lVert \Omegan - \Omega \rVert \pto 0$ for any matrix norm $\lVert \, \cdot \, \rVert$ on $\mathbb{R}^{q \times q}$.
If we take the spectral norm $\lVert \Omega \rVert$ (i.e., $\lVert \Omega \rVert^2$ is the largest eigenvalue of $\Omega^T \Omega$), then Weyl's perturbation theorem \citep[page 145]{jiang2010} states that
\begin{equation*}
\max_{i=1,\ldots,q} \left| \lambda_{n,i} - \lambda_{i} \right| \leq \lVert \Omegan - \Omega \rVert,
\end{equation*}
so that the desired result follows immediately.
\end{proof}

By the diagonalization of $\Omegan$ in terms of its eigenvectors and eigenvalues, the norm $\lVert \, \cdot \, \rVert_{\Omegan}$ is equivalent to the Euclidian norm $\lVert \, \cdot \, \rVert$ in the sense that
\begin{equation*}
\lambda_{n,1}  \lVert L_{n,k} (\theta) \rVert^2 \leq  \lVert L_{n,k} (\theta) \rVert^2_{\Omegan} \leq \lambda_{n,q}  \lVert L_{n,k} (\theta) \rVert^2.
\end{equation*}

\begin{proof}[Proof of Theorem \ref{resultmain1}]
Let $\eps_0 > 0$ be such that the closed ball $B_{\eps_0} (\theta_0) = \{\theta : \left\Vert \theta - \theta_0 \right\Vert \leq \eps_0 \}$ is a subset of $\Theta$; such an $\eps_0$ exists since $\theta_0$ is an interior point of $\Theta$. Fix $\eps > 0$ such that $0 < \eps \leq \eps_0$.
We show first that
\begin{equation}
\label{eq:consistency}
  \Pr[
    \widehat{\Theta}_n \neq \varnothing
    \text{ and }
    \widehat{\Theta}_n \subset B_{\eps} (\theta_0)
  ]
  \rightarrow 1, \qquad n \to \infty.
\end{equation}
Because $\psi$ is a homeomorphism, there exists $\delta >0$ such that for $\theta \in \Theta$, $\left\Vert \psi(\theta) - \psi(\theta_0) \right\Vert \leq \delta$ implies
$ \left\Vert \theta - \theta_0 \right\Vert \leq \eps$.
Equivalently, for every $\theta \in \Theta$ such that $\left\Vert \theta - \theta_0 \right\Vert > \eps$ we have $\left\Vert \psi(\theta) - \psi(\theta_0) \right\Vert > \delta$. Define the event
\begin{equation*}
A_n = \left\{ \left\Vert  \psi(\theta_0) - \int \widehat{\ell}_{n,k} \, \mu \right\Vert \leq \frac{\delta \sqrt{\lambda_{n,1}}}{ 2 + \sqrt{\lambda_{n,q}}}\right\}.
\end{equation*}
If $\theta \in \Theta$ is such that $\left\Vert \theta - \theta_0 \right\Vert > \eps$, then on the event $A_n$, we have
\begin{align*}
\left\Vert L_{n,k} (\theta) \right\Vert_{\Omegan} & = \left\Vert \psi(\theta_0) - \psi(\theta) - \left(  \psi (\theta_0) - \int \widehat{\ell}_{n,k} \, \mu \right)   \right\Vert_{\Omegan} \\
& \geq \left\Vert \psi(\theta_0) - \psi (\theta) \right\Vert_{\Omegan}  - \left\Vert \psi(\theta_0) - \int  \widehat{\ell}_{n,k} \, \mu \right\Vert_{\Omegan} \\
& \geq \sqrt{\lambda_{n,1}} \left\Vert \psi(\theta_0) - \psi (\theta) \right\Vert  -  \sqrt{\lambda_{n,q}} \left\Vert  \psi(\theta_0) - \int \widehat{\ell}_{n,k} \, \mu \right\Vert  \\
& >   \delta \sqrt{\lambda_{n,1}} - \delta \frac{\sqrt{\lambda_{n,1} \lambda_{n,q}}}{2 + \sqrt{\lambda_{n,q}}} = \frac{2 \delta \sqrt{\lambda_{n,1}}}{2 + \sqrt{\lambda_{n,q}}}
.
\end{align*}
It follows that on $A_n$,
\begin{equation*}
  \inf_{\theta: \left\Vert \theta - \theta_0 \right\Vert > \eps}
  \left\Vert L_{n,k} (\theta) \right\Vert_{\Omegan}
  \ge \frac{2 \delta \sqrt{\lambda_{n,1}}}{2 + \sqrt{\lambda_{n,q}}}
  > \left\Vert \psi(\theta_0) - \int \widehat{\ell}_{n,k} \, \mu \right\Vert
  \geq \inf_{\theta: \left\Vert \theta - \theta_0 \right\Vert \leq \eps}
  \left\Vert \psi(\theta) - \int \widehat{\ell}_{n,k} \, \mu \right\Vert.
\end{equation*}
The infimum on the right-hand side is actually a minimum since $\psi$ is continuous and $B_{\eps} (\theta_0) $ is compact.
Hence on $A_n$ the set $\widehat{\Theta}_n$ is non-empty and $\widehat{\Theta}_n \subset B_{\eps} (\theta_0)$.

To show \eqref{eq:consistency}, it remains to be shown that $\Pr[A_n] \to 1$ as $n \to \infty$.
Uniform consistency of $\widehat{\ell}_{n,k}$ for $\dm = 2$ was shown in \citet{huang1992}; see also \citet[page 237]{dehaanferreira2006}. The proof for $\dm > 2$ is a straightforward extension. By the continuous mapping theorem, it follows that $\int \widehat{\ell}_{n,k} \, \mu$ is consistent for $\int \ell \, \mu = \psi(\theta_0)$. By Lemma~\ref{firstlem}, $\lambda_{n,m}$ is consistent for $\lambda_{m}$ for all $m \in \{1, \ldots, q\}$. This yields $\Pr[A_n] \rightarrow 1$ and thus \eqref{eq:consistency}.

Next we wish to prove that, with probability tending to one, $\widehat{\Theta}_n$ has exactly one element, i.e., the function $f_{n,k,\Omegan}$ has a unique minimizer. To do so, we will show that there exists $\eps_1 \in (0, \eps_0]$ such that, with probability tending to one, the Hessian of $f_{n,k,\Omegan}$ is positive definite on $B_{\eps_1}(\theta_0)$ and thus $f_{n,k,\Omegan}$ is strictly convex on $B_{\eps_1}(\theta_0)$. In combination with \eqref{eq:consistency} for  $\varepsilon \in (0,\varepsilon_1]$, this will yield the desired conclusion.

For $\theta \in \Theta$, define the symmetric $p \times p$ matrix $\Hessian(\theta; \theta_0)$ by
\[
  \bigl( \Hessian(\theta; \theta_0) \bigr)_{i,j}
  \coloneqq
  2  \left( \frac{\partial \psi (\theta)}{\partial \theta_j} \right)^T \Omega \left( \frac{\partial \psi (\theta)}{\partial \theta_i}  \right) - 2 \left( \frac{\partial^2 \psi (\theta)}{\partial \theta_j \partial \theta_i}  \right) \, \Omega \,
  \bigl( \psi(\theta_0) - \psi (\theta) \bigr)
\]
for $i, j \in \{1, \ldots, p\}$. The map $\theta \mapsto \Hessian(\theta; \theta_0)$ is continuous and
\[
  \Hessian(\theta_0)
  \coloneqq
  \Hessian(\theta_0; \theta_0)
  = 2 \, \dot{\psi}(\theta_0)^T \, \Omega \, \dot{\psi}(\theta_0),
\]
is a positive definite matrix. Let $\lVert \, \cdot \, \rVert$ denote a matrix norm. By an argument similar to that in the proof of Lemma~\ref{firstlem}, there exists $\eta > 0$ such that every symmetric matrix $A \in \RR^{p \times p}$ with $\norm{A - \Hessian(\theta_0)} \le \eta$ has positive eigenvalues and is therefore positive definite. Let $\eps_1 \in (0, \eps_0]$ be sufficiently small such that the second-order partial derivatives of $\psi$ are continuous on $B_{\eps_1}(\theta_0)$ and such that $\norm{ \Hessian(\theta; \theta_0) - \Hessian(\theta_0)} \le \eta/2$ for all $\theta \in B_{\eps_1}(\theta_0)$.

Let $\mathcal{H}_{n,k,\Omegan} (\theta) \in \mathbb{R}^{p \times p}$ denote the Hessian matrix of $f_{n,k,\Omegan}$. Its $(i,j)$-th element is
\begin{align*}
\bigl( \mathcal{H}_{n,k,\Omegan} (\theta) \bigr)_{ij}
& = \frac{\partial^2}{\partial \theta_j \partial \theta_i}
\left[ L_{n,k} (\theta)^T \, \Omegan \, L_{n,k} (\theta) \right]
 = \frac{\partial}{\partial \theta_j}
\left[ - 2 L_{n,k} (\theta)^T \, \Omegan \frac{\partial \psi (\theta)}{\partial \theta_i} \right] \\
& = 2  \left( \frac{\partial \psi (\theta)}{\partial \theta_j} \right)^T \Omegan \left( \frac{\partial \psi (\theta)}{\partial \theta_i}  \right) - 2 \left( \frac{\partial^2 \psi (\theta)}{\partial \theta_j \partial \theta_i}  \right) \Omegan \,  L_{n,k} (\theta).
\end{align*}
Since $L_{n,k}(\theta) = \int \widehat{\ell}_{n,k} \, \mu - \psi(\theta)$ and since $\int \widehat{\ell}_{n,k} \, \mu$ converges in probability to $\psi(\theta_0)$, we obtain
\begin{equation}
\label{eq:unicon}
  \sup_{\theta \in B_{\eps_1}(\theta_0)}
  \norm{ \Hessian_{n,k,\Omegan}(\theta) - \Hessian(\theta; \theta_0) }
  \pto 0,
  \qquad n \to \infty.
\end{equation}
By the triangle inequality, it follows that
\[
  \Pr \biggl[ \sup_{\theta \in B_{\eps_1}(\theta_0)}
  \norm{ \Hessian_{n,k,\Omegan}(\theta) - \Hessian(\theta_0) } \le \eta \biggr]
  \to 1,
  \qquad n \to \infty.
\]
In view of our choice for $\eta$, this implies that, with probability tending to one, $\Hessian_{n,k}(\theta)$ is positive definite for all $\theta \in B_{\eps_1}(\theta_0)$, as required.
\end{proof}

\begin{proof}[Proof of Theorem \ref{resultmain2}]
First note that, as $n \rightarrow \infty$,
\begin{equation*}
\sqrt{k} \, L_{n,k} (\theta_0) \dto \widetilde{B}, \qquad \text{ where } \widetilde{B} \sim \mathcal{N}_q (\vc{0},\Gamma(\theta_0)).
\end{equation*}
This follows directly from \citet[Proposition 7.3]{einmahl2012} by replacing $g(\vc{x}) \,\text{d}\vc{x}$ with $\mu (\text{d} \vc{x})$. Also, from $(C2)$ and Slutsky's lemma, we have
\begin{align*}
  \sqrt{k} \, \nabla f_{n,k,\Omegan} (\theta_0)
  &= - 2 \sqrt{k} \, L_{n,k} (\theta_0)^T \, \Omegan \, \dot{\psi} (\theta_0) \\
  &\dto - 2 \, \widetilde{B}^T \, \Omega \, \dot{\psi}(\theta_0)
  \sim \mathcal{N}_p
  \bigl(
    \vc{0}, \;
    4 \, \dot{\psi}(\theta_0)^T \, \Omega \, \Gamma(\theta_0) \, \Omega \, \dot{\psi} (\theta)
  \bigr).
\end{align*}
Since $\widehat{\theta}_n$ is a minimizer of $\widehat{f}_{k,n}$ we have $\nabla f_{n,k,\Omegan} (\widehat{\theta}_n) = 0$. Applying the mean value theorem to the function $t \mapsto \nabla f_{n,k,\Omegan} (\theta_0 + t (\widehat{\theta}_n - \theta_0))$ at $t = 0$ and $t = 1$ yields
\begin{equation*}
0 = \nabla f_{n,k,\Omegan}  (\widehat{\theta}_n) = \nabla f_{n,k,\Omegan}  (\theta_0) + \Hessian_{n,k,\Omegan} (\widetilde{\theta}_n) \, (\widehat{\theta}_n - \theta_0)
\end{equation*}
where $\widetilde{\theta}_n$ is a random vector on the segment connecting $\theta_0$ and $\widehat{\theta}_n$. As $\widehat{\theta}_n \pto \theta_0$, we have $\widetilde{\theta}_n \pto \theta_0$ too. By \eqref{eq:unicon} and continuity of $\theta \mapsto \Hessian(\theta; \theta_0)$, it then follows that $\mathcal{H}_{n,k,\Omegan} (\widetilde{\theta}_n) \pto \mathcal{H} (\theta_0)$. Putting these facts together, we conclude that
\begin{equation*}
\sqrt{k} (\widehat{\theta}_n - \theta_0)
= - \bigl( \mathcal{H}_{n,k,\Omegan} (\widetilde{\theta}_n) \bigr)^{-1} \, \sqrt{k} \, \nabla f_{n,k,\Omegan}  (\theta_0) \dto \mathcal{N}_p \bigl( 0, M(\theta_0) \bigr),
\end{equation*}
as required.
\end{proof}

\begin{proof}[Proof of Corollary \ref{cor1}]
Assumption (C6) implies that the map $\theta \mapsto \Gamma (\theta)$ is continuous at $\theta_0$ \citep[Lemma 7.2]{einmahl2008}. Further, $\Gamma (\widehat{\theta}^{(0)}_n)^{-1}$ converges in probability to $\Gamma (\theta_0)^{-1}$,
because of the continuous mapping theorem and the fact that $\widehat{\theta}^{(0)}_n$ is a consistent estimator of $\theta_0$. Finally, the choice $\Omega_{\textnormal{opt}} = \Gamma (\theta)^{-1}$ in \eqref{eq:asym} leads to the minimal matrix $M_{\textnormal{opt}}(\theta)$ in \eqref{eq:Mopt}; see for example \citet[page 339]{abadir2005}.
\end{proof}

\section*{Acknowledgments}
The authors are grateful to the editor, the associate editor, and two referees for their helpful suggestions. This research is supported by contract ``Projet d'Actions de Recherche Concert\'ees" No.~12/17-045 of the ``Communaut\'{e} fran\c{c}aise de Belgique'' and by IAP research network grant nr.~P7/06 of the Belgian government (Belgian Science Policy). The second author gratefully acknowledges funding from the Belgian Fund for Scientific Research (F.R.S.-FNRS).

\renewcommand\refname{REFERENCES}
\bibliographystyle{chicago}
\bibliography{library}

\end{document}